\pdfoutput=1
\newif\ifbibtex
\bibtextrue  
\newif\ifsubmission
 \submissionfalse
\newif\iflongversion
\longversionfalse
\newif\iferlang
\erlangtrue
\newif\ifnotbal
\notbaltrue
\newif\ifwithcomments
\withcommentstrue
\withcommentsfalse

\ifsubmission
\else

\documentclass[acmsmall,screen, nonacm
\ifbibtex\else ,natbib=false\fi 
 ]{acmart}
\fi

\usepackage[T1]{fontenc}
\usepackage[utf8]{inputenc}
\usepackage{mathtools}
\usepackage[english]{babel}
\addto\extrasenglish{
  
}

\usepackage{setup}

\setcopyright{cc}
\setcctype{by}
\acmDOI{10.1145/3720497}
\acmYear{2025}
\acmJournal{PACMPL}
\acmVolume{9}
\acmNumber{OOPSLA1}
\acmArticle{132}
\acmMonth{4}
\received{2024-10-15}
\received[accepted]{2025-02-18}

\ifbibtex
\else
\RequirePackage[
  datamodel=acmdatamodel,
  style=acmauthoryear, 
  uniquelist=false,   
  ]{biblatex}
\DeclareLabeldate[software]{%
  \field{date}
  \field{year}
  \field{eventdate}
  \field{origdate}
  \field{urldate}
}
\fi

\usepackage{booktabs}   
\usepackage{subcaption} 

\usepackage{etoolbox}

\begin{document}

\title[
]{Polymorphic Records for Dynamic Languages}                     


\author{Giuseppe Castagna}
\orcid{0000-0003-0951-7535}
\affiliation{%
  \institution{CNRS - Université Paris Cité}
  \city{Paris}
  \country{France}
}

\author{Loïc Peyrot}
\orcid{0000-0002-1398-7460}
\affiliation{%
  \institution{IMDEA Software Institute}
  \city{Madrid}
  \country{Spain}
}
\begin{abstract}
We study \emph{row polymorphism} for records types in systems with set-theoretic types, specifically, union, intersection, and negation types. We consider record types that embed row variables and define a subtyping relation by interpreting record types into sets of record values, and row variables into sets of \emph{rows}, that is, ``chunks'' of record values where some record keys are left out: subtyping is then containment of the interpretations.
We define a λ-calculus equipped with operations for field extension, selection, and deletion, its operational semantics, and a type system that we prove to be sound. We provide algorithms for deciding the typing and subtyping relations, and to decide whether two types can be instantiated to make one subtype of the other.

This research is motivated by the current trend of defining static type systems for dynamic languages and, in our case, by an ongoing effort of endowing the Elixir programming language with a gradual type system.

\end{abstract}

\ccsdesc[100]{Theory of computation~Type theory}
\ccsdesc[100]{Software and its engineering~Polymorphism}


\keywords{subtyping, union types, intersection types, record types, row polymorphism}

\maketitle

\section{Introduction}
\label{sec:intro}
The goal of this work is to define and study \emph{row polymorphism} for a type system with
set-theoretic types, in particular, union, intersection, and negation types.
Row polymorphism was originally introduced by~\citet{wand87,wand91,remy89,remy94} and
further studied and developed in several works
(e.g.,~\cite{gaster96,leijen05,morris19,tang23}). However, as we argue
in this work, the theories in the current literature---which use
algebraic approaches---cannot
be reused for systems with set-theoretic types
and, hence, an original theory must be developed.
More broadly, the primary conceptual contribution of our work is
demonstrating that row polymorphism is not limited to syntactic or
algebraic approaches. Instead, it can be effectively integrated into
the semantic subtyping framework using the same three conceptual steps
employed to incorporate parametric polymorphism into the
framework~\cite{castagna14,castagna15}.

The interest of developing this new theory is grounded in
practice. There is a growing effort to develop type systems for
dynamic languages: TypeScript~\cite{typescript} and Flow~\cite{facebookflow} are two major examples of such
an effort and, as most of the existing attempts to add type systems to
dynamic languages, they include union and
intersection types to accommodate the flexible programming patterns
prevalent in such languages.
In particular, some widely used dynamic programming languages---such as
Luau~\cite{luau,luausemsub}, and \citet{elixir}---have embraced set-theoretic
types as defined in the \emph{semantic subtyping} framework by~\citet{frisch08}.
There, $(i)$ types are interpreted as sets of
values, $(ii)$ unions, intersections, and negation types are
interpreted as the corresponding set-theoretic interpretations, and
$(iii)$ the subtyping between two types is defined as set-containment
of their interpretations. Since semantic subtyping is the
framework that we use here, our work can be considered
as a study on adding row polymorphism to Luau and
Elixir. In particular, the rest of this section uses Elixir's syntax to showcase
practical motivations.

For space reasons, several examples, definitions, and all proofs are omitted. They can be found in the appendixes available in  the full version of this work and in the Auxiliary Material section of the ACM Digital Library entry for this work. 
\subsection{A Motivating Example}\label{sec:example}
To show the interest of having row polymorphism in dynamic languages, let us
consider the \code{logger} module of Elixir's standard library, which
exports the following function (see \cite{elixirlogger}):
\begin{minted}{elixir}
def add_elixir_domain(x) do
   case x do
     %{domain: y} when is_list(y) -> %{x | domain: [:elixir | y]} !\phantomsection\label{three}!
     _ -> Map.put(x, :domain, [:elixir]) !\phantomsection\label{four}!
   end
end 
\end{minted}
This code snippet defines the function \elix{add_elixir_domain}, whose argument, bound to \elix{x}, is matched
in a case
expression. The argument must be a record value (in Elixir records are delimited by curly brackets prefixed by the
``\code{\%}'' symbol).\footnote{For consistency, we follow
Hoare's original terminology~\cite{Hoa66a,Hoa66b}, and always use the term ``record'' to
denote finite key-value maps. In contrast, Elixir uses the word
``maps'' for finite key-value maps and reserves ``record'' for a different concept.} If it is a record 
with (at least) a field for the key \elix{:domain} whose content is a list (line~\ref{three}), then the function
returns a copy of the argument in which the field with key \elix{:domain} is
updated by consing the atom\footnote{Atoms are user-defined
constants prefixed by a colon and of type \felix{atom()}. Field keys
are atoms, too, in which the starting colon can be omitted if the
field is not optional (see later on): \felix{!\%!{domain: 42}}
has a field with key \felix{:domain} and value \felix{42}.}
\elix{:elixir} to the list
in the argument. Otherwise, (line~\ref{four}), the function adds (if absent) or replaces (if
present but does not contain a list) in the record \elix{x} a
field \elix{:domain} whose content is the singleton list whose only
element is the atom \elix{:elixir} (this is performed by the
function \elix{Map.put} of Elixir's standard library).

Elixir's type system~\cite{elixirtypes,castagna24a}
features record types without row polymorphism.
Hence, the function \elix{add_elixir_domain} can be given this type
(the
symbol  \textcolor{cadmiumgreen}{\tt \$} is used to introduce type
signatures \cite{castagna24a}):\footnote{Record types were added
  in Elixir's 1.17 release,
while ``\textcolor{cadmiumgreen}{\tt \$}'' and type variables are planned
for latter releases: \cite{elixirtypes}. Elixir's syntax for
row variables (cf.\ code in line~\ref{eight} and on) is not fixed yet:
the notation \felix{!\%!{...f, }} is also being considered.}
\begin{minted}[framesep=6pt%,framerule=.1pt,rulecolor=\color{red}
]{elixir}
!\textcolor{cadmiumgreen}\$! %{...} -> %{..., domain: list(term())}
\end{minted}
In Elixir, record types starting by \elix{...} are ``open'', that is, the values of these record types may define other
fields besides those specified in the type. Therefore, the type above
states that \elix{add_elixir_domain} is a function that accepts any
record value (since \elix{!\%!{...}} is the top record type) and
returns a record with at least a \elix{:domain} field that contains a list of
values (in Elixir \elix{term()} is the top type,  which types all
values, so \elix{list(term())} is the type of all lists; following the Erlang convention, type names are post-fixed
by \code{()}, e.g., \elix{integer()}, \elix{term()}, ...). 

The problem with such a type is well known: the type does not specify
that the fields corresponding to the ellipsis in the result are the same as
those in the argument. Thus, any static knowledge of these fields is lost after the application. The practical consequence of this loss is that an expression such as
\elix{add_elixir_domain(!\%!{file: "foo.txt", line: 42}).line}, 
which tries to select the field \elix{:line} in the result of the
application, is rejected by the type checker despite being
correct. This hinders the applicability of the type system to
existing code and obliges the programmer to resort to the less precise
gradual typing features of Elixir's type system.

The solution to this problem is also well known, and resorts to
using \emph{row variables}, which are variables that range over ``rows'' of
fields of a record
type. For Elixir, this would correspond to extending the current type system so
that \elix{add_elixir_domain} could be typed as follows:%
\begin{minted}[framesep=6pt%,framerule=.1pt,rulecolor=\color{red}
]{elixir}
!\textcolor{cadmiumgreen}\$! %{f} -> %{f, domain: list(term())} when f: fields()       !\phantomsection\label{eight}!
\end{minted}
where \elix{f} is a \emph{row variable}---as stated by the post-fix
declaration \elix{f: fields()} (in Elixir, type variables are
post-fixedly quantified by a \elix{when} declaration)---which, intuitively, ranges over all fields not already specified by a record type \elix{f} occurs in (here, all fields but \elix{:domain}).
Now, when typing the previous expression
\elix{add_elixir_domain(!\%!{file: "foo.txt", line: 42}).line}, the row variable \elix{f}
can  be instantiated to include the two fields for the keys \elix{:file} and
\elix{:line} in the argument, enabling the system to deduce for the expression
the type \elix{integer()}.

The type system for Elixir defined in~\cite{castagna24a}  also features union, intersection,
and negation types---denoted by \elix{or}, \elix{and}, and \elix{not}---
as well as polymorphic type variables.
Combined with row polymorphism, they refine the previous type of
\elix{add_elixir_domain} as follows:\footnote{Such a refinement can also be done
without row variables, but the problem with forgotten fields is still the same.}
\begin{minted}[%framesep=6pt%,framerule=.1pt,rulecolor=\color{red}
]{elixir}
!\textcolor{cadmiumgreen}\$! (%{f, domain: list(a)} -> %{f, domain: list(atom() or a)}) and          !\phantomsection\label{nine}!
  (%{g, :domain => not(list(term())} -> %{g, domain: list(atom())})       !\phantomsection\label{ten}!
  when a: term(), f: fields(), g: fields()                                !\phantomsection\label{eleven}!
\end{minted}
The type above uses the whole palette of set-theoretic types. It also
uses both parametric and row polymorphism, since the type features a
type variable \elix{a} and two row variables \elix{f} and \elix{g}, as
stated by the \elix{when} declaration in line~\ref{eleven}. The
combination of these two features renders a more precise description of the behavior of \elix{add_elixir_domain}:
\begin{itemize}[left=5pt]
  \item The arrow type in line~\ref{nine} states that when
    \elix{add_elixir_domain} is applied to a record value formed by a
    field \elix{:domain} that contains a list of \elix{a} elements, and by
    some other fields captured by \elix{f}, then the function returns  a
    record that is of the same type as the argument except in the \elix{:domain}
    field that now contains a list of ``\elix{atom() or a}'' elements (union type).
  \item The arrow type in line~\ref{ten} specifies as input type a record
    type with an optional field (denoted by ``\elix{=>}'')\footnote{A
      record type such as \felix{!\%!{:bar => t()}} means that in the values
      of this type, a field for  the key \felix{:bar} may be absent, but if
    present it must contain a value of type \felix{t()}.  This is Erlang's
    Typespec syntax: in Elixir the keyword \felix{optional} is
    mandatory, as in \felix{!\%!{optional(:bar) => t()}} , but we
    omit it to streamline the presentation.} of
    type \elix{not(list(term()))}.
    Thus, the function type in line~\ref{ten} types functions that when applied
    to values in which the field \elix{:domain} is either absent or bound to a
    value that is \elix{not} a list (negation type), then it returns a record
    where the \elix{:domain} field contains a list of atoms and, thanks to the
    row variable \elix{g}, where the other fields of the argument are preserved.
  \item The two arrows above are connected by an intersection
    (the \elix{and} connective at the end of line~\ref{nine}) meaning that
    the function has both types and, thus, obeys both specifications.
\end{itemize}
If \elix{add_elixir_domain} is given the type in lines~\ref{nine}--\ref{eleven}
and if \elix{x} is defined as:
\begin{minted}[framesep=6pt]{elixir}
x = add_elixir_domain(%{domain: [41, 43], file: "foo.txt", line: 42})
\end{minted}
then the type deduced for the expression \elix{[x.line | x.domain]}  
is \elix{list(atom() or integer())}.

It is out of the scope of this work to explain why the precision
allowed by set-theoretic types is necessary to the typing of dynamic
languages. For such an explanation, we invite the reader to refer to~\cite{castagna24a} which treats the case of Elixir. We want nevertheless to stress
the importance of negation types for a precise typing of pattern
matching. This is shown in line~\ref{four} of the definition
of \elix{add_elixir_domain}, where the type deduced for the
variable \elix{x} occurring in that line is obtained by removing, by a
negation type, from the type of the input the type of all the values
captured by the preceding pattern occurring in line~\ref{three}.


\subsection{The Need for Row Polymorphism}\label{sec:need}
Before delving into the necessity of a new theory of row polymorphism for records, we might first question whether row polymorphism is needed at all. The system we intend to extend already incorporates first-order (\emph{a.k.a.} prenex) polymorphism alongside set-theoretic types. This combination suffices to encode a limited form of bounded polymorphism \emph{à la} \textsf{Fun} \cite{CardelliWegner85}, which was employed to type polymorphic records.
Additionally, \citet{castagna16} demonstrated that semantic subtyping with polymorphic types is enough to capture polymorphic variant types, which are usually typed by row polymorphism since, in these specifics cases, they are dual of polymorphic record types.

In terms of records, the combination of prenex polymorphism and set-theoretic
types can be used to type the following function without losing the static
information of the fields of the argument:
\begin{minted}{elixir}
def bump_counter(x), do: %{x | counter: x.counter+1}
\end{minted}
This function increments the \elix{:counter} field of its argument and
can be typed by bounded polymorphism by the type
$\forall(\alpha{\leq}\erecord{..., \texttt{counter}{\is}\texttt{integer()}}).\alpha\to\alpha$. The type variable $\alpha$ captures the whole
type of the argument, thus also its extra fields. Thanks to that
the type system deduces for
\elix{(bump_counter(!\%!{counter: 1, file:
"foo.txt"})).file} the type string.
The bounded polymorphic  type above is encoded by set-theoretic types as 
follows (see~\cite{castagna24}):\footnote{Intersections and unions have a precedence
higher than arrows and records, and negation has the
highest precedence of all.}
\begin{minted}[framesep=6pt%,framerule=.1pt,rulecolor=\color{red}
]{elixir}
!\textcolor{cadmiumgreen}\$! (%{..., counter: integer()}!\;!and!\;!a  ->  %{..., counter: integer()}!\;!and!\;!a) when a: term()
\end{minted}
for which Elixir provides the nifty shorthand \elix{(a -> a) when a:
!\%!{..., counter: integer()}} of bounded polymorphism.

The example above may suggest that intersecting record types with
type variables could play the same role as row
polymorphism. Unfortunately, the example above is one of the few cases
in which this works.\footnote{Even this is not
true: it is just an approximation we did for presentation purposes,
since it is unsound to deduce the type \felix{a} for the result
type. To see why, try to instantiate \felix{a}
with \felix{!\%!{counter: 42}}. } This technique may work to type
functions in which the input and the output have the same fields; but
even in that case it is easy to have a partial loss of
information. For instance, consider the following function that takes
as input a record with a field \elix{foo} and redefines its content (in Elixir,
functions are annotated by the \code{\color{cadmiumgreen}\$}-prefixed
type that precedes their definition~\cite{castagna24a}):
\begin{minted}{elixir}
!\textcolor{cadmiumgreen}{\$}! (a, b) -> a when b: term(), a: %{..., foo: b}
def redefine_foo(x,y), do: %{x | foo: y}
\end{minted}
Although we do not lose the static type information of any
field of the argument (thanks to the type variable \elix{a}), the type for the \elix{:foo} field may lose
precision. For example, the type deduced for \elix{redefine_foo(!\%!{foo: 42},true).foo}
is the union \elix{integer() or boolean()}---rather than just \elix{boolean()}---since
the type variable \elix{b} is unified with the \emph{union} of the type of
the second argument and the type of the field \elix{foo:} of the first
argument. Furthermore, this technique fails when we try to add
a new field to a record or delete an existing field from it. For
instance, consider again the function \elix{Map.put}, whose use in
line~\ref{four} can be abstracted as follows:
\begin{minted}{elixir}
!\textcolor{cadmiumgreen}{\$}! (%{...} and a -> %{..., domain: list(atom())} and a) when a: term() !\phantomsection\label{seventeen}!
def put_domain(x), do: Map.put(x, :domain, [:elixir])
\end{minted}
One might think that the use of the intersection with the type
variable \elix{a} captures all the fields of the argument, but the
type declaration is wrong---and, as such, rejected by the type system---because if we instantiate the variable \elix{a}
by a type such as \elix{!\%!{domain: integer()}}, then we
deduce for \elix{put_domain} the type \elix{!\%!{domain: integer()}
-> none()}, where \elix{none()} is the empty type (resulting by simplifying the
intersection in the codomain): this type states that the application
of \elix{put_domain} to an argument of type \elix{!\%!{domain:
integer()}} must diverge (since any returned value must be in the
empty type, which contains none), which is clearly wrong. A similar
problem happens when we apply the record operations of deleting a field or of
adding a field that is not already present.

The problem with the type in line~\ref{seventeen} is that for each
intersection in it not to be empty, the type variable \elix{a} must be
instantiated by a record type that contains all the fields of the input and of
the output, in particular the field for \elix{:domain}. If, as above, \elix{:domain} has different types in the
input and in the output, then one of the two intersections results
empty (making the whole record type containing it to be empty) which, as explained above, is unsound.
For the typing to work, we need the type
variable to instantiate all the fields of the input \emph{except the
field for \elix{:domain}}. This is exactly the role of row variables,
which instantiate ``all the other fields'' of the record type. The
function \elix{put_domain} can be given any of the two following types (one for each line):
\begin{minted}{elixir}
!\textcolor{cadmiumgreen}{\$}! %{f}  -> %{f, domain: list(atom())} when f: fields() !\phantomsection\label{nineteen}!
!\textcolor{cadmiumgreen}{\$}! %{f, :domain => term()} -> %{f, domain: list(atom())}  when f: fields() !\phantomsection\label{twenty}!
\end{minted}
\new{The two types above are essentially two different syntaxes for the same type: the type in line~\ref{nineteen} is syntactic sugar for the one in line~\ref{twenty}, which, despite being more verbose, explicitly states in its input part, that the field \elix{:domain} is either undefined or contains a value of any type. The type in line~\ref{twenty}, thus, explicitly shows that the row variable \elix{f} captures all fields except \elix{:domain}.}

Likewise, we need row variables for typing a
function that deletes, say, the \elix{:domain} field since, once
again, the type of the deleted field will be different in the input
and in the output:
\begin{minted}{elixir}
!\textcolor{cadmiumgreen}{\$}! %{f, :domain => term()} -> %{f, :domain => none()}  when f: fields() !\phantomsection\label{twentyone}!
def del_domain(x), do: Map.delete(x, :domain)
\end{minted}
The codomain of the type in line~\ref{twentyone} states that the field
for  \elix{:domain} must be absent (i.e., if  present, it
must contain a value of the empty type \elix{none()}, of which there
is none). As before, the type above can be more conveniently written
as: \elix{!\%!{f} -> !\%!{f, :domain => none()}  when f: fields()}, and again, this works because \elix{f} can only be instantiated to rows that do not contain a field for \elix{:domain}. 

Since we have established that we need row polymorphism, the next question is
why not use the existing theory? A first reason for that is that to integrate row
polymorphism in a semantic subtyping
setting, we need to define subtyping for polymorphic records, and this
requires to interpret the row-polymorphic record types as sets of values which,
to our knowledge, was never done before.
\new{A second reason is the handling of optional fields, which, to the best of our knowledge, is unique to our setting and not addressed by other existing approaches
(cf.\  the discussion on presence polymorphism in \cref{sec:related}).}
A third quite challenging reason arises from the limitations of conventional unification techniques applied in row polymorphism; for instance, this issue becomes apparent when considering the following example:
\label{ex:figure}%
\begin{minted}{elixir}
!\textcolor{blue}{type}! figure() = %{shape: "circle", perim: integer(), diam: float()} or   !\phantomsection\label{23}!
                %{shape: "polygon", perim: integer(), edges: integer()}    !\phantomsection\label{24}!
!\textcolor{cadmiumgreen}{\$}! {f, perim: integer()} -> {f, perim: float()} when f: fields()    !\phantomsection\label{25}!
def perim_to_float(x), do: %{x | perim: to_float(x.perim)}         !\phantomsection\label{26}!
\end{minted}
The first two lines define the type of ``figures'', which are records with an integer field \elix{:perim} and with either a \elix{:diam} or an \elix{:edges} field, according to the value of their \elix{:shape} field. In line~\ref{26} we define a function that transforms the integer field \elix{:perim} of the input into float. Its type is given in line~\ref{25}. Now, if we apply the function \elix{perim_to_float} to an argument of type \elix{figure()} we expect to deduce for it a type like \elix{figure()} but where the \elix{:perim} field is of type \elix{float()}, that is 
\begin{minted}{elixir}
!\textcolor{cadmiumgreen}{\$}! %{shape: "circle", perim: float(), diam: float()} or     !\phantomsection\label{27}!
  %{shape: "polygon", perim: float(), edges: integer()}    !\phantomsection\label{28}!
\end{minted}
but current theories of row polymorphism unify record types component-wise,
which in our case would yield the following type which is less precise than (i.e., is a supertype of) the type above:
\begin{minted}[framesep=6pt]{elixir}
!\textcolor{cadmiumgreen}{\$}! %{shape: "circle" or "polygon", perim: float(), :diam => float(), :edges => integer()} !\phantomsection\label{29}!
\end{minted}
where the \elix{:shape} field has now a union type and \elix{:diam} and \elix{:edges} have become optional. To deduce the type in lines~\ref{27} and \ref{28} the row variable \elix{f} must be expanded into the union of two rows, one for the shape circle and the other for the shape polygon.

Even if the problem above can be solved by particular typing techniques,
the presence of negation types in our theory
invalidate such techniques in general (see example in
\cref{sec:app-atomic})\loic{Si on a la place on peut le mettre ailleurs que dans
l'appendice}. Therefore,
we need to develop an original technique to replace unification, so that it
takes into account subtyping and enables substitutions to expand row
variables into Boolean combinations of rows.

In conclusion, to embed row polymorphism in a
type system featuring semantically defined union, intersection, and
negation types, we need new theoretical developments to cope with the semantic interpretation of types and the inference of substitutions for type variables in the presence of subtyping and set-theoretic types.

\subsection{Overview}

In \cref{sec:types} we define the syntax of types and the subtyping relation. We abandon Elixir's syntax for maps/records and introduce in~\cref{sec:typesyntax} a more
theoretically-oriented one: we denote records types as finite lists
of field-type specifications of the form $\ell\is\tau$, followed by
a \emph{tail} $\tl$ specifying an infinite row of fields as in
$\rec{\ell_1\is\tau_1,...,\ell_n\is\tau_n}{\tl}$. In this type,
each $\ell_i$ denotes a label (or key) of a field, and labels are pairwise distinct; $\tau_i$ is either a type
$t$ or the union $t\vee\bot$, meaning that the field is optional with
type $t$ (i.e.,  Elixir's \elix{:key => t()} field) and absent if $t$
is the empty type (i.e., \elix{:key => none()}); $\tl$ is either a row
variable $\rho$, or  ``$\orecsign$'' (meaning that the record type is
open), or  ``$\crecsign$'' (meaning that the record is closed: its values
contain all and only the fields specified in the type).

To define the subtyping relation on types, we give
in  \cref{sec:models} an interpretation
of types as sets of elements of a domain $\Domain$---whose elements,
intuitively, represent the values of the language---and then define
subtyping as containment of the
interpretations. Following~\citet{frisch04}, records values are
interpreted as quasi-constant functions, that is, functions that map
all labels into $\bot$ (meaning that the field for that label is
undefined) except for a finite set of labels that are mapped into
values. Therefore, (ground) record types are interpreted as
sets of quasi-constant functions. More subtle is the interpretation
for row variables which, as we saw, define the type of \emph{all the
labels except a few ones}; as a consequence our interpretation will map
them into \emph{partial} quasi-constant functions (as in \cite{remy94}),
requiring a careful handling of their domains.
In \cref{sec:subtyping} we define an algorithm to decide the
subtyping relation just defined. We do so by extending the subtyping
algorithm of the monomorphic record types of CDuce~\cite{benzaken03}---on which we base our theory--- to our polymorphic records.
The resulting (backtrack-free) algorithm has the same order of complexity as the
one for monomorphic record types, which is currently used in Elixir.
%
Finally, we give a formal definition for type substitutions in
\cref{sec:rowsub}, in particular for the case of row variables which
are expanded into Boolean combinations of rows of fields, and prove
that its application preserves subtyping.

In \cref{sec:language} we define a language with operations on records.
For those, several equivalent choices are
possible~\cite{cardelli91}. We build records
starting from the empty record value, noted $\erecord$, and adding  new
fields to it by the expression $\extrecord e \ell {e'}$ which extends the record
(resulting from the evaluation of) $e$ with the field $\ell\is e'$, provided
that a field for $\ell$ is not already present in $e$.
The other operations on records are field selection, noted
$e.\ell$, which returns the content of the field $\ell$ in $e$, and field
deletion, noted $\delrecord e \ell$ which removes from $e$ the field labeled
$\ell$, if any. We define an operational semantics and a declarative type
system, and we show that the latter is sound in the sense of~\citet{wright94}, by proving
that every well-typed expression either diverges or returns a value of the
expression's type (\cref{sec:syntax}). Next, we  define an algorithmic type system
and prove it to be sound and complete with respect to the declarative
one. The system is derived from the declarative one in a standard way:
subsumption is embedded in the elimination rules, intersection
introduction is essentially embedded in the typing of
$\lambda$-abstractions, and the rule for applications performs
instantiation and expansion by looking for a set of substitutions that
make the type of the argument be a subtype of the domain of the
function  (\cref{sec:typealgo}).

\Cref{sec:tallying} studies the tallying problem, which plays the same role as
the unification problem in type inference, but for a subtyping---rather than an
equality---relation on types. The algorithmic system in \cref{sec:typealgo} is
effective, provided that we produce an algorithm  to deduce the type
substitutions to apply to the types of the function and the argument when typing
an application. Following~\citet{castagna15} this can be done by solving the
tallying problem for our types, namely, the problem of deciding whether given two types, there
exists a type substitution that makes one type subtype of the other.
\citet{castagna15} prove that the problem is decidable for a system with
function and product type constructors and set-theoretic types, and give a sound
and complete algorithm. However, defining a tallying algorithm for types with row
variables is far more difficult. This is because substitutions
replace row variables by Boolean combinations of rows of fields. We tackle this
problem in \cref{sec:tallying}, where  we define a tallying algorithm for row polymorphic types. We prove that the algorithm is sound but not complete, and we conjecture completeness for the case in which row variables are substituted by a single row of fields.

We conclude by discussing related work (\cref{sec:related}) and further research
directions (\cref{sec:conclusion}).

\subsection{Contributions and Limitations}
The overall contribution of this work is threefold, since it provides $(i)$ a
theory for a first-order polymorphic type system with row polymorphism and set-theoretic
types, $(ii)$ the practical motivations for such a system, as well as
$(iii)$ the relevant algorithms to apply it in practice. In particular, all the
examples we presented in \cref{sec:example,sec:need} are typed by our system.

The technical contributions can be summarized as follows:
\begin{enumerate}[left=0pt]
\item We describe a first-order polymorphic type theory with union,
intersection and negation type connectives, and function and record
type constructors, where record types can be either closed, open, or specify
a row variable, and their fields can be declared optional
(\cref{sec:typesyntax}). We define a subtyping relation for these types by
providing an interpretation 
where types are interpreted as set of values and subtyping as set containment
(\cref{sec:models}).
\item We prove that the subtyping relation is decidable and provide a
backtrack-free algorithm to decide it (\cref{sec:subtyping}) with
the same order of complexity as the one implemented for Elixir~\cite{elixirtypes}.
\item We define type substitutions that map row variables into Boolean
combinations of rows, and prove that the application of type substitutions
preserves subtyping (\cref{sec:rowsub}).
\item We define a declarative type system
for a record calculus with record extension, selection, and deletion
and prove its soundness (\cref{sec:syntax}).
\item We define an algorithmic system that we prove  sound and complete with
  respect to the declarative one (\cref{sec:typealgo}). 
\item We define an algorithm for the \emph{tallying problem} for our system, that is,
the problem of deciding whether given two types there exists a type
substitution that make one subtype of the other; we prove soundness of the
algorithm (\cref{sec:tallying}).
\end{enumerate}
The system defined here presents some limitations. Some  are
expected and characteristic of the kind of systems we consider here: the typing
relation is not decidable (this is typical of systems with intersection
types) and the type system has no principal types (which is already
the case both for systems with polymorphic set-theoretic
types~\cite{castagna14,castagna15} and for expressive record type
systems~\cite{cardelli91}). Other limitations are instead new, in
particular that the tallying algorithm is sound but not
complete (an example is given in \cref{ex:incomplete}; a complete algorithm exists when record types are kept out
of the equation: see~\citet{castagna15}). We prove that one of the
reasons for incompleteness is that we interpret row variables into
Boolean combinations of rows rather than into single rows, and we
conjecture that completeness can be recovered in the latter case, but at
the expenses of the type system which can type fewer expressions (cf.\
\cref{ex:boolrows}).
Note however that all the examples given in the introduction
fall outside the incompleteness area: as \cref{ex:incomplete} shows, building an example for
incompleteness requires the application of higher-order functions
whose types map unions of record types into similar unions of record
types.

\ifsubmission
\else
From a practical point of view, the main limitation of this
system is that it does not feature first class record labels, that is, the
operations for field selection, extension, and deletion must specify
nominal labels which, thus, cannot be obtained as the result of a
computation. This important omission might hinder the
application of our theory to dynamic languages where such a feature is
widely used. This omission, however, is deliberate since we wanted to
focus on the problem of row polymorphism, and having
first-class labels is mostly orthogonal to it. We believe that
it will not be  hard
to extend our work on the lines of~\citet{castagna23a} to add first
class labels to our system, and we leave it for future work.
\fi
Finally, we \new{expect} that implementing our system in a language such as
Elixir will not affect the performance of the type checker: the complexity of
our subtyping algorithm is the same as the one currently implemented in Elixir,
while the complexity of the tallying algorithm and, thus, of type inference, is
the same as the one implemented in CDuce~\cite{castagna15} and planned for
Elixir~\cite{castagna24a}, \new{up to a low multiplicative factor (cf.\ Footnote~\ref{fn:multiplicative}).
Upcoming work on implementing the type system (in CDuce)
shall provide empirical evidence.}

\section{Types}
\label{sec:types}
We introduce the syntax of types (\cref{sec:typesyntax}) and their set-theoretic
interpretation from which we derive the subtyping relation
(\cref{sec:models}).
We define the algorithm to decide the subtyping relation (\cref{sec:subtyping})
and prove that subtyping is preserved by type and row substitutions (\cref{sec:rowsub}).

\subsection{Syntax of Types}\label{sec:typesyntax}

\begin{definition}[Types, rows and kinds]\label{def:types}
  Let $\Labels$ be a countable set of \emph{labels} ranged over by $\ell$.
  The set $\Types$ of \emph{types} (ranged over by $t$) contains all terms
  \emph{coinductively} generated by the corresponding grammar below, and that $(1)$ have a
  finite number of different sub-terms (\emph{regularity}) and $(2)$ in which every
  infinite branch contains an infinite number of occurrences of the record or
  arrow type constructors (\emph{contractivity}).
  The set $\Rows$ of \emph{rows} (ranged over by $\rw$) as well as the set of
  \emph{field-types} (ranged over by $\tau$) contain all terms
  \emph{inductively} generated by the corresponding grammars below.

  \[
    \begin{array}{lrclr}
      \textbf{Kinds} & \kappa &\Coloneqq&
      \ktype
      \alt \kfield
      \alt \krow L \\
      \textbf{Types} & t &\Coloneqq&
      \alpha \alt
      b \alt
      \arrow t t \alt
      \rec{\ell = \tau, \dots, \ell = \tau}{\tl} \alt
      t \lor t \alt
      \neg t \alt
      \Empty\\
      \textbf{Field-types} & \tau &\Coloneqq&
      \fvar \alt
      t \alt
      \Undef \alt
      \tau \lor \tau \alt
      \neg \tau \\
      \textbf{Tails} & \tl &\Coloneqq&
      \rho \alt
      \crecsign \alt
      \orecsign\\
      \textbf{Rows} & \rw &\Coloneqq&
      \row{\ell = \tau, \dots, \ell = \tau}{\tl}{L} \alt
      \rw \lor \rw \alt
      \neg \rw
    \end{array}
  \]
  where, $\alpha$, $\fvar$, $\rho$ range, respectively, on type variables, field-type variables, and row variables; $\bot$ is a distinguished symbol different from all types; and (from now on) $L \in \Pf(\Labels)$ is a finite set of labels.
\end{definition}
\noindent
Following a mathematical logic terminology, basic types, arrows, and records are
called \emph{type constructors} and yield \emph{type atoms}, while unions, intersections, and
negations are \emph{type connectives}.
Our system use kinds. Types are of kind $\ktype$, field-types of kind $\kfield$, and we have an
infinite set of kinds for rows, parametrized by a finite set $L$: a
row indexed by the set $L$ is of kind $\krow L$.
We use $\tterm$ to range over types, field-types, and rows
\new{(i.e., $\tterm$ stands for either $t$, $\tau$, or $\rw$)}, and
define $\tterm_1 \land \tterm_2 \eqdef \neg(\neg \tterm_1 \lor \neg \tterm_2)$
and $\tterm_1{\setminus}\tterm_2 \eqdef \tterm_1 \land \neg\tterm_2$.
For every kind, besides the full set of type connectives, there are a bottom
and a top element (forming a lattice w.r.t.\ subtyping).
In $\ktype$, the top type is noted $\Any \eqdef \neg\Empty$;  in $\kfield$,
the top field-type is $\Any \lor \Undef$ (notice that $\Any$ contains all types, but not $\Undef$ which is not a type); in each kind
$\krow L$, the top row element is $\orow{}L$ 
We use the generic notation $\Empty \eqdef \neg\orow{}L$ for the
bottom element in $\krow L$, with $L$ being, thus, implicitly given by the context.

Besides the aforementioned connectives, the types $t$ of \cref{def:types} are
made of variables (from a countable set $\typeVars$ and ranged over by
$\alpha$), a finite set $\Basics$ of basic types (e.g., \Int, \Bool; ranged
over by $b$), function types, and record types.
Coinduction accounts for recursive types and comes with the usual
restrictions of regularity ---necessary for the decidability of the
subtyping relation--- and
contractivity ---which rules out meaningless types such as an infinite tower of
negations, while providing a well-founded order for inductive proofs (see~\cite{castagna24} for details).
\new{Notice that since both field-types and rows are built inductively, they are always \emph{finite} unions or negations of, respectively, field-types or rows ---this aligns with the contractivity requirement of types---, which are formed of (possibly infinite) types.}

\paragraph{Records}
Our records are based on the theory for records defined by
\citet{frisch04} (see~\cite{castagna23a} for a description in English)
and first used in CDuce. Our work extends the (monomorphic) record theory of~\cite{frisch04}  with row and field-type variables.
In Frisch's theory, a record \emph{value} is a total function on $\Labels$ that maps
a finite set of labels into  values, and all the remaining labels to a distinguished symbol
$\Undef$ representing the undefined (such a function is dubbed
\emph{quasi-constant} by Frisch: cf.\ \cref{def:qcf}).

Record type atoms (ranged over by $\R$) are types of the form
$\rec{\ell_1 = \tau_1, \dots, \ell_n = \tau_n}{\tl}$, that is, an
unordered list describing the  mapping  of a finite set of pairwise
distinct labels into field-types, which is followed by a \emph{tail} $\tl$
that covers the
infinitely many remaining labels.
We often use the more compact notation
$\rec{(\ell = \tau_\ell)_{\ell \in L}}{\tl}$, where $L \in \Pf(\Labels)$.
For $\R=\rec{(\ell = \tau_\ell)_{\ell \in L}}{\tl}$, we define
$\fin(\R) \eqdef L$ and $\tail(\R) \eqdef \tl$.

\paragraph{Fields}
In~\cite{frisch04}, field-types $\tau$ are either:
(a) \emph{mandatory} when $\tau = t$ with type $t$;
or (b) \emph{optional} when $\tau = t \lor \Undef$ with type $t$,
and in particular always undefined if $\tau = \Undef$---i.e., $\tau = \Empty\lor\Undef$ (notice that
  undefined fields cannot be typed just by $\Empty$:
  records are modeled as (indexed) products, in which a single empty
  field makes the whole record empty; also notice that since $\Undef$ is distinct from all types, then $t \lor \Undef \neq t$ and $\Undef \wedge t = \emptyset$ for any type~$t$).
Our definition adds field-type variables (\emph{field variables} for short),
ranged over by $\fvar$ and drawn from a countable set $\fldVars$.
In what follows, they are used in two ways:
(1) to solve the tallying problem (see \cref{sec:tallying})
and (2) to implement \emph{presence polymorphism} without additional effort (see
the discussion in \cref{sec:related}).
The introduction of field variables forces us to loosen the form of
field-types by allowing arbitrary Boolean combinations. 

\paragraph{Rows}
The tail $\tl$ of a record type atom $\R = \rec{(\ell = \tau_\ell)_{\ell \in L}}\tl$
can be of three sorts.
If $\tl = \crecsign$, then $\R$ is \emph{closed}, and only includes
records that
assign $\Undef$ to every label outside $L$.
If $\tl = \orecsign$, then $\R$ is \emph{open}, and imposes no restriction on the
values of the fields outside $L$.
That is, those labels are given field-type $\Any \lor \Undef$.
  Giving $\Any$ would be unsound, as it would imply that the
  selection of any remaining field of a record of this type would always be
  well-typed, even for fields not containing a value (such fields always exist
  as records are \new{built inductively}).
New in our work with respect to \cite{frisch04} is that $\tl$ may also be a row variable (taken from a countable set $\rowVars$ ranged
over by $\rho$) in which case $\R$ is \emph{polymorphic}.

In $\R$, the row variable $\rho = \tl$ defines the fields for the cofinite set of
labels $\rdef(\rho)$ that we call the \emph{definition space} of $\rho$.
Now, since record types are \textbf{total} functions on $\Labels$, we cannot use
them to interpret row variables.
For this, we use the new syntactic category of \emph{rows} from \cref{def:types},
that denote \textbf{partial} functions on $\Labels$ defined on co-finite sets of labels.

Rows are of the form $\row{(\ell = \tau_{\ell})_{\ell \in L}}{\tl}{L'}$.
The relevant part is the set $L'\in\Pf(\Labels)$ at the index, which denotes the
finite set of labels on which the row is \emph{not} defined. In
other terms, the row above is a total function from
$\Labels\setminus L'$ into field-types: the $\tau_\ell$'s define the
fields for the labels in $L$ and the tail $\tl$ the fields for the
labels in $\Labels\setminus(L\cup L')$.

Of course, not every row or record type is well-formed, since we must ensure
that the various $L$, $L'$, and $\rdef(\rho)$ form a partition of $\Labels$.
They have to verify the following three properties (enforced statically by a kinding
system whose straightforward definition is given in \cref{app:types}, \cref{fig:kinding}):
\begin{enumerate}[leftmargin=15pt]
  \item In a type $\rec{(\ell = \tau_\ell)_{\ell \in L}}{\rho}$ we must have
    $\Labels\setminus L = \rdef(\rho)$;
  \item In a row $\row{(\ell = \tau_\ell)_{\ell \in L}}{\tl}{L'}$ we must have
    $L \cap L' = \emptyset$ and, if $\tl = \rho$, then
    $\Labels\setminus(L\cup L')=\rdef(\rho)$;
  \item Unions (and, thus, intersections) are only on rows defined on the
    same set of labels.
\end{enumerate}
Finally, for all $\ell \in \Labels$, we define $\R(\ell)$
to yield the field-type associated to $\ell$ by $\R$:

\(\R(\ell) \eqdef \left\{\begin{array}{ll}
    \tau_\ell &\text{if }\ell\in L \text{; otherwise:}\\
    \Undef   &\text{if } \tl=\crecsign\\
    \Any\lor\Undef &\text{if } \tl = \orecsign \text{ or } \tl \in \rowVars
\end{array}\right.\)

\noindent
This operator, as well as others, is trivially transferred from record types to rows.
Notice that in the case of open and closed records (and rows), the choice of $L$
for a record type is not canonical, as for instance
$\crec{\ell_1 = \Int \lor \Undef, \ell_2 = \Undef}$ and $\crec{\ell_1 = \Int
\lor \Undef}$ are semantically equivalent.
Since row variables have a constant definition space however, a record type like
$\rec{\ell = \Any \lor \Undef}\rho$ has a single (top-level) representation with
$L = \{\ell\} = \Labels \setminus \rdef(\rho)$.

\begin{example}[Examples of record types and rows]
  $t = \orec{\ell_1 = \Int, \ell_2 = \Bool}$ is an open record type, whose values
  are records with at least an integer in $\ell_1$ and a boolean in
  $\ell_2$---e.g., $\erecord{\ell_1 = 3, \ell_2 = \True, \ell_3 = \Keyw{‘a’}}$.
  Its closed counterpart, $\crec{\ell_1 = \Int, \ell_2 = \Bool}$, types records in which both and only $\ell_1$ and $\ell_2$ have some value. The type
  $t$ is equivalent to (i.e., it is both a subtype and a supertype of)
  both $\orec{\ell_1 = \Int, \ell_2 = \Bool, \ell_4 = \Any \lor \Undef}$
  and $\orec{\ell_1 = \Int} \wedge \orec{\ell_2 = \Bool}$.
  By De Morgan's  laws, $\neg t$ is equivalent to
  $\neg \orec{\ell_1 = \Int} \vee \neg\orec{\ell_2 = \Bool}$ and thus to $\orec{\ell_1 = \neg\Int \lor \Undef} \vee \orec{\ell_2 = \neg\Bool \lor \Undef}$.
  The type $\orec{\ell_1 = \Empty}$ is equivalent to $\Empty$, and so are both
  $\crec{\ell_2 = \Empty}$ and $\rec{\ell_1 = \Empty}\rho$.\hfill\qed
\end{example}

\subsection{Subtyping Relation}
\label{sec:models}
The subtyping relation characterizes the type system:
union and intersection types are the least upper bounds and
greatest lower bounds of this relation, and the typing relation
relies on subtyping to compare types.
Our goal is to extend the monomorphic type theory of records of~\citet{frisch04}
to a polymorphic one, obtained by adding row variables.
For this, we stick to the technique employed to extend the theory of semantic
subtyping on monomorphic types to type-polymorphic ones
\cite{frisch08,castagna11,gesbert15}.
This technique can be distilled into three conceptual steps:
\begin{enumerate}[leftmargin=15pt]
  \item \emph{Define the subtyping relation on monomorphic types.} This is
    done in \cite{frisch08} by the definition of a \emph{set-theoretic model}.
    This definition describes how to interpret types as subsets of some domain
    whose elements represent the values of the language. Given a specific
    model (i.e., a specific domain $\Domain$ and an interpretation
    function $\TypeInter{\,}$ from types in $\Types$ to sets in $\Pd(\Domain)$), this induces a subtyping
    relation defined as the containment of
    the interpretations  (i.e., $t\leq t'\iffshortdef\TypeInter t \subseteq \TypeInter{t'}$).
  \item  \emph{Extend the definition of the \emph{model} to polymorphic
    types}. This is done by making
    the interpretation of types parametric in the interpretation $\eta$
    of \emph{type variables} as sets of values. Subtyping can then be defined
    as containment for every type variable interpretation $\eta$ (i.e., $t\leq t'\!\iffshortdef\!\forall\eta.(\TypeInter
      t\eta \subseteq \TypeInter{t'}\eta)$), but only for models that satisfy a
      so-called “convexity” condition~\cite{castagna11} (see after
      \cref{def:interpretation}).
  \item \emph{Exibit a specific convex model for polymorphic types and deduce
    the subtyping relation}.  \citet{gesbert15} show that a
    convex model for polymorphic types can be obtained by taking a specific model for
    monomorphic types and indexing all its elements by finite sets of
    type variables. 
\end{enumerate}
Henceforth, we apply the same approach to record types:
\begin{enumerate}[leftmargin=17pt]
\item[(r1)] The subtyping relation for monomorphic record types is defined
by interpreting them as sets of \emph{record values}
(i.e., \textbf{total} functions from labels to values).
\item[(r2)] We extend this interpretation to polymorphic record types, by
making it parametric in the interpretation $\eta$ of \emph{row
variables} into sets of \emph{rows}  (i.e., \textbf{partial}
functions from labels to values).
\item[(r3)] We define a convex model of polymorphic record types, by taking
a specific model for monomorphic record types and indexing all its
elements by finite sets of \emph{row variables}.
\end{enumerate}
The sole distinction between the two approaches, thus, lies in the
interpretation of \emph{type variables} and of \emph{row
variables}. While \emph{type} variables are mapped into sets of
values, \emph{row} variables are mapped into sets of \emph{rows}. Rows
are not values themselves but rather “chunks” of values: both (record)
values and rows are quasi-constant functions; however, (record) values
are total functions, whereas rows are partial ones.
We thus achieve the same conceptual simplicity as the polymorphic
extension of semantic subtyping, though the technical development is, in
our case, far more involved.

For space reasons, we present here only the final result of this
process, that is, the specific convex model of step (r3), given by the
domain of~\cref{def:interpretation} and the interpretation
of~\cref{def:indint}. We then derive from
it the definition of the subtyping relation (\cref{def:subtyping}) and
a decision  procedure that we prove sound, complete, and terminating
(\cref{l:correct_subalg} and \cref{l:term_subalg}). The development of the first
two steps is necessary only to prove that the domain and the interpretation
given below satisfy the properties of model and of convexity. This
development and the corresponding proofs are given
in \cref{sec:app-models}%
. Although a detailed explanation of these properties
is outside the scope of this presentation, we want to stress that
without the model property the definition of  subtyping
would not be well-founded and without the convexity property the decision
procedure would not be sound (see~\cite{frisch08,castagna11} for details).

To interpret record values we follow~\citet{frisch04}, and represent each record
value by a quasi-constant function that either maps labels into values
(i.e., the elements of $\Domain$) or into $\Undef$.
Let us write  $\Domain_\Undef$ for $\Domain\cup\{\Undef\}$  where $\Undef$ is a
distinguished element not in $\Domain$.
Quasi-constant functions are total functions that map all but a finite set of elements of their
domain into the same default value.
Formally, we have the following definition.
\begin{restatable}[\cite{frisch04}]{definition}{quasiconstant}
  \label{def:qcf}
  A function $r:\Labels \to \Domain_\Undef$ is \emph{quasi-constant} if the set
  $\lbrace \ell \in \Labels \mid r(\ell)\not = \Undef\rbrace$ is finite.
  We use $\Labels \qcfun \Domain_\Undef$ to denote
  the set of quasi-constant functions from $\Labels$ to $\Domain_\Undef$ and
  $\qcf{\ell_1=\delta_1, \ldots,\ell_n=\delta_n,\,\wild=\Undef}$ to denote the
  quasi-constant function $r:\Labels\qcfun \Domain_\Undef$ defined
  by $r(\ell_i)=\delta_i$ for $i=1..n$ and $r(\ell)=\Undef$ for
  $\ell\in\Labels{\setminus}\lbrace\ell_1,\ldots,\ell_n\rbrace$.
\end{restatable}
We also have to interpret \emph{rows}, which are defined only on a cofinite \emph{subset} of $\Labels$.
In other terms, rows are \emph{partial} quasi-constant functions from
$\Labels$ to $\Domain_\Undef$, that we note $\Labels\pqcfun \Domain_\Undef$.
Since a total function is also a partial one, then we need just
the latter in our domain to interpret record types and rows. This
yields the following definition of domain.
\begin{definition}[Domain]\label{def:interpretation}
  The \emph{interpretation domain} $\Domain$ for types, is the set of finite terms $d$
  inductively produced by the following grammar, where $c$ ranges over the set
  $\Constants$ of constants, $\ell$ over the set $\Labels$ of labels, and $V$ over sets of variables contained
  in $\Vars=\typeVars\cup\fldVars\cup\rowVars$.
  The interpretation domain $\Domain_\Undef$ for fields (resp. $\Domainrow$ for
  rows) is the set of terms $\dundef$ (resp. $\drow$).
  \begin{align*}
    d & \Coloneqq  c^V \mid \{(d, \domega), \dots, (d, \domega)\}^V
    \mid \domrec{\drow}^V
      & \rdef(\drow) = \Labels \\[-1.3mm]
    \drow & \Coloneqq \domrow{\ell_1 = \dundef, \dots, \ell_n = \dundef}{L}^V
          & L \in \Pf(\Labels\setminus\{\ell_1,...,\ell_n\}) \\[-1.3mm]
    \domega & \Coloneqq d \mid \Omega\\[-1.3mm]
    \dundef & \Coloneqq d \mid \Undef^V
  \end{align*}
  We use $\dterm$ for an element that is either $d$, $\drow$ or $\dundef$.
  We define $\rdef(\domrow{(\ell = \dundef_\ell)_{\ell \in L_1}}{L_2}^V) =
  \Labels \setminus L_2$. We use $\Tag(\dterm)$ to denote the set of
  variables indexing $\dterm$, that is, $\Tag(c^V)=\Tag(\{(d, \domega), \dots, (d, \domega)\}^V)=\Tag(\domrec{\drow}^V)=\Tag(\domrow{\ell_1 = \dundef, \dots, \ell_n = \dundef}{L}^V)=\Tag(\Undef^V) =V$.
\end{definition}
The elements of the domain are constants $\Constants$ to interpret basic types,
sets of finite binary relations  $\Pf(\Domain \times \Domain_\Omega)$ to
interpret function types, and \emph{partial} quasi constant functions
$\Labels \pqcfun \Domain_\Undef$ to interpret rows (and record types
by the total ones). The fact
that functions are finite binary relations is a standard technique of semantic
subtyping, and corresponds to interpreting function spaces
into the infinite set of their finite approximations; that these binary relations can yield a distinguished
element $\Omega$ (which, intuitively, represents a type error) is also
a standard technique of semantic subtyping used to avoid $\Any\to\Any$ to be a supertype  of
all function types: since both aspects do not play any specific role in our
work we will not further
comment on them (see~\cite{frisch08} for a detailed explanation
or~\cite[Section 3.2]{castagna23a} for a shorter one).

If we look more closely at the definition of the row elements
in \cref{def:interpretation}, we see
that they are \emph{partial} quasi-constant functions in
$\Labels\pqcfun\Domain_\Undef$ with default value $\Undef$. More
precisely, the row element $\domrow{\ell_1 = \dundef_1, \dots, \ell_n = \dundef_n}{L}$
is the quasi-constant function $\qcf{\ell_1
= \dundef_1, \dots, \ell_n = \dundef_n,\wild=\Undef}$ in $(\Labels\setminus
L)\qcfun\Domain_\Undef$. When a row is total on $\Labels$, then it can
be wrapped in a $\domrec{}$ constructor yielding (the interpretation
of) a record value (the inverse operation is given in the \cref{def:row}
below).

All these elements are indexed by
a finite set of variables ranged over by $V$. This technique was introduced
by~\citet{gesbert15} to interpret type variables (cf.~\cref{def:indint}), while ensuring that the model
we obtain is \emph{convex} in the sense of~\citet{castagna11}. Convexity is
a property that prevents the definition of meaningless subtyping relations, by
imposing that the interpretation of types changes uniformly for any possible
change of the interpretation of the type variables.\footnote{Formally,
convexity states that for every finite set of types, if for every interpretation
of the type variables this sets contains at least one empty type, then it is because it
contains a type that is empty for all interpretations.}  A sufficient condition
to satisfy convexity is that the interpretation maps every type into an infinite
set. Indexing each element of the domain with a finite set of variables is an
easy way to guarantee this since, for instance, even the interpretation of the
singleton type $c$ is the infinite set $\{c^V\alt V\in\Pf(\Vars)\}$.
In summary, the domain of \cref{def:interpretation} is the
one by~\citet{gesbert15}, but where pairs (inhabiting product types) are replaced by
record values of the form $\domrec\drow$ and rows.

\begin{definition}\label{def:row}
  Let $\R = \rec{(\ell = \tau_\ell)_{\ell \in L}}{\tl}$.
  We define $\rectorow\R = \row{(\ell = \tau_\ell)_{\ell \in L}}{\tl}{\emptyset}$.
  We extend this definition homomorphically to Boolean combinations of
  record type atoms.
\end{definition}

We have now all ingredients needed to define our set-theoretic interpretation
for the types:

\begin{definition}[Interpretation]
  \label{def:indint}%
  Let $\Domain$ be the domain of~\cref{def:interpretation} and $\Types$ the
  types of \cref{def:types}. We define a binary predicate $(\dterm : \tterm)$
  (“the element $\dterm$ belongs to $\tterm$”) on $\Domain\times\Types \cup
  \Domain_\Undef \times \Types_\Undef \cup \Domainrow \times \Rows$
  by induction on the pair $(\dterm, \tterm)$ ordered lexicographically.
  The predicate is only defined if $\dterm$ is coherent with the kind of
  $\tterm$:
  $\dterm = d$ if $\tterm = t$, $\dterm = \dundef$ if $\tterm = \tau$, and
  $\dterm = \drow$ if $\tterm = \rw$ and $\dom(\dterm) = \rdef(\rw)$.

  \begin{description}[left=0mm]
    \item[{\framebox[14mm]{\rm Types:}}]
     \qquad\!\!\!\(\begin{array}[t]{c}
        (d : \alpha)
        = \alpha \in \Tag(d) \label{due}
        \qquad (c^V : b)
        = c \in \ConstantsInBasicType(b)
        \qquad (\domrec{\drow}^V:\R)
        = (\drow:\rectorow{\R})
\\[.2mm]        (\{(d_1, \domega_1), \dots, (d_n, \domega_n)\}^V : t_1 \to t_2)
        = \forall i \in [1.. n] . \:
        \mathsf{if} \: (d_i : t_1) \mathrel{\mathsf{then}} (\domega_i : t_2)
        \end{array}\)

    \item[{\framebox[14mm]{\rm Fields:}}] \qquad\(        (\dundef : \fvar) = \fvar \in \Tag(\dundef)
        \qquad\qquad\qquad (\Undef^V : \Undef) = \mathsf{true}\)\\[-1mm]

    \item[{\framebox[14mm]{\rm Rows:}}]\mbox{}\\[-9mm]
      \begin{equation}
        \label{eq:rowint}%
        \begin{split}
          \;\;(\domrow{(\ell = \dundef_\ell)_{\ell \in L_1}}{L_2}^V:\rw)
        &= (\forall \ell \in L_1. (\dundef_\ell:\rw(\ell)))\\[-1.8mm]
        &\phantom{=}\ \textand
        (\forall\ell \in \rdef(\rw)\setminus L_1. (\Undef^\emptyset:\rw(\ell)))\\[-1.8mm]
        &\phantom{=}\ \textand \tail(\rw) = \rho \implies \rho \in V
        \end{split}\vspace{-2mm}
    \end{equation}
    
    \item[{\framebox[14mm]{\rm All:}}]
      \qquad\!\!\!\(\begin{array}[t]{rll}
        (\dterm : \tterm_1 \lor \tterm_2)
        &=\quad(\dterm : \tterm_1) \textor (\dterm : \tterm_2)\hspace*{8mm}
        &\text{if $\tterm_1$, $\tterm_2$ of the same kind}\\
        (\dterm : \lnot \tterm)  &=\quad\mathsf{not} \: (\dterm : \tterm)
        &\text{if the kinds of $\dterm$ and $\tterm$ correspond} \\
        (\dterm : \tterm)  &=\quad\mathsf{false} & \text{otherwise}
        \end{array}\)
  \end{description}
  We define the interpretation $\TypeInter{\cdot} : \Types \to \Pd(\Domain)$
  as $\TypeInter{t} = \{d \in \Domain \mid (d : t)\}$.
\end{definition}\noindent
We cannot define $\TypeInter{\cdot}$ by induction on types, since
their coinductive definition would make the definition
ill-founded. Thus, \cref{def:indint} uses the predicate
$(\dterm : \tterm)$ for which an inductive definition is possible thanks to the inductive definition of $\Domain$. 
The interpretation of types in \cref{def:indint} is mostly the same as in~\cite{gesbert15}: a type variable $\alpha$ is interpreted as the set of all elements tagged by $\alpha$; a type $t_1\to t_2$ is the set all the finite approximations of functions that map inputs of type $t_1$ into results of type $t_2$; and union, intersection, and negation types are mapped into the corresponding set-theoretic counterparts.  The main difference with \cite{gesbert15} is
the interpretation of rows and, thus, of record types. Equation~\eqref{eq:rowint}
defines when a row element is in the
interpretation of a row  $r$: it requires that all the components of
the row element are in the interpretations of the types specified by
$r$ (first two lines), and if the tail of $r$ is a row variable,
then it must index the row element (last line).

\begin{definition}[Subtyping]\label{def:subtyping}
Let  $\TypeInter{\cdot} : \Types \to \Pd(\Domain)$ be the interpretation from \cref{def:indint}. It induces the
following subtype relation in $\Types\Times\Types$:\\[1pt]
\centerline{
\(
t_1\leq t_2\iffdef\TypeInter {t_1}\subseteq\TypeInter{t_2}\)}\\[2pt]
The interpretation also induces the subfield relation in
$\Types_\Undef\Times\Types_\Undef$  and subrow relation in
$\Rows\Times\Rows$ defined as\\[1pt]
\centerline{
\( \tau_1\leq\tau_2 \iffdef \IntF{\TypeInter
{\tau_1}}\subseteq\IntF{\TypeInter{\tau_2}}\qquad r_1\leq r_2 \iffdef \IntR{\TypeInter
{r_1}}\subseteq\IntR{\TypeInter{r_2}}\)}\\[2pt]
where, the interpretation $\IntF{\TypeInter{\cdot}}
  : \Types_\Undef \to \Pd(\Domain_\Undef)$ is defined as
  $\IntF{\TypeInter{\tau}} = \{\dundef \in \Domain_\Undef \mid
  (\dundef:\tau)\}$, and the interpretation $\IntR{\TypeInter{\cdot}}
  : \Rows \to \Pd(\Domainrow)$ is defined as $\IntR{\TypeInter{\rw}}
  = \{\drow \in \Domainrow \mid (\dbar:\rw)\}$.
  \end{definition}

\subsection{Deciding Subtyping}
\label{sec:subtyping}

Now that subtyping is defined, we need an
effective decision procedure sound and complete
with respect to this definition.
Deciding subtyping amounts to deciding the emptiness of a type, since $t_1 \leq
t_2$ is equivalent to $t_1 \wedge \neg t_2 \leq \Empty$.
From \cite{frisch08}, we know that any type can be equivalently rewritten into a
disjunctive normal form (DNF) of the form $\bigvee_{i\in I}(\bigwedge_{a\in
P_i} a\wedge\bigwedge_{a\in N_i}\neg a \wedge \bigwedge_{\alpha \in
V^p_i} \alpha \wedge \bigwedge_{\alpha \in V^n_i} \neg\alpha)$ where each intersection
contains only atoms $a$'s with the same type constructors: they are
all basic types, or all arrows, or all records.
Thus, checking emptiness of a type amounts to checking emptiness of all
these intersections.

We suppose the sets $V_i^p$ and $V_i^n$ to be disjoint, as otherwise the
$i$-th intersection is trivially empty and can be discarded.
Then, emptiness of the intersections cannot depend on the type variables
whose intersection is not empty.
Thus, we just need decision procedures for the emptiness of the
$\bigwedge_{a\in P_i} a\wedge\bigwedge_{a\in N_i}\neg a$ parts.
Those are already known for every intersection of atoms, except for polymorphic records.
What is still missing is a formula that characterizes the emptiness of an intersection of the form
$\bigwedge_{\R \in P} \R \wedge \bigwedge_{\R \in N} \neg\R$, that is,
that decides whether $\bigwedge_{\R \in P} \R \leq \bigvee_{\R \in N} \R$ holds.

To rephrase, given any type $t$, the subtyping procedure recursively apply
these two steps:
\begin{enumerate}
  \item  Reduce $t$ to a DNF $\bigvee_{i\in I} t_i$ with $t_i =
    \bigwedge_{a\in P_i} a\wedge\bigwedge_{a\in N_i}\neg a
    \wedge \bigwedge_{\alpha \in V^p_i} \alpha \wedge \bigwedge_{\alpha \in V^n_i} \neg\alpha$;
  \item Check the emptiness of each $t_i$ by checking if
    $\bigwedge_{a \in P_i} a \leq \bigvee_{a \in N_i} a$ is empty:
    \begin{itemize}[leftmargin=10pt]
      \item[-] If the atoms are not records, we use the existing corresponding
        functions given in \cite{frisch04}.
      \item[-] If they are (polymorphic) records, we use the new algorithm
        that we describe below.
    \end{itemize}
\end{enumerate}

\paragraph{Subtyping algorithm}

Let $t = \bigwedge_{\R \in P} \R \wedge \bigwedge_{\R \in N} \neg\R$.
  Deciding emptiness of this type is done in two main steps.
First, preprocess $t$ by normalizing the positive
side of the type to \emph{isolate} row variables.
For this, we use the equivalence between
$\rec{(\ell = \tau_\ell)_{\ell \in L}}{\tl}$ and
$\orec{(\ell = \tau_\ell)_{\ell \in L}} \wedge \rec{(\ell = \Any \lor
\Undef)_{\ell \in L}}{\tl}$.
Note that a type $\rec{(\ell = \Any \lor \Undef)_{\ell \in L}}{\tl}$ will be
abbreviated as $\rec{L}{\tl}$.
Second, apply the function $\Phi$, the core of our algorithm that we describe
below.

Let $L = \bigcup_{\rw \in P \cup N} \fin(\rw)$ be the set of labels appearing
explicitely in every (top-level) record atom of $t$.
Our starting type $t$ is equivalent to the following intermediate one:\vspace*{-1.3mm}
\begin{equation}
  \label{eq:subtyping_two}%
  \orec{(\ell = {\textstyle\bigwedge_{\R \in P} \R(\ell)})_{\ell \in L}}
  \wedge \bigwedge_{\R \in P} \rec {\fin(\R)}{\tail(\R)}
  \wedge \bigwedge_{\R \in N} \neg \R \vspace{-2.1mm}
\end{equation}
This type was obtained by transforming the positive $\bigwedge_{\R \in P} \R$ part of $t$: we merged the fields over $L$ into a single atomic record type,
and grouped the tails of positive records in a separate intersection.
Next, we are going to rewrite the type in \eqref{eq:subtyping_two} by
splitting the middle intersection $\bigwedge_{\R \in P} \rec
{\fin(\R)}{\tail(\R)}$ in two: an intersection with all the record
atoms whose tail is a row variable, and all the others that we will merge
with the leftmost record in \eqref{eq:subtyping_two}. For that, let us define
$\tl_\circ$ to represent the intersection of the tails of
the records in $P$ whenever this tail is either $\crecsign$ or $\orecsign$, that
is, $\tl_\circ = \crecsign$ if there is $\R \in P$ such that $\tail(\R) = \crecsign$, and
$\tl_\circ = \orecsign$ otherwise.
If we take all the records in $\bigwedge_{\R \in P} \rec
{\fin(\R)}{\tail(\R)}$ whose tail is not a row variable and intersect
them with the leftmost record in \eqref{eq:subtyping_two}, then
we obtain the record type $\R_\circ = \rec{(\ell = {\textstyle\bigwedge_{\R \in P}
\R(\ell)})_{\ell \in L}}{\tl_\circ}$.
Notice that $\R_\circ$ is a monomorphic record type.
For the remaining records in the middle intersection, let us denote by $V_p = \{\rho \mid \exists \R \in P. \tail(\R) = \rho\}$ the set of all
top-level type variables occurring in $P$.
The intersection of atoms in \eqref{eq:subtyping_two} is then equivalent to the following type, which is the one
for which we have to decide emptiness:\vspace*{-1.3mm}
\begin{equation}\label{eq:subtyping_three}
  \R_\circ
  \wedge \bigwedge_{\rho \in V_p} \rec{\Labels{\setminus}\rdef(\rho)}{\rho}
  \wedge \bigwedge_{\R \in N} \neg \R\vspace{-2mm}
\end{equation}
The second step of our algorithm is realized by the function
$\Phi(\R_\circ, V_p, N)$, where $\R_\circ\not\leq\Empty$:\\[.3mm]
\( \begin{array}{rcl}
  \Phi(\R_\circ, V_p, \emptyset)
  &\coloneq& \textsf{false}\\
  \Phi(\R_\circ, V_p, N \cup \{\R\})
  &\coloneq& \textsf{if } (\tail(\R) = \orecsign\ \textsf{ or}\ \tail(\R) = \tail(\R_\circ)
  \ \textsf{ or}\ \tail(\R) \in V_p) \textsf{ then}\\
  &&\quad \forall \ell{\in}\fin(\R_\circ).\ (\R_\circ(\ell) \leq \R(\ell) \textsf{ or }
  \Phi(\R_\circ \wedge \orec{\ell:\neg\R(\ell)}, V_p, N))\\
  &&\textsf{else } \Phi(\R_\circ, V_p, N)
\end{array} \)\\[.3mm]
The function must decide whether $\R_\circ \wedge \bigwedge_{\rho \in V_p} \rec{\Labels{\setminus}\rdef(\rho)}{\rho}
\leq \bigvee_{\R\in N} \R$, so it  picks an $\R\in N$ and
generates the conditions to test the containment. The first clause of
the definition states that if we already examined all $\R\in
N$, then subtyping does not hold, since $\R_\circ\nleq\Empty$ and
so its intersection with some row variables is also non-empty. If $\R_\circ$ is
open and $\R$ is closed, or if $\R$ is polymorphic, but its row
variable is not one in $V_p$, then the
containment cannot come from this particular $\R$, and it is
discarded: this corresponds to the \textsf{else} branch of second
clause. Otherwise, we are in the case in which either $\R_\circ$ is
closed, or we are comparing two records types with a common row
variable, or $\R$ is open. In these cases we compare $\R_\circ$ and $\R$
component-wise, and for each $\ell$ we check that either
$\R_\circ(\ell)\leq\R(\ell)$ or that the part that is in excess in
$\R_\circ(\ell)$ is contained in the records remaining in
$N$.\footnote{This formula generalizes
  the decomposition for tuples: e.g., if
  $t_\circ\Times t_\circ'\not\leq\Empty$ then
  $t_\circ\Times t_\circ'\leq t_1\Times t_1' \vee  t_2\Times t_2' \iff
  (t_\circ\leq t_1\textor (t_\circ{\setminus} t_1)\Times t_1'\leq
  t_2\Times t_2')\textsf{ and } (t_\circ'\leq t_1'\textor t_1\Times (t_\circ'{\setminus} t_1')\leq
  t_2\Times t_2')$: see~\cite[Appendix D]{castagna23a} for a
longer explanation.}

\loic{Nouvel exemple}
\begin{example}[Subtyping]
  Consider the verification of the following subtyping relation:\\[.3mm]
  \centerline{\(
    \rec{\Keyw{a} = \True, \Keyw{b} = \Int {\lor} \Bool}{\rho}
    \wedge \orec{\Keyw{b} = \Int {\lor} \String, \Keyw{c} = \Int}
    \leq \rec{\Keyw{a} = \Bool, \Keyw{b} = \Int}{\rho}
    \vee \rec{\Keyw{a} = \Int}{\rho'}
  \)}\\[.3mm]
  The first step of the subtyping algorithm isolates $\rho$, reducing the relation by equivalence into:
  \\[.3mm]
  \centerline{\(
    \orec{\Keyw{a} = \True, \Keyw{b} = \Int, \Keyw{c} = \Int}
    \wedge \rec{\Keyw{a}, \Keyw{b}}{\rho}
    \wedge \neg\rec{\Keyw{a} = \Bool, \Keyw{b} = \Int}{\rho}
    \wedge \neg\rec{\Keyw{a} = \Int}{\rho'}
    \leq \Empty
    \)}\\[.3mm]
  Following the notation used in equation~\eqref{eq:subtyping_three}, we have
  $\R_\circ = \orec{\Keyw{a} = \True, \Keyw{b} = \Int, \Keyw{c} =
  \Int}$,
  $V_p = \{\rho\}$ and
  $N = \{\rec{\Keyw{a} = \Bool, \Keyw{b} = \Int}{\rho},
  \rec{\Keyw{a} = \Int}{\rho'}\}$.
  The function $\Phi(\R_\circ, V_p, N)$ verifies subtyping in three steps.
  In the first step, it selects $\rec{\Keyw{a} = \Int}{\rho'}$ from $N$,
  which is discarded (i.e., in the second clause for $\Phi$, the \textsf{else} branch is taken), since $\rho' \notin V_p$.
  Second, it selects $\rec{\Keyw{a} = \Bool, \Keyw{b} = \Int}{\rho}$ from $N$.
  Since $\rho \in V_p$, it finishes by checking every label $\Keyw{a},
  \Keyw{b}$ and $\Keyw{c}$ w.r.t. $\R_\circ$ (i.e., the \textsf{then} branch in $\Phi$).
  Here, it successfully concludes for each of them, by directly checking the
  subtyping.\hfill\qed
\end{example}

The function presented generalizes the version for monomorphic records
given in \cite{castagna23a} and currently used in Elixir.
Interestingly, the sole difference is the addition of the test ($\tail(\R) \in V_p$).
Since $V_p$ is constant in the function, then the complexity of
this function and of its monomorphic version are the same.

\begin{restatable}[Soundness and completeness of $\Phi$]{lemma}{correctsubalg}
  \label{l:correct_subalg}%
  Let $\R_\circ$ be a monomorphic record type, $V_p \subset \rowVars$
  finite and $N$ a finite set of (polymorphic) atomic record types. Then,
  \[
    \R_\circ \wedge \bigwedge_{\rho \in V_p} \rec{\Labels{\setminus}\rdef(\rho)}{\rho}
    \leq \bigvee_{\R \in N} \R \iff
    \R_\circ \leq \Empty \textor \Phi(\R_\circ, V_p, N).
  \]\vspace{-3mm}
\end{restatable}

\begin{restatable}{proposition}{termsubalg}
  \label{l:term_subalg}%
  The subtyping algorithm terminates. As a corollary, subtyping is decidable.
\end{restatable}

\subsection{Substitutions}%
\label{sec:rowsub}
The upcoming descriptions of the type system and of the inference algorithm
rely on type, row, and field substitutions.
\begin{definition}
  Substitutions, ranged over by $\sigma$, are total mappings from
  variables of kind $\kappa$ to terms of kind $\kappa$ (i.e., type
  variables to types, field variables to field-types, and row variables of
  definition space $L$ to rows of definition space $L$) that are the identity everywhere except on
  a finite set of variables. This set is called the \emph{domain}
  of the substitution $\sigma$ and is defined as
  $\dom(\sigma) = \{\alpha \mid \sigma(\alpha) \neq \alpha\}
  \cup \{\fvar \mid \sigma(\fvar) \neq \fvar\}
  \cup \{\rho \mid \sigma(\rho) \neq \row{}\rho{\Labels{\setminus}{\rdef(\rho)}}\}$.
\end{definition}

The application of a substitution $\sigma$ to a term $\tterm$ is denoted by
$\tterm\sigma$.
Notice that the application is defined both on field-types and  on rows, the
latter being useful only for tallying.
The application of a substitution must satisfy the following equalities.
\begin{equation*}
  \begin{array}{c@{\qquad}c@{\qquad}c@{\qquad}c}
    \alpha\sigma = \sigma(\alpha)
    & b\sigma = b
    & \Empty\sigma = \Empty
    &(t_1 \to t_2)\sigma = t_1\sigma \to t_2\sigma
    \\ \fvar\sigma = \sigma(\fvar)
    & \Undef\sigma = \Undef
    & (\neg \tterm)\sigma = \neg (\tterm\sigma)
    &(\tterm_1 \lor \tterm_2)\sigma = \tterm_1\sigma \lor \tterm_2\sigma
  \end{array}\vspace{-1.5mm}
\end{equation*}
\begin{equation}
  \label{eq:subrec}%
  \rec{(\ell = \tau_\ell)_{\ell \in L}}{\tl} \sigma
  = \begin{cases}
    \rec{(\ell = \tau_\ell\sigma)_{\ell \in L}}{\sigma(\rho)},
    &\text{ if } \tl = \rho \\
    \rec{(\ell = \tau_\ell\sigma)_{\ell \in L}}{\tl},
    &\text{ otherwise.}
  \end{cases}
\end{equation}\vspace{-1.5mm}
Where $\rec{(\ell = \tau_\ell)_{\ell \in L}}{\rw}
\eqdef \orec{(\ell = \tau_\ell)_{\ell \in L}} \wedge \rec{L}{\rw}$ and:
\begin{equation*}
  \rec{L}{\row{(\ell = \tau_\ell)_{\ell \in L'}}{\tl}{L}}
  \eqdef \rec{L, (\ell = \tau_\ell)_{\ell \in L'}}{\tl}
  \quad \rec{L}{\rw_1 \vee \rw_2} \eqdef \rec{L}{\rw_1} \vee \rec{L}{\rw_2}
  \quad \rec L {\neg\rw} \eqdef \neg\rec L \rw
\end{equation*}
The equalities above are standard, apart from the one
in \eqref{eq:subrec} which needs a definition for the notation
$\rec{(\ell = \tau_\ell\sigma)_{\ell \in L}}{\sigma(\rho)}$,
since $\sigma(\rho)$ is a row rather than a tail. The definition is
given right after \eqref{eq:subrec}
and simply states that $\rec{(\ell = \tau_\ell\sigma)_{\ell \in
L}}{\sigma(\rho)}$ stands for the record type obtained by recursively decomposing
the Boolean combinations of the rows in $\sigma(\rho)$, until we arrive at single rows that are
expanded in the record type (recall that rows are inductively defined).
Substitution for rows is defined in the same way as for records (it
suffices to
change the delimiting brackets).

As expected, if $\dom(\sigma) = \emptyset$, then $\tterm\sigma = \tterm$.
If $\sigma(\rho) \leq \Empty$ and $\tail(\R) = \rho$,
then $\R\sigma \leq \Empty$.
Thanks to the parametric interpretation of types, substitution
preserves subtyping:

\begin{restatable}{proposition}{subpreservesub}
  If $t_1 \leq t_2$, then $t_1\sigma \leq t_2\sigma$ for any row substitution
  $\sigma$.
\end{restatable}

\section{Language}
\label{sec:language}
We define the syntax, static and dynamic semantics of
a record calculus that we prove to be type sound
(\cref{sec:syntax}) and define a sound and complete typing algorithm for it (\cref{sec:typealgo}).
\subsection{Syntax and Semantics} \label{sec:syntax}
\[
  \begin{array}{lrclr}
    \textbf{Expressions} & e &\Coloneqq&
    c \alt
    x \alt
    e e \alt
    \lambda^{\wedge_{i \in I} (t_i \to t_i')}x.e \alt
    \erecord{} \alt
    \extrecord e \ell e \alt
    e.\ell \alt
    \delrecord e \ell
    \\
    \textbf{Values} & v &\Coloneqq&
    c \alt
    \lambda^{\wedge_{i \in I} (t_i \to t_i')}x.e \alt
    \erecord{} \alt
    \extrecord v \ell v
    \\
    \textbf{Evaluation contexts} & \evalctx &\Coloneqq&
    \ectx \alt
    \evalctx e \alt
    v \evalctx \alt
    \extrecord e \ell \evalctx \alt
    \extrecord \evalctx \ell v \alt
    \evalctx.\ell \alt
    \delrecord \evalctx \ell
  \end{array}
  \]
The syntax above describes a functional language with constants, functions, and
records with field selection ($e.\ell$), addition  ($\extrecord e \ell
e$), and deletion ($\delrecord e \ell$). As customary in
semantic subtyping, $\lambda$-abstractions are annotated by their type,
which is an
intersection of arrow types. We use
$\erecord{\ell_1 = e_1, \dots, \ell_n = e_n}$ as syntactic sugar for
\(  \extrecord{\dots\extrecord{\erecord{}}{\ell_1}{e_1}\dots} {\ell_n}{e_n}\).

The  semantics of the language is given by the call-by-value weak reduction defined
below:
\[ \begin{array}{rr@{\;}c@{\;}ll}
  \rrulename{app} &
  (\lambda^t x.e)v &\reduces& e \subs x v\\
  \rrulename[=]{sel} &
  \extrecord v {\ell}{v'}.\ell &\reduces& v'\\
  \rrulename[\neq]{sel} &
  \extrecord v {\ell'}{v'}.\ell &\reduces& v.\ell &\text{ if } \ell \neq \ell'\\
  \rrulename[=]{del} &
  \delrecord{\extrecord v {\ell}{v'}} \ell &\reduces&
  \delrecord v \ell\\
  \rrulename[\neq]{del} &
  \delrecord{\extrecord v {\ell'}{v'}} \ell &\reduces&
  \extrecord{\delrecord v \ell}{\ell'}{v'} &\text{ if } \ell \neq \ell'\\
  \rrulename{emp} &
  \delrecord{\erecord{}}\ell &\reduces& \erecord{}\\
  \rrulename{ctx} &
  \ctx{e} &\reduces& \ctx{e'} &\text{ if } e \reduces e'
\end{array} \]
where $e \subs x v$ is the term obtained by standard capture-avoiding substitution of
$v$ for $x$ in $e$, defined modulo
$\alpha$-equivalence.
Notice that the deletion of a label $\ell$ is defined for the empty
record $\erecord{}$ but selection is not: selection requires
the presence of the field $\ell$ while deletion does not.

\begin{figure}\vspace{-.5mm}
  \begin{mathpar}
    \inferrule*[left=\rulename{Const}]
    { }{\decseq \Delta \Gamma c {\basic(c)}}
    \and \inferrule*[left=\rulename{Var},right=\(x \in \dom(\Gamma)\)]
    { }{\decseq \Delta \Gamma x {\Gamma(x)}}
    \vspace{-1.4mm}\\
     \inferrule*[left=\rulename{Abs},right=\(\Delta'{=} \vars(\textstyle\bigwedge_{i \in I} t_i \to t_i')\)]
    {(\decseq {\Delta \cup \Delta'}{\Gamma, x:t_i} e {t_i'})_{i \in I}}
    {\decseq \Delta \Gamma {\lambda^{\wedge_{i \in I} (t_i \to t_i')} x.e}
    {\textstyle\bigwedge_{i \in I} (t_i \to t_i')}}
    \and \inferrule*[left=\rulename{App}]
    {\decseq \Delta \Gamma {e_1} {t_1 \to t_2}
    \and \decseq \Delta \Gamma {e_2} {t_1}}
    {\decseq \Delta \Gamma {e_1 e_2} {t_2}}
    \and \inferrule*[left=\rulename{Emp}]
    { }{\decseq \Delta \Gamma {\erecord{}}{\crec{}}}
    \vspace{-1.4mm}\\
    \inferrule*[left=\rulename{Ext}]
    {\decseq \Delta \Gamma {e} {t \leq \orec{\ell = \Undef}}
    \and \decseq \Delta \Gamma {e'} {t'}}
    {\decseq \Delta \Gamma {\extrecord{e}{\ell}{e'}}
    {\rec{\ell = t'}{\delrec{t}{\ell}}}}
    \and \inferrule*[left=\rulename{Del}]
    {\decseq \Delta \Gamma e {t \leq \orec{}}}
    {\decseq \Delta \Gamma {\delrecord e \ell}
    {\rec{\ell = \Undef}{\delrec t \ell}}}
    \vspace{-1.4mm}\\
    \inferrule*[left=\rulename{Sel}]
    {\decseq \Delta \Gamma e {\orec{\ell = t}}}
    {\decseq \Delta \Gamma {e.\ell} t}
    \and \inferrule*[left=\rulename{Inter}]
    {\decseq \Delta \Gamma e {t_1} \and \decseq \Delta \Gamma e {t_2}}
    {\decseq \Delta \Gamma e {t_1 \wedge t_2}}
    \vspace{-1.4mm}\\
    \inferrule*[left=\rulename{Sub}]
    {\decseq \Delta \Gamma e {t' \leq t}}
    {\decseq \Delta \Gamma e t}
    \and \inferrule*[left=\rulename{Inst},
    right={$\dom(\sigma) \cap \Delta = \emptyset$}]
    {\decseq \Delta \Gamma e t}
    {\decseq \Delta \Gamma e {t\sigma}}
  \end{mathpar}
  \caption{Declarative type system}
  \label{fig:decrules}
\end{figure}

The terms of the language are typed by the declarative type system
in \cref{fig:decrules}, whose judgments have the form
$\decseq \Delta \Gamma {e} {t}$, where $\Delta\subseteq\Pf(\typeVars{\cup}\fldVars{\cup}\rowVars)$ is a set of
monomorphic variables (i.e., variables that cannot be instantiated)
and $\Gamma$ a type environment from expression variables to types.

The rules for the functional part are inspired from those
by~\citet{castagna14}.%
\iflongversion
\else\ 
\fi
Constants are typed by a given function
$\basic$ which maps each constant to its basic type
\rulename{const}.%
\footnote{
The functions $\basic$ and $\ConstantsInBasicType$ used in
  \cref{def:indint} must satisfy
$c\in\ConstantsInBasicType(\basic(c))$ for all $c\in\Constants$.}
Rule~\rulename{Abs} checks that a function has all the types declared in its
annotation: for each $t_i\To t_i'$, it checks that the body $e$ is
of type $t_i'$ under the environment in
which $x$ is given type $t_i$ and the set $\Delta'$ of
\emph{all} the variables in the annotation is added to the set of monomorphic
variables
($\vars(t)$ returns the set of type, field, and row variables in
$t$: cf.\ \cref{def:vars}).
The rules for intersection introduction \rulename{Inter} and subsumption
\rulename{Sub} are the usual ones: if an expression $e$ has two
types, then it also has their intersection; if an expression $e$ has
type $t'$, then it also has any supertype of $t´$ (we use the notation
$\decseq \Delta \Gamma {e} {t'}\leq t$ in the premises of a rule, to
indicate that the rule has premise $\decseq \Delta \Gamma {e} {t}$ and
side condition $t'\leq t$). The instantiation rule~\rulename{Inst},
can instantiate any type variable unless it is in the set of monomorphic variables $\Delta$, as this
would be unsound.

The new rules of this system are those for record expressions and their
operations. The empty record value has the closed empty record type:
\rulename{Emp}. Rule~\rulename{Sel} states that selection is typable only
if the selected field is present, in which case its type is given to the select
expression. Rule~\rulename{Ext} types a strict extension of an expression $e$ of
type $t$ by the expression $e'$ on label $\ell$, only if the field $\ell$ is
undefined in $e$, that is, the type of $e$ is a subtype of $\orec{\ell =
\Undef}$. Rule~\rulename{Del} states that we can delete a field $\ell$ from
an expression $e$ provided that it is  record (i.e., its
type is a subtype of $\orec{}$).

The types of the expressions typed by~\rulename{Ext} and~\rulename{Del} are
both obtained in similar ways.
First, we compute the operator $\delrec t \ell$ ---whose formal definition we
give below--- which returns $\rectorow t$ truncated by the field of label $\ell$.
Then, we put back the field of label $\ell$ with the desired field-type ($t'$ or
$\Undef$) using the operation $\rec{(\ell = \tau_\ell)_{\ell \in
L}}{\rw}$ we defined in \cref{sec:rowsub} for substitutions (cf.\ \cref{eq:subrec}).

We define the operator $\delrec t \ell$, on DNFs:
let $t = \bigvee_{i \in I} \bigwedge_{\R \in P_i} \R
\wedge \bigwedge_{\R \in N_i} \neg\R
\wedge \bigwedge_{\alpha \in V_i} \alpha
\wedge \bigwedge_{\alpha \in V_n} \neg\alpha$.
Then, $\delrec t \ell \eqdef
\bigvee_{i \in I} \bigwedge_{\R \in P_i} \delrec \R \ell
\wedge \bigwedge_{\R \in N_i} \delrec {(\neg\R)} \ell$ where for the
literals $\R$ and $\neg\R$ the definition is: \\[1mm]
\centerline{\(
  \begin{array}{r{c}ll}
    \delrec{\rec{({\elll} = \tau_{\elll})_{{\elll} \in L}}{\tl}}\ell
    & \!{=}\! & \begin{cases}
      \row{({\elll} = \tau_{\elll})_{{\elll} \in L \setminus \{\ell\}}}{\tl}{\{\ell\}}
    &\text{if } \ell \in L \text{ or } \tl \notin \rowVars\\
    \orow{({\elll} = \tau_{\elll})_{{\elll} \in L}}{\{\ell\}}
    &\text{otherwise}
    \end{cases}\\
    \delrec{(\neg\rec{({\elll} = \tau_{\elll})_{{\elll} \in L}}{\tl})}\ell
    & \!{=}\! & \begin{cases}
      \neg\row{({\elll} = \tau_{\elll})_{{\elll} \in L \setminus \{\ell\}}}{\tl}{\{\ell\}}
      &\hspace{-0.55em}\text{if } (\ell {\in} L \textand \tau_\ell = \Any {\lor} \Undef)
      \textor (\ell {\notin} L \textand \tl = \orecsign)\\
      \orow{}{\{\ell\}}
      &\hspace{-0.55em}\text{otherwise}
  \end{cases}
\end{array}\)}\\[1mm]
\noindent
First, notice that the type variables in the DNF are simply erased.
This is not restrictive in practice, as intersections with top-level type
variables are used only to implement bounded polymorphism which, as argued in
the introduction, cannot be used for extensions and deletions, these requiring
instead the use of row variables. The definition for the positive
literal  $\R$ consists of two cases. The first case is the intuitive one,
with no interference from a negation or a row variable: the row is
undefined on $\ell$, so we remove the field for $\ell$ (if any) and
index the row by $\{\ell\}$.
In the second case, $\R$ is polymorphic and $\ell$ is in the definition space of
its tail; therefore, similarly to what we do for type variables, we subsume $\R$ to an open record
before deleting $\ell$.
For instance, if $x: \rec{a = \Int}{\rho}$,
then the type deduced for $\delrec x b$ is $\orec{a = \Int, b = \Undef}$:
the information bound to $\rho$ must be forgotten since $b$ is in definition
space of $\rho$.
\footnote{\setlength{\baselineskip}{3pt}
  We can avoid this in practice.
  For instance, we can give to
  \felix{def add_delete(x), do: Map.delete(!\%!{x | a: x.a + 1}, :b)}
  the type:
  \felix{!\%!{f, a: integer()} -> !\%!{f, a: integer(), :b => none()} when f: fields()},
  where \felix{f} plays the same role as $\rho$ in the types above.
  Here we do not need to replace the second occurrence of \texttt{f} by
  ``\texttt{...}'' since, as we explained for the code in lines~\ref{nineteen}
  and \ref{twenty}, we enforce in practice that this type stands for
  \felix{!\%!{f, a: integer(), :b => term()} -> !\%!{f, a: integer(), :b => none()} when f: fields()}
and, thus, \felix{:b} is not in the definition space of \mbox{\felix{f}.}}
%

Next consider the case of a negative atom $\neg\R$.
First, suppose $\R(\ell) \neq \Any \lor \Undef$ (i.e., either $\ell \in L$ and
$\tau_\ell \neq \Any \lor \Undef$ or $\ell \notin L$ and $\tl = \crecsign$).
Then, the type $\neg\R$ contains, among others, all row elements such that the
field $\ell$ is not of type $\R(\ell)$ (since $\R(\ell) \neq \Any \lor \Undef$,
then there exists at least one such element), and every other field is of
arbitrary value.
Hence, the set obtained from removing the field $\ell$ from these values gives
all quasi-constant functions on $\Labels \setminus \{\ell\}$, without
restriction: the set entailed by $\orow{}{\{\ell\}}$.
Now, when $\R(\ell) = \Any \lor \Undef$, there are two possibilities.
The first is $\tl = \rho$ and $\ell \in L$.
There is a clash, and we subsume $\neg\R$ to the
top record type $\orec{}$ before deleting $\ell$, as we did in the positive case.
Note that $\orec{}$ is the smallest type to which we can subsume $\neg\R$.
The second possibility is given by the third case of the definition
(i.e., the first case for the negated literal).
There, it is impossible to have a value of a type different from
$\R(\ell) = \Any \lor \Undef$.
Hence, the constraints on the type must be on the other fields: we
simply remove $\ell$ from the definition definition space of the row.

The language satisfies the property of soundness.
Its proof is routine (subject reduction and progress, using inversion and generation
lemmas) and can be found in \cref{app:language}.
\begin{restatable}[Type soundness]{theorem}{typesoundness}
  \label{l:type_soundness}%
  Let $e$ be a well-typed closed expression, that is,
  $\decseq \emptyset \emptyset e t$ for some $t$.
  Then either $e$ diverges or it reduces to a value of type $t$.
\end{restatable}

\begin{figure}
  \begin{mathpar}\vspace{-3mm}
    \inferrule*[left=\rulename{Abs},right=\(\Delta'{=} \vars(\textstyle\bigwedge_{i \in I} t_i \to t_i')\)]
    {\left(\algseq{\Delta \cup \Delta'}{\Gamma, x:t_i} e {t_i''}
      \leqsub{\Delta \cup \Delta'} t_i' \right)_{i \in I}}
    {\algseq \Delta \Gamma {\lambda^{\wedge_{i \in I} (t_i \to t_i')} x.e}
    {\textstyle\bigwedge_{i \in I} (t_i \to t_i')}}\\\vspace{-3mm}
    \inferrule*[left=\rulename{App},right=\( t \in t_1 \apply{\Delta} t_2\)]
    {\algseq \Delta \Gamma {e_1} {t_1} \and \algseq \Delta \Gamma {e_2} {t_2}}
    {\algseq \Delta \Gamma  {e_1e_2} t}
    \and \inferrule*[left=\rulename{Del},
    right=\(\rw \in \del t \ell\)]
    {\algseq \Delta \Gamma e t}
    {\algseq \Delta \Gamma {\delrecord e \ell}
    {\rec{\ell = \Undef}{\rw}}}
    \and \inferrule*[left=\rulename{Sel},
    right=\(t_\ell \in \sel(t)\)]
    {\algseq \Delta \Gamma e t}
    {\algseq \Delta \Gamma {e.\ell} t_\ell}
    \and \inferrule*[left=\rulename{Ext},
    right=\(\rw \in \delu t \ell\)]
    {\algseq \Delta \Gamma e t
    \and \algseq \Delta \Gamma {e'} {t'}}
    {\algseq \Delta \Gamma {\extrecord{e}{\ell}{e'}}
    {\rec{\ell = t'}{\rw}}}
  \end{mathpar}
  \caption{Algorithmic type
  system (plus \rulename{Const}, \rulename{Var}, \rulename{Emp},
  which are as in~\cref{fig:decrules})}
  \label{fig:algrules}%
\end{figure}

\subsection{Algorithmic Type System}\label{sec:typealgo} The system in
\cref{fig:decrules} is not algorithmic: it is not syntax-directed and some of
its rules are not analytic\footnote{A rule is \emph{analytic} (as opposed to
\emph{synthetic}) when the input (i.e., $\Gamma$ and $e$) of the judgment at the
conclusion is sufficient to determine the inputs of the judgments at the
premises (cf.\ \cite{ML1994,types2019}).}.
In \cref{fig:algrules}, we give an algorithmic system that is sound
and complete with respect to the system in \cref{fig:decrules} (we
omitted three rules that are the same as in \cref{fig:decrules}).
The new system includes the algorithmic counterparts of the typing rules for record
operations which, apart from \rulename{Emp}, must be changed to
account for the fact that there is no \rulename{Inst} rule in the
algorithmic system and, thus, instantiation must be performed by the
algorithmic elimination rules.
For instance, in the declarative system, if $x:\alpha$ and $\alpha\not\in\Delta$, then
$x.\ell: \Int$ can be deduced by instantiating by \rulename{Inst} $\alpha$ to $\orec{\ell=\Int}$.
In the algorithmic system, in the absence of \rulename{Inst}, this instantiation must be done within
the algorithmic
rule~\rulename{Sel}.
For all record operations, the algorithmic system needs to perform a possible
instantiation of the type of the record. This is done in the side
conditions of the rules for record operations by the following
operators (that we explain after the \cref{def:c1} for $\leqsub\Delta$), which instantiate
the type of the record to match the conditions in the declarative system:
\begin{align*}
  \sel(t) &= \{ t' \mid [\sigma_i]_{i \in I} \Vdash t
    \leqsub\Delta \orec{\ell = \Any}
  \textand t' = (\textstyle\bigwedge_{i \in I} t\sigma_i).\ell\}
\\
  \del t \ell &= \{ \rw \mid [\sigma_i]_{i \in I} \Vdash t
    \leqsub\Delta \orec{}
  \textand \rw = \delrec{(\textstyle\bigwedge_{i \in I} t\sigma_i)}\ell \}
\\
  \delu t \ell &= \{ \rw \mid [\sigma_i]_{i \in I} \Vdash t
    \leqsub\Delta \orec{\ell = \Undef}
  \textand \rw = \delrec{(\textstyle\bigwedge_{i \in I} t\sigma_i)}\ell \}
\end{align*}
These side conditions are primarily of theoretical interest to satisfy
completeness. They address the case of record expressions that
return (parametric) polymorphic values (e.g., of type $\alpha$), which
are never encountered in practice. Consequently, the sets in
these side conditions are never computed in practice. In a practical setting, the rules \rulename{Ext} and
\rulename{Del} of \cref{fig:decrules} should be used instead, while  the rule
to use in practice for selection has premises $\algseq \Delta \Gamma e t$, side-condition $t
\leq \orec{\ell = \Any}$, and conclusion $\algseq \Delta \Gamma
{e.\ell}{t.\ell}$. Keeping the rules of \cref{fig:decrules}, we loose completeness, but only in theory, and we gain in
efficiency (\cite[Appendix B.3]{castagna15} discusses this point in detail for type polymorphism).

In any case, selection uses a new type operator $t.\ell$ (theoretically, in its
side condition to compute $t'$; in practical setting, in its conclusion).
Since the algorithmic type system does not
include a subsumption rule, we cannot assume that the type $t$ deduced for the
expression $e$ in \rulename{Sel} will be a record type atom of the form required
by the declarative system (i.e., $\orec{\ell = t}$): in general, $t$ will be a
union of intersections of such atoms, type variables, and their negations. Thus
the rule checks that $t$ is a record type in which the field $\ell$ is surely
defined, and delegates to the operator $t.\ell$ (defined below) the computation
of the type of the result.
\begin{restatable}[Field Selection]{definition}{typesel}
  \label{def:typesel}
  Let $t \leq \orec{\ell = \Any}$ be a DNF.
  We define the selection of the field $\ell$ of $t$ as
  $(\bigvee_{i \in I} t_i).\ell \eqdef \bigvee_{i \in I} t_i.\ell$
  and\vspace{-2mm} \[
    (\bigwedge_{\R \in P} \R \wedge \bigwedge_{\R \in N} \neg\R
    \wedge \bigwedge_{\alpha \in V_p} \alpha
    \wedge \bigwedge_{\alpha \in V_n} \neg\alpha) . \ell
    \eqdef \bigvee_{N' \subseteq N} \Bigg(
      \bigwedge_{\R \in P} \R(\ell)
      \wedge \bigwedge_{\R \in N'} \neg\R(\ell)
    \Bigg)
  \]
\end{restatable}
The condition $t \leq \orec{\ell = \Any}$ assures that $t.\ell \leq \Any$,
so that selection always returns a type (and not a generic field-type).
Indeed, $t.\ell$ is equivalent to $\min\{t'\mid
t\leq\orec{\ell=t'}\}$ (\cref{app:typealgo}).
Once more, the presence of top-level
intersections with type variables does not play any role in selection.

To finish explaining the algorithmic system, we need to introduce
the notations $\Vdash t \leqsub{\Delta} t'$ and $t \apply\Delta t'$,
whose definitions are taken verbatim from~\cite{castagna15}:
\begin{definition}[\cite{castagna15}]\label{def:c1}
  Let $t$ and $t'$ be two types and $\Delta$ a set of variables.
  \begin{align*}
    [\sigma_i]_{i \in I} \Vdash t \leqsub\Delta t'
    &\iffdef \textstyle\bigwedge_{i \in I} t\sigma_i \leq t' \textand
      \forall i \in I. \dom(\sigma_i) \cap \Delta = \emptyset\\[-1mm]
      \Vdash t \leqsub{\Delta} t'
    &\iffdef \exists[\sigma_i]_{i \in I} \text{ such that }
      [\sigma_i]_{i \in I} \Vdash t \leqsub{\Delta} t'
  \end{align*}
\end{definition}

\begin{definition}[\cite{castagna15}]\label{def:c2}
  Let $t_1$ and $t_2$ be two types with  $t \leq \Empty \to \Any$, and $\Delta$ a set of variables.
  We define $t_1 \apply\Delta t_2$ as the set of types for which there exist two
  sets of type substitutions (for variables not in $\Delta$) that make $t_2$
  compatible with the domain of $t_1$ (defined below):
  \[
    t_1 \apply\Delta t_2 \eqdef \left\{
   ~ t ~~ \begin{array}{|l}
        {[\sigma_j]_{j \in J}} \Vdash t_1 \leqsub\Delta \Empty \to \Any\\
        {[\sigma_i]_{i \in I}} \Vdash t_2 \leqsub\Delta \arrdom(\bigwedge_{j \in J} t_1\sigma_j)\\
        t = \textstyle\bigwedge_{j \in J} t_1\sigma_j \mapply \bigwedge_{i \in I} t_2\sigma_i
    \end{array}\right\}
  \]
  Where $t_1 \mapply t_2 \eqdef \min\{t \mid t \leq t_1 \to t_2\}$.
  For an arrow type $t' \leq \Empty \to \Any$,
  we have
  $t' \simeq \bigvee_{i \in I} (\bigwedge_{p \in P_i} (t_p \to t'_p) \wedge
  \bigwedge_{n \in N_i} \neg(t_n \to t'_n) \wedge \bigwedge_{\alpha \in
  V^p_i}
  \alpha \wedge \bigwedge_{\alpha \in V^n_i} \neg\alpha)$, and define
  $\arrdom(t') = \bigwedge_{i \in I} \bigvee_{p \in P_i} t_p$.
\end{definition}
These two definitions are used in the rule \rulename{App}, the key rule for the
algorithmic system (which again is taken verbatim from \cite{castagna15} where
more details can be found).
Essentially, \rulename{App} merges together intersection elimination (in this
case the standard terminology is \emph{expansion}), instantiation, and subsumption.
For the application $e_1e_2$ to be well typed, the type of the function must be
a functional type (i.e., a subtype of $\Empty\to\Any$, the type of all
functions) whose domain is a supertype of the type of the argument.
Therefore, the rule looks for two finite sets of type substitutions for the
variables not in $\Delta$, that make the type of the function subtype of
$\Empty\to\Any$ and the type of the argument subtype of the function's domain.
This search is collapsed in the definition of $t_1\apply\Delta t_2$.
Concretely, this operation finds two sets of substitutions $[\sigma_j]_{j\in J}$ and
$[\sigma_i]_{i \in I}$ such that (1) $\bigwedge_{j \in J} t_1\sigma_j\leq
\Empty\to\Any$ (this corresponds to the notation $t_1 \leqsub{\Delta}
\Empty\to\Any$ of \cref{def:c1}) and (2)~$\bigwedge_{i\in I}
t_2\sigma_i$ is a subtype of the domain of $\bigwedge_{j\in J} t_1\sigma_j$.
It then returns all the types of the result of the application of such two types.

The operator $\leqsub\Delta$ is also used to type
record operations, since it is used by the side conditions of the algorithmic
rules \rulename{Del}, \rulename{Sel}, and \rulename{Ext}: when applying
an operation on a record expression $e$ of type $t$, these side
conditions check whether there exists a set of substitutions making $t$
a subtype of, depending on the rule,  $\orec{\ell = \Any}$, $\orec{}$, or $\orec{\ell
= \Undef}$ and, if so, they apply the corresponding operation on the
instantiation of $t$. For instance, $\sel(t)$ looks for some
substitutions  $[\sigma_i]_{i \in I}$ such that $\bigwedge_{i \in I}
t\sigma_i\leq \orec{\ell = \Any}$ (i.e., a solution for $t
    \leqsub\Delta \orec{\ell = \Any}$) and returns its projection on $\ell$.

Finally, notice that, even if the typing rules~\rulename{Abs} and~\rulename{App} themselves are not new,
the process behind $t_1 \leqsub\Delta t_2$ and $t_1 \apply\Delta t_2$ are.
The novelty is that these operators now infer type substitutions that range not only over types,
but also over rows and field-types.
In the next section, we describe how to adapt the existing algorithms
to our framework.

As expected, the algorithmic type system is sound and complete with respect to
the declarative one, as stated by the following theorems (proofs are given in
\cref{app:typealgo}):\vspace*{-1mm}
\begin{restatable}[Soundness]{theorem}{soundnessalgo}
  \label{t:soundness_algo}%
  If $\algseq \Delta \Gamma e t$, then $\decseq \Delta \Gamma e t$.\vspace{-2mm}
\end{restatable}\vspace*{-.5mm}
\begin{restatable}[Completeness]{theorem}{completenessalgo}
  \label{t:completeness_algo}%
  If $\decseq \Delta \Gamma e t$, then there is $t'$ such that
  $\algseq \Delta \Gamma e t'$ and $t' \leqsub\Delta t$.\vspace{-.5mm}
\end{restatable}

\subsection{Tallying}
\label{sec:tallying}
\loic{Exemples enlevés de POPL: 4.1, 4.3, 4.4}

The algorithmic type system we have defined is parametric in
the decision procedure $s \leqsub\Delta t$,
which looks for an appropriate set of substitutions $[\sigma_i]_{i \in I}$ such that
$\bigwedge_{i \in I} s\sigma_i \leq t$.
While this has been tackled for type variables by \citet{castagna15},
here we need to extend the procedure to row and field variables
along the lines we outline below (for a full description, see the Appendix~\ref{app:tallying}).

Deciding $t_1 \leqsub\Delta t_2$ is done by testing the cardinality of the
set of substitutions $[\sigma_i]_{i \in I}$ we are looking for, by incremental
steps.
For $t_1 \apply\Delta t_2$, two sets of substitutions are sought for, and we can
follow a dove-tail order (more details are in
\cite[\S3.2.2-3.2.3]{castagna15}).

For each cardinality, we apply an instance of the \emph{tallying} algorithm.
Tallying is a unification problem acting on inequalities.
Given an initial set of subtyping constraints, tallying looks for a substitution
that satisfies these constraints.
For instance, deciding $t_1 \leqsub\Delta t_2$ is done first by trying to solve the
tallying problem for the constraint $(t_1, \leq, t_2)$, looking for a singleton
substitution set.
In case of a (non-fatal) failure, the next step is to look for a set of two
substitutions such that $[\sigma_1,\sigma_2] \Vdash t_1 \leqsub\Delta t_2$, which is
equivalent to solving the tallying problem for the constraint
$(t_1^1 \wedge t_1^2, \leq, t_2)$, where each $t_1^i$ is obtained from $t_1$ by
replacing all variables not in $\Delta$ by fresh ones.

Tallying returns a set of constraint-sets.
Indeed, the presence of subtyping, and in particular of the empty type,
implies that to solve a single set of constraints $C$ we may need to generate
several sets of constraints, yielding different solutions.
For instance, solving
$\{(\orec{\ell_1 = t_{\ell_1}, \ell_2 = t_{\ell_2}}, \leq,
\orec{\ell_1 = t'_{\ell_1}, \ell_2 = t'_{\ell_2}})\}$
generates three independent subproblems:
$\{(t_{\ell_1}, \leq, \Empty)\}$, $\{(t_{\ell_2}, \leq, \Empty)\}$ and
$\{(t_{\ell_1}, \leq, t'_{\ell_1}), (t_{\ell_2}, \leq, t'_{\ell_2})\}$.

We define the solving procedure for the type tallying of a
constraint-set as an extension of the existing one for type variables.
The procedure follows the steps that are given by~\cite{castagna15} (plus
an additional one), namely:
$(\oldstylenums 1)$ normalization, 
$(\oldstylenums 2)$ merging and saturation, 
$(\oldstylenums 3)$ harmonization (which is new and specific to row variables), 
$(\oldstylenums 4)$ transformation of constraints into equations, 
and $(\oldstylenums 5)$~creation of the substitution solutions. 
For space reasons, in what follows we only outline these steps.
For formal definitions, refer to the Appendixes~\ref{sec:normalization}
to~\ref{sec:solution}.

A constraint is defined as a triple $(\tterm_1,c,\tterm_2)$ ($c \in
\{\leq,\geq\}$) such that $\tterm_1$ and $\tterm_2$ are two terms of
the same kind (type, field-type, or rows of the same domain).
An arbitrary order on variables is given as an input to the algorithm in order
to guarantee well-formedness of the solution.

\paragraph{Normalization}
The first step of tallying decomposes the initial constraints on types
into a set $\constrset$ of \emph{normalized} constraint-sets. These
are constraint-sets where all constraints are of the shape
$(\vterm, c, \tterm)$, where $\tterm$ stands for a type, field-type or row and
$\vterm$ for a variable of the kind of $\tterm$.

Normalization works by decomposing types into elementary constraints.
For records and rows, it might be necessary to decompose the row variables across
different fields.
For instance, the record $\rec{}{\rho_0}$ might need to be decomposed
over, say, labels
$\ell_1$ and $\ell_2$.
In that case, we will spread the row variable $\rho_0$ into
$\rec{\ell_1=\rho_0.\ell_1,\ell_2=\rho_0.\ell_2}{\cutvar{\rho_0}
{\{\ell_1,\ell_2\}}}$, introducing $\rho.\ell_1$, $\rho.\ell_2$ and
$\cutvar{\rho_0}{\{\ell_1,\ell_2\}}$ as field and row variables with special names.
Then, the tallying algorithm may give constraints over $\rho_0.\ell_1$,
$\rho_0.\ell_2$ and $\cutvar {\rho_0}{\{\ell_1,\ell_2\}}$,
and we expect a solution of the shape $\sigma(\rho_0) = \row{\ell_1 =
\tau_1, \ell_2 = \tau_2}{\rw}{\emptyset}$, where each $\tau_i$ has been obtained
from the constraints on $\rho_0.\ell_i$, and $\rw$ from the ones on
$\cutvar{\rho_0}{\{\ell_1,\ell_2\}}$.

In the literature where rows are all atomic,
unification of a row variable with another row is performed component-wise.
Considering each field independently, while it works well with atomic rows, is
unsound when working with Boolean combination of rows, as shown by this example:
\begin{example}
  \label{ex:componentwise_unsound}%
  Let $\Keyw{result} =
  \crow{\Keyw{log} = \String, \Keyw{succ} = \True, \Keyw{val} = \Any}{\emptyset} \vee
  \crow{\Keyw{log} = \String, \Keyw{succ} = \False, \Keyw{val} = \Undef}{\emptyset}$.
  For the constraint
  $(\row{\Keyw{val} = \Any \lor \Undef}{\rho}{\emptyset}, \leq, \Keyw{result})$,
  independent component-wise tallying gives
  $\sigma(\rho) = \crow{\Keyw{log} = \String, \Keyw{succ} =
  \Bool}{\{\Keyw{val}\}}$ as a solution, which does not verify the
  constraint (the type obtained as ``solution'' contains record values in
  which \Keyw{succ} is \True{} and \Keyw{val} is undefined, which are
  not included in \Keyw{result}).
  Indeed, such a procedure effectively solves tallying for
    $(\row{\Keyw{val} = \Any \lor \Undef}{\rho}{\emptyset}, \leq,
    \crow{\Keyw{log} = \String \lor \String, \Keyw{succ} = \True \lor \False,
  \Keyw{val} = \Any \lor \Undef}{\emptyset})$ instead.\hfill\qed
\end{example}

To obtain a sound decomposition of rows, we can adapt the formula underlying our
subtyping algorithm (this formula is found in the statement of \cref{l:subtyping}).
Given a constraint $C$ on DNFs of rows, we still consider the set $L$ of all top-level
labels in the DNFs (in the last example, $L = \{\Keyw{log}, \Keyw{succ},
\Keyw{val}\}$), for which we will end up with independent constraints and
constraints over the cofinine domain of the rows deducted from $L$.
Doing so, we find no solution for the constraint in
\cref{ex:componentwise_unsound}, which is correct, since indeed the constraint cannot be satisfied.

Although this method yields a correct set of solutions, this set is
far from being complete: the more fields we decompose over, the more solutions
we might loose.
Yet, the previous method decomposes records and rows over the whole set of
top-level labels.
Our solution is to decompose rows over a set of labels as small as possible.
This technique is based on a general decomposition formula given in the
statement of \cref{l:decomposition}.
Now, while normalization is sound,
the algorithm is still not complete, due to the potentially necessary
decomposition of some row variables.
\begin{example}[Incompleteness]
  \label{ex:incomplete}%
  Consider the constraint-set
  $
  \{(\row{\Keyw{log} = \String}{\rho_2}{\emptyset} \to
    \row{\Keyw{log} = \String}{\rho_2}{\emptyset}, \geq,
    \Keyw{result} \to \Keyw{result}),
    (\row{\Keyw{log} = \String}{\rho_2}{\emptyset} \to
    \orow{}{\emptyset}, \geq,
    \row{\Keyw{succ} = \Any \lor \Undef, \Keyw{val} = \Any \lor
    \Undef}{\rho_1}{\emptyset} \to
    \row{\Keyw{succ} = \Any \lor \Undef, \Keyw{val} = \Any \lor
  \Undef}{\rho_1}{\emptyset})\}
  $.
  Let $\rho_1$ be of smaller order than $\rho_2$.
  By a decomposition over $L = \{\Keyw{log}\}$ in the first constraint
  and $L = \{\Keyw{succ}, \Keyw{val}\}$ in the second constraint,
  we derive the constraint-set (omitting trivial constraints)
  $\{(\String, \leq, \rho_1.\Keyw{log}),
    (\cutvar{\rho_2}{\{\Keyw{succ},\Keyw{val}\}}, \leq,
    \cutvar{\rho_1}{\{\Keyw{log}\}}),
    (\rho_2, \leq, \Keyw{result}),
  (\rho_2, \geq, \Keyw{result})\}$.

  As we describe below (cf.\ \emph{Harmonization}), a further step of the tallying algorithm harmonizes
  the decomposition of the row variables across all constraints.
  In particular, it decomposes $\rho_2$ over $\{\Keyw{succ}, \Keyw{val}\}$ in
  the constraints $(\rho_2,\leq,\Keyw{result})$ and
  $(\rho_2,\geq,\Keyw{result})$.
  Since when decomposing $\rho_2$ in this way, no solution might apply
  to both of these constraints ($\rho_2$ needs to be instantiated exactly to the
  union type $\Keyw{result}$), tallying fails.
  The solution mapping $\rho_2$ to $\Keyw{result}$ is not found.\hfill\qed
\end{example}

\paragraph{Constraint merging and saturation}
After normalization, a constraint set may have for the same variable $\vterm$
different constraints of the form $\tterm_1 \leq \vterm, \dots, \tterm_n \leq
\vterm$ and $\vterm \leq \tterm'_1, \dots \vterm \leq \tterm'_m$:
we replace them by two constraints $\bigvee_{1 \leq i \leq n}\tterm_i\leq \vterm$ and
$\vterm\leq\bigwedge_{1 \leq j \leq m}\tterm_j'$,  then we
add  $\bigvee_{1 \leq i \leq n}\tterm_i\leq \bigwedge_{1 \leq j \leq
m}\tterm_j'$ to (i.e., saturate) the constraint set, and we normalize again.

\paragraph{Harmonization}
Take a constraint-set $C$ with $(\cutvar \rho {L_1}, c, \rw) \in C$,
$L_1 \subsetneq L$ and $L$ the set of labels appearing in the names of variables
$\rho'.\ell$ and $\cutvar {\rho'} {L_2}$ in $C$.
As mentioned in \cref{ex:incomplete}, harmonization of constraint-set rewrites
the constraint $(\cutvar \rho {L_1}, c, \rw)$ into $(\cutvar \rho
L, c, \rw)$ and feeds this constraint to normalization again.
Harmonization ends with a constraint-set of homogeneous domain, where for each
initial row variable $\rho$ all occurrences of $\cutvar \rho L$ are defined on
(i.e., harmonized to) the same $L$.\footnote{\label{fn:multiplicative}\new{
    This step introduces a multiplicative factor with respect to the complexity of the version of tallying without row polymorphism.
    This factor is expected to be low in practice, since it is bounded by the
    number of labels a programmer uses in a single type and by the number of
constraints in which a variable appears.}}

\paragraph{Equations generation}
At this point, for each type and field variable $\vterm$ there are two unique
constraints of the form $\tterm\leq \vterm$ and $\vterm \leq \tterm'$, that we transform into the equation
$\vterm=(\tterm\vee \vterm')\wedge \tterm'$ with $\vterm'$ fresh.
For a row variable $\rho$, there is a set of labels $L$ and constraints of the
form $\tau_1 \leq \rho.\ell \leq \tau_2$ for each $\ell \in L$, and a constraint
$\rw_1 \leq \cutvar \rho L \leq \rw_2$.
We define the terms
$\tau_\ell = (\tau_1 \vee \fvar_\rho^\ell) \wedge \tau_2$ and
$\rw = (\rw_1 \vee \rho') \wedge \rw_2$, with fresh variables $\fvar_\rho^\ell$
and $\rho'$.
We finally create an equation $\rho = \row{(\ell = \tau_\ell)_{\ell \in
L}}{\rw}{\Labels{\setminus}\rdef(\rho)}$.

\paragraph{Solution}
We solve the equations following the given order on the variables on the
left-hand side, by collecting the equations $\vterm = \tterm$ and replacing in
all other equations $\vterm$ by $\mu\vterm'.(\tterm\esubs{\vterm}{\vterm'})$
with $\vterm'$ fresh.
Thanks to the order, the type $\mu\vterm'.(\tterm\esubs{\vterm}{\vterm'})$
is contractive and, thus, well-formed.

A solution to a constraint-set $C$ is a substitution $\sigma$ such that
$\forall(\tterm_1,\leq,\tterm_2){\in} C\,.\, \tterm_1\sigma \leq \tterm_2\sigma$ and
$\forall(\tterm_1,\geq,\tterm_2) {\in} C\,.\, \tterm_1\sigma \geq \tterm_2\sigma$ hold.
If $\sigma$ is a solution to $C$, we write $\sigma \Vdash C$.

\subsubsection{Properties of the algorithm}

We call $\Sol C {}$ the solving procedure for the type tallying of $C$.
We write $\Sol C \Theta$ if $\Sol C {}$ terminates yielding the set of
substitutions $\Theta$, called the \emph{solution} of the type tallying problem for $C$.
\vspace{-.8mm}
\begin{restatable}[Soundness]{theorem}{soundnesstallying}
  \label{l:soundness-tallying}%
  Let $C$ be a constraint-set.
  If $\Sol C \Theta$, then for all $\sigma \in \Theta$,
  $\sigma \Vdash C$.
\end{restatable}
\vspace{-2.6mm}
\begin{restatable}[Termination]{theorem}{terminationtallying}
  \label{l:termination-tallying}%
  Let $C$ be a constraint-set.
  Then $\Sol C {}$ terminates.
\end{restatable}
\vspace{-2.6mm}
\begin{restatable}{proposition}{finitetallying}
  \label{pr:wellformedness}
  Let $C$ be a constraint-set and $\Sol C \Theta$. Then
  (1) $\Theta$ is finite and
  (2) for all $\sigma \in \Theta$ and for all $\vterm \in \dom(\sigma)$,
  the types in $\sigma(\vterm)$ are contractive.
\end{restatable}


\section{Related Work}
\label{sec:related}
\paragraph{Row polymorphism}
Our formalization of records as quasi-constant total functions and the inclusion
of row polymorphism are directly inspired from the formalism of \citet{remy89,remy94}.
Remy's  work contains neither set-theoretic types nor subtyping, and
therefore commutation of fields is obtained by structural equations.
Unlike our case, the types in \cite{remy94} are not recursive.
An extension of this system to recursive types is given by \cite{remy98} when
describing a type system for objects.

Before Rémy, \citet{wand87} introduced row polymorphism to type object
inheritance.
His type inference algorithm, corrected in \cite{wand91}, considers free
record extension (i.e., right priority record concatenation) but lacks principal solutions.
Instead, it deduces a finite set of solutions of which all types of the term
under consideration are instances, as we also do.
However, unlike us, Wand's  type grammar lacks intersection types, so
it is not possible to merge the multiple
solutions into a single type, as we instead do.
In an earlier attempt of our work we considered free record extension, as it is present both in CDuce and in its generalization by~\citet{castagna23a}, but this made the theory much
more involved, in particular tallying.
It would be worthwhile to study this possibility again, now that the theory
with strict extension is precisely laid down. \new{However, we anticipate this study to be quite complex: typing the concatenation of two polymorphic records presents significant challenges. It will likely require enriching the type algebra to explicitly represent row concatenations (row variables included), as the type operators used in CDuce for typing record concatenation are insufficient for this purpose.}

(Syntactic) subtyping is present in the work of \citet{cardelli91}.
Operations on records (deletion, selection, extension) are  directly defined
within the syntax of types.
In our system instead, we compute these operations on the types during typing.
It is thus currently impossible to postpone the extension or
deletion of a field with label $\ell$ that is affected by a row variable, until
the point where the row variable will be instantiated and, in that case, we must resort to an approximation.
Cardelli and Mitchell must however define syntactic equivalence relations on
operators and fields, and their system lacks principal typing, as well.

Row polymorphism and extensible records are implemented using predicate on
types by \citet{harper91} and, later, by \citet{gaster96} under the name qualified types.
In the latter, positive information is given in the type and negative
information (absent field) in the predicate.
\citet{morris19} use qualified types with uninterpreted predicates for
concatenation and membership of fields.
Their system can be instantiated to most of the standard approaches of the literature.
We aimed to minimize changes to the type syntax of CDuce and Elixir, opting to refrain from incorporating qualified types.
\new{On the theory side, we wished to keep the implementation of records as total
functions---as already present in the semantic subtyping framework---while in
systems with qualified types, records are constructed from partial functions.}

A convenient way to define extensible records is with scoped labels
\cite{leijen05}, where labels may appear several times in a record, the most
recent occurrence shadowing earlier ones.
\citet{paszke23} recently extended the formalism to deal with first-class labels,
first-class rows and concatenation.
This gives a simple formulation of types, albeit too syntactic for our
semantic approach to typing and subtyping, and our desire to interpret records as
quasi-constant functions.
\new{Using scoped labels would impose more structure on the interpretation of
records, and therefore could weaken the intuitive understanding of set-like
operations and types as sets.}

Variants, the dual of records, are studied in a semantic subtyping setting by
\citet{castagna16}, where they show that adopting full-fledged union (and
intersection and negation) types as well as let-polymorphism gives a
more intuitive and expressive theory of polymorphic variants.
However, as we showed in \cref{sec:need}, our semantic interpretation seems
to break this duality, insofar as these same features are not enough to encode record
polymorphism too.

\paragraph{Presence polymorphism}

\citet{remy94} describes records types with fields that can be either present or absent,
as indicated by an additional annotation.
He shows how to add presence polymorphism over these annotations, yielding
records parametric in the absence or presence of some
fields.
\citet{garrigue15} proposes a weaker system,
where constraints apply to single variables (rather than being distributed over row
and presence variables), and where absent fields are determined
by intersecting type constraints. \citet{garrigue15} justifies this choice by the fact that it
yields types that are simpler to understand, since they do not require
different variables of different sorts. This simplicity is the reason why it is
the system used for OCaml. Our work shares
Garrigue's simplicity, since only row variables are
visible to the programmer and, under the hood, field-type variables---which range over types augmented with $\Undef$---provide
a natural notion of presence polymorphism. Field-type variables can be instantiated with
$\Undef$ for an absent field, with a type $t$ for a mandatory field, but also with
$t \lor \Undef$ for an optional field of type $t$.
To our knowledge, our work is the first to allow presence polymorphism over
optional fields: in the existing literature, \new{instantiated} field-types of polymorphic records
can just be either present or absent, but not both.
Moreover, integrating presence polymorphism in our framework is
straightforward since field-type variables are almost  handled as
ordinary type
variables: it is the kinding system that enables the extra $\Undef$
possibility (for an example-based comparison between presence
polymorphism of \cite{remy94} and ours, see \cref{sec:app-presence}).

\paragraph{Relation between different kinds of polymorphism}

Our type system features three kinds of parametric polymorphism: on types, rows, and field-types.
It also features subtype polymorphism and ad-hoc polymorphism via
intersection and union types.  The question naturally arises: how are these concepts interconnected, and where do they overlap?
It is folklore that unrestricted intersections combined with parametric
polymorphism encode a form of bounded quantification, as described in the
introduction (see also~\cite[Section 2]{castagna24}).
\citet{tang23} formally compare the expressiveness of row and presence
polymorphism to \emph{structural} subtyping, for calculi with records and
variants.
More precisely, they encode diverse subtyping using relevant polymorphic
systems.
\new{Whether their techniques can be adapted to gather insights into semantic
subtyping, or polymorphism in the semantic subtyping framework, remains an open question.}

\citet{xie20} show that row and bounded polymorphism can be encoded with
disjoint polymorphism, obtained by adding parametric polymorphism to a system
with \emph{disjoint} intersections, and having a disjointedness predicate in the
quantification of types.
Contrary to our intersections, that are uninhabited if applied to separate types,
disjoint intersections type a merge operator that generalizes the disjoint
concatenation of extensible data types, like records, to arbitrary types.
The deletion operator of extensible data types can also be generalized to
arbitrary types using disjoint polymorphism and a merge operator. This operator
once again differs from our set exclusion operation $t_1 \setminus
t_2 = t_1 \wedge \neg t_2$.

\citet{parreaux22} present MLstruct, an extension with set-theoretic types of
MLsub~\cite{dolan17}, based on \emph{algebraic subtyping}.
\new{Their approach differs from ours, as it provides less expressive subtyping and type connectives.
For instance, intersections cannot type general overloaded functions, and a
union of records in MLstruct is equivalent to the single record with the union
applied component-wise.
Keeping unions of records distinct is a crucial feature of our system (cf., the code in lines
\ref{23} to \ref{29} in the introduction) and a major difficulty for the
tallying problem.
The gain from this reduced expressivity is the existence of principal types (a property absent in semantic subtyping even for simple monomorphic types),
the ability to encode record extension and deletion without the need
for row polymorphism,
and maintaining the duality between variants and records.}

\new%
{\paragraph{Effect systems} Record calculi and row polymorphism had a significant impact on
effect systems.  As suggested by an anonymous reviewer, our work may usefully apply to that setting. For instance, in the language Koka~\cite{leijen14}, the introduction of row polymorphism for effects is justified by the fact that inference cannot solve constraints for unions. This yields to approximations as the ones we pointed out in the last examples of Section~\ref{sec:need}. But in our system we are able to keep separate unions while having row polymorphism. This should in principle allow the system to track effects more finely than in~\cite{leijen14}. For instance, consider this example directly derived from~\cite{leijen14}{\small
\begin{alltt}
   \ function foo(f,g) \{ if (random()) then \{ f(); error("hi") \} else \{ g(); print("hi") \} \}
\end{alltt}}
If  $t_1\to \varepsilon\, t_2$ is the type of functions that can be applied to an input $t_1$ and produce an output of type $t_2$ and effect $\varepsilon$, then using our row polymorphism we should be able to give the example above the type $\forall \epsilon_1\epsilon_2.(()\to \langle~|\epsilon_1\rangle\,(), ()\to \langle~|\epsilon_2\rangle\,()) \to ((\langle \text{rnd}, \text{exn}|\epsilon_1\rangle\cup\langle \text{rnd}, \text{io}|\epsilon_2\rangle)\;())$, which states that \texttt{foo} can either produce the \texttt{exn} effect (by erroring) or the \texttt{io} effect (by printing) but not both. The system for Koka would not be able to distinguish the two effects.
This suggests that exploring the use of our system in effect systems could be a fruitful direction for future work.}

\paragraph{Practice}

Efficient compilation of polymorphic or extensible records
has been widely explored by researchers such as \citet{gaster96,ohori95}.
These works advocate moving away from Rémy's formalism.
Ohori's calculus, in particular, stores information in elaborated kinds.
It was expanded to extensible records by \citet{alves21}, but with no
mention of the compilation method.
Yet, \citet{hillerstroem16} provide a compelling abstract machine for a calculus
employing this formalism, serving as a foundation for their language Links.
Regarding general-purpose languages, several of them propose either
a flavor of set-theoretic types (like Typescript \cite{typescript} or Flow
\cite{facebookflow}), sometimes based on the theory of semantic subtyping
(CDuce, Luau~\cite{luausemsub}, Elixir~\cite{elixirtypes}, or Ballerina~\cite{ballerina-typing}), or polymorphic extensible records (like
Purescript \cite{freeman13} or OCaml) but, to our knowledge, none of them
offers both of these features, as we propose in this work.

\section{Conclusion}
\label{sec:conclusion}

We designed a type system featuring set-theoretic types and semantic
subtyping for record calculi with row and presence polymorphism.
We instantiated this type system on a specific \textlambda-calculus incorporating
record selection, extension, and deletion, and devised a unification (\textit{tallying}) algorithm. 
We believe there are at least five key takeaways from our work:
\begin{enumerate}
    \item \emph{Intersection types and parametric polymorphism fall short}: while their combination is expressive, it cannot fully address practical examples that row polymorphism can handle effectively.  
 
    \item \emph{Surprising similarity in subtyping algorithms}: the subtyping algorithms for monomorphic and polymorphic records are strikingly and disturbingly similar. This similarity is ``disturbing'' because of the extensive effort we had to put in to finally prove that the polymorphic case is obtained from the monomorphic case by adding just a single non-recursive case that tests whether a row variable belongs to a finite set.  

    \item \emph{Comparable type inference}: from an algorithmic and complexity perspective, type inference for row variable polymorphism closely parallels that for parametric polymorphism.  

    \item \emph{Distinct semantics}: records represent total functions on the set of all labels, while rows (and consequently row variables) represent partial functions. This subtle distinction introduces significant challenges and subtleties in the formal treatment.  

    \item \emph{Set-theoretic types disrupt type dualities}: incorporating
        set-theoretic types and semantic subtyping undermines the traditional duality between variant and record types.
\end{enumerate}
Our next goal is to implement our results in the CDuce compiler, thereby enhancing the language with row polymorphism. We closely adhered to the
theory used for CDuce, in that
record types were simply extended with row variables and presence polymorphism,
leveraging the union connective and an existing constant for undefinedness.
Our model and algorithms, specifically for subtyping and tallying, naturally extend
the existing ones, therefore their implementation will not affect the performance of the type-checker.
While we did not address the problem of type
reconstruction, \citet{castagna24b} provide the theory and an
implementation of type reconstruction for an ML-like language that uses
CDuce's tallying library as a black box. We are confident that 
integrating our tallying algorithm into that system and adding records to
the language should yield, with some extra effort, a 
reconstruction system with polymorphic records.

CDuce is, of course, more complex than our record calculus.
Likewise, any type system for a dynamic language needs to
account for features like pattern matching, type cases, guards, or type narrowing.
At first sight, these
features seem mostly orthogonal to the introduction of row polymorphism.
Our hope is to be able to integrate the latter seamlessly into any existing
set-theoretic type system with semantic subtyping, such as Elixir
\cite{castagna24a}, Ballerina \cite{ballerina-typing}, Luau \cite{luausemsub}, or Erlang
\cite{schimpf22}. We are closely monitoring the ongoing efforts to port
CDuce's tallying algorithm into the Elixir compiler.

Our record calculus lacks first-class labels.
\citet{castagna23a} defines a system in which records can be
used both as ``structs'' (i.e., records without first-class labels) and as
dictionaries/maps with first-class labels.  His work is carried out in the
same setting as ours, set-theoretic types and
records as quasi-constant functions, but without row polymorphism. However, some
solutions proposed by \citet{castagna23a} are explicitly motivated by
having a system that could be easily extended with row and presence
polymorphism, which is why we believe that merging the two systems
should not pose any fundamental issue.
The actual expressiveness of a system with first-class labels where operations on
records are not part of the syntax of types remains to be investigated.


\section*{Data Availability Statement}
Additional definitions, proofs, and extensions omitted from the main text can be found in the appendix. This appendix is available online both in the ACM Digital Library and in the Arxiv repository at \url{https://arxiv.org/abs/2404.00338}.

\begin{acks}
  This work benefitted from the valuable feedback of individuals working on
  Elixir and semantic subtyping: Guillaume Duboc, Kim Nguyen, and José Valim. In
  particular, José directed us to the \elix{logger} module, which we used for
  our examples in Section~\ref{sec:example}. We also thank the anonymous
  reviewers for their insightful comments and suggestions, especially one
  reviewer who highlighted the connection with effect systems discussed in
  Section~\ref{sec:related}. This work was partially supported by a CIFRE grant
  \grantnum{}{n°2021/0420} to develop the type system of Elixir and by the
  \grantsponsor{GA:101039196}{Horizon Europe ERC Starting Grant
  CRETE}{https://cordis.europa.eu}
  \grantnum{GA:101039196}{(GA:101039196)}.
\end{acks}
\nocite{FCB02}
\ifbibtex
\bibliography{main}
\else
\printbibliography
\fi
\clearpage
\appendix

\section{APPENDIX FOR TYPES}
\label{app:types}
\subsection{Kinding System}
The kinding rules for types are given in \cref{fig:kinding}.
\begin{figure}[h]
  \begin{mathpar}
    \inferrule*{ }{\Empty:\kappa}
    \and \inferrule*{t:\kappa}{\neg t:\kappa}
    \and \inferrule*{t_1:\kappa \and t_2:\kappa}{t_1 \lor t_2:\kappa}
    \and \inferrule*{ }{\alpha:\ktype}
    \and \inferrule*{ }{b:\ktype}
    \and \inferrule*{t_1:\ktype \and t_2:\ktype}{t_1 \to t_2:\ktype}
    \\
    \inferrule*
    {\row{(\ell = \tau_\ell)_{\ell \in L}}{\tl}{\emptyset}:\krow{\emptyset}}
    {\rec{(\ell = \tau_\ell)_{\ell \in L}}\tl:\ktype}
    \and \inferrule*[right={$\begin{array}{l}
        \tl \in \{\crecsign,\orecsign\}\\
        L_1 \cap L_2 = \emptyset
    \end{array}$}]
    {\forall \ell \in L_1. \tau_\ell:\kfield}
    {\row{(\ell = \tau_\ell)_{\ell \in L_1}}{\tl}{L_2}:\krow{L_2}}
    \and \inferrule*[right={$\begin{array}{l}
        L_1 \cap L_2 =\emptyset\\
        \rdef(\rho) = \Labels{\setminus}(L_1 \cup L_2)
    \end{array}$}]
    {\forall \ell \in L_1. \tau_\ell:\kfield}
    {\prow{(\ell = \tau_\ell)_{\ell \in L_1}}{\rho}{L_2}:\krow{L_2}}
    \and \inferrule*{ }{\fvar:\kfield}
    \and \inferrule*{ }{\Undef:\kfield}
    \and \inferrule*{t:\ktype}{t:\kfield}
  \end{mathpar}
  \caption{Kinding rules}
  \label{fig:kinding}
\end{figure}

\subsection{Example of a Presence Polymorphic Type}
\label{sec:app-presence}

Presence polymorphism has been introduced by \citet{remy94} to let a field be
polymorphic in its presence.
A presence variable $\theta$ can be instantiated to one of
$\{\texttt{abs},\texttt{pre}\}$.
Our calculus also supports presence polymorphism thanks to field variables.
Let us reproduce an example of a presence polymorphic type declaration from
\cite{remy94}, that we will then transcribe to our setting.
\begin{minted}{ocaml}
type tree(!$\theta$!) = Leaf of Int | Node of {left:pre.tree(!$\theta$!), right:pre.tree(!$\theta$!), annot:!$\theta$!.int}
\end{minted}
Instantiating $\theta$ to \texttt{abs} gives the type of trees with no
annotations on the nodes,
while instantiating it with \texttt{pre} gives the type of trees with integers
annotated on the nodes.

But there are several problems: adding yet another kind of polymorphism 
adds a bit of work, and most of all, even absent fields must have a type
attached.
In \cite{remy94}, this can cause losing unification of two semantically
equivalent records, if for instance one has a field $\ell$ of type
\texttt{abs.int}, and the other a field $\ell$ of type \texttt{abs.bool}.
To avoid this problem, \citet{tang23} requires every record value to be
annotated with its type, so that they can forget about absent types, and also so
that they can have deterministic typing.
So for instance the superfluous type
$\{\ell_1=M;\ell_2=N\}^{\{\ell_1:\texttt{pre}.A;\ell_2:\texttt{abs}.B\}}$ is equivalent to
$\{\ell_1=M\}^{\{\ell_1:\texttt{pre}.A\}}$.

In CDuce, there is no presence polymorphism at all, and the tree type defined
previously cannot be expressed without code redundancy.
In our formalism, presence polymorphism is simply obtained by means of field
variables, and adding support for it in CDuce would make it possible to write
the following (note that we use unions instead of tagged unions):
\begin{minted}{ocaml}
type tree(!$\theta$!) = Int | {left:tree(!$\theta$!), right:tree(!$\theta$!), annot:!$\theta$!}
\end{minted}
Fields do not have any superfluous type information when they are absent, so we
do not have the redundant record expressions, and do not need to put type
annotations directly on the values.
Moreover, in CDuce and in our type system, fields can be optional.
This is achieved thanks to union types (an optional field is encoded as $t \vee \Undef$).
We do not know of any presence polymorphic type system dealing with optional
types.

We could instantiate the type of the example to any field-type, such as $\Int$ for
a tree with all nodes annotated by an integer, $\Int \vee \Undef$ for a tree
with some nodes annotated by an integer, $\Undef$ for a tree with no
annotation, and even for instance to a type variable $\alpha$, to create
a type parametric in the type of its annotation, but where the annotation
is mandatory on every node.
In Rémy's system, having polymorphism also on the type of the
annotations would require having two parameter variables (for presence and for type).

\subsection{Models}
\label{sec:app-models}

In this section we give the detailed technical development to define
our subtyping relation.

To interpret record values we follow~\citet{frisch04} and represent
a record value by a quasi-constant function that maps labels into
either values
(i.e., the elements of $\Domain$) or $\Undef$. Quasi-constant
functions are total functions that map all but a finite set of elements of their
domain into the same value (called default value). Thus record values
can be represented by quasi-constant functions whose default value is
$\Undef$ (see~\citet{castagna23a} for a more detailed explanation). Formally, let us write  $\Domain_\Undef$ for
$\Domain\cup\{\Undef\}$  where $\Undef$ is a distinguished element not in
$\Domain$. We represent our record values as quasi-constant functions
from $\Labels$ to $\Domain_\Undef$ and, thus, interpret record types
as sets of these functions.

The formal definition of quasi-constant function has been given in
\cref{def:qcf} and is repeated below for convenience.
\quasiconstant*
Although this
notation is not univocal (unless we require $z_i\not=z$ and the $\ell_i$'s to be pairwise distinct), this is largely sufficient for the
purposes of this work. If $(Z_\ell)_{\ell\in\Labels}$ is a family of
subsets of $Z$ indexed by \Labels, we denote by
$\prodc{\ell{\in}\Labels}Z_\ell$ the subset of
$\Labels\qcfun Z$ formed by all quasi-constant functions $r$ such that
$r(\ell)\in Z_\ell$ for all $\ell\in\Labels$ (intuitively,
$\prodc{\ell{\in}\Labels}Z_\ell$ is a ``type'' of quasi-constant functions).

Next we have to give an interpretation for the
variables. \citet{castagna11} tell us that type variables must be
interpreted as sets in the domain $\Domain$. Therefore, an
interpretation for type variables is a function in
$\typeVars \to \Pd(\Domain)$. Field variables are not much harder,
since the only difference with type variables is that their
interpretation can contain $\Undef$, and therefore it is a function in
$\fldVars \to \Pd(\Domain_\Undef)$. More difficult is the
interpretation of row variables, since these are mapped
into \emph{rows}, that is, \emph{partial} quasi-constant functions on
$\Labels$. Let us write 
$\Labels \pqcfun \Domain_\Undef$ for the partial quasi-constant
functions from $\Labels$ to $\Domain_\Undef$. Thus, an interpretation
of row variables must map an element of $\rowVars$ into a set of functions in $\Labels \pqcfun \Domain_\Undef$. However, for a
given $\rho$ we cannot consider any element in
$\Pd(\Labels \pqcfun \Domain_\Undef)$: we need that the functions in
the interpretation of $\rho$ are total on $\rdef(\rho)$. Formally, we
have:
\begin{definition}[Well-Kinded Interpretation]
Let $\eta$ be a function in
$\rowVars \to \Pd(\Labels \pqcfun \Domain_\Undef)$. We say that $\eta$
is \emph{well kinded} if for every $\rho\in\rowVars$ and for every
$f\in\eta(\rho)$, $f$ is a (total) quasi-constant function in
$\rdef(\rho)\qcfun \Domain_\Undef$. We denote by
$\rowVars \wkfun \Pd(\Labels \pqcfun \Domain_\Undef)$ the set of
well-kinded functions.
\end{definition}
In conclusion our interpretation of types will be parametric in an assignment $\eta$ for the variables, which will be a function in  \[
  \etaInter \eqdef (\typeVars \to \Pd(\Domain))
  \cup (\fldVars \to \Pd(\Domain_\Undef))
  \cup (\rowVars \wkfun \Pd(\Labels \pqcfun \Domain_\Undef))
\]

The next step is to define the domain $\Domain$ in which to give
the interpretation of types. This is quite simple for us since it suffices to take the
model defined by~\citet{castagna11} and replace products by
quasi-constant functions. The hard problem for defining this model, and
thus the interpretation of types, is to give an interpretation of the
function spaces, but this problem was solved by~\citet{frisch08} whose
solution is 
reused by~\citet{castagna11}. In a nutshell we want to define an
interpretation function
$\TypeInter .: \Types\to\etaInter\to\Pd(\Domain)$. Since the elements of
$\Domain$ represent the values of the language, then $\Domain$ must
contain the set $\Constants$ of constants of the language, the quasi-constant functions (to represent
record values), and the functions
from $\Domain$ to $\Domain$, but the last containment is impossible for cardinality
reasons. The solution by~\cite{frisch08} is to associate to every
domain $\Domain$ and function  $\TypeInter .: \Types\to\etaInter\to\Pd(\Domain)$ a unique
\emph{extensional interpretation}
$\ExtInter(\cdot): \Types \to \etaInter \to \Pd(\ExtInter_\Domain)$
which fixes the semantic of the type constructors, and then to accept as a valid
interpretation of the types only the pairs  $(\TypeInter .,\Domain)$
such that for all $\eta$, $\TypeInter{t}_\eta=\emptyset\iff\ExtInter(t)_\eta=\emptyset$. 

We invite the reader to refer to \citet[Section 2.2]{castagna11} for a
more detailed explanation of how the extensional interpretation
works and to \citet{frisch08} for full details. Henceforth, we just present
how to extend the extensional interpretation of \citet{castagna11} to
include quasi-constant functions and the interpretation of
row-variables, and pinpoint the differences between the two definitions.
We suppose to be given an interpretation
$\ConstantsInBasicType:\Basics\to\Pd(\Constants)$ of basic types into
sets of constants. Given a set $S$ we use the notation $\overline S$
to denote its complement in an appropriate universe: this notation is in particular used for the
set
\raisebox{-.5pt}{$\Pd(\overline{\TypeInter{t_1}\times\overline{\TypeInter{t_2}}})$}
 which corresponds to interpreting the elements in $t_1\to t_2$ as
binary relations, namely as elements of the set $\{\ f\subseteq\Domain^2\mid \text{for all }(d_1,d_2){\in}f , \text{ if } d_1{\in}\TypeInter{t_1}_\eta\text{ then } d_2{\in}\TypeInter{t_2}_\eta\ \}$.

\begin{definition}[Extensional interpretation]
  \label{d:extint}
  Let $\Domain$ be a set.  The \emph{extensional domain}
  of $\Domain$ is defined as:
  \(
    \ExtInter_\Domain
    = \Constants + \Domain + \Pd(\Domain \times \Domain_\Omega)
    + (\Labels \qcfun \Domain_\Undef)
  \)
  where $\Omega$ and $\Undef$ are two different distinguished elements not in $\Domain$.

  Let $
  {\TypeInter{\cdot}}: \Types \to \etaInter \to \Pd(\Domain)$
  be an interpretation of types parametric in a well-kinded
  interpretation of variables.
  The associated \emph{extensional interpretation} of types is the unique function
  $\ExtInter(\cdot): \Types \to \etaInter \to \Pd(\ExtInter_\Domain)$
  such that:

\noindent
   \(\begin{array}{rl}
    \ExtInter(\Empty)_\eta
    &= \emptyset \\   
    \ExtInter(\alpha)_\eta
    &= \eta(\alpha) \\
    \ExtInter(b)_\eta
    &= \ConstantsInBasicType(b)\\
    \ExtInter(\neg t)_\eta
    &= \ExtInter_\Domain \setminus \ExtInter(t)_\eta\\
    \ExtInter(t_1 \lor t_2)_\eta
    &= \ExtInter(t_1)_\eta \cup \ExtInter(t_2)_\eta\\
    \ExtInter(t_1 \to t_2)_\eta
    &= \raisebox{-.5pt}{$\Pd(\overline{\parbox[][10.7pt][c]{48pt}{$\TypeInter{t_1}_\eta\times\overline{\parbox[][8.8pt][c]{18pt}{$\TypeInter{t_2}_\eta$}}$}}\,)$}
    \end{array}
    \begin{array}{rl}
    \ExtInter(\R)_\eta
    = \left\{
      \begin{array}{ll}\bigcup_{\drow \in \eta(\rho)}
         \left(\prodc{\ell \in \fin(\R)} \IntF{\TypeInter{\R(\ell)}}_\eta \pfconcat \drow \right)
    &\text{if $\tail(\R) = \rho$}\\
      \prodc{\ell \in \Labels} \IntF{\TypeInter{\R(\ell)}}_\eta &\text{otherwise}
      \end{array}\right.
      \end{array}
\)

\noindent
where

\(\begin{array}{rll}
   \IntF{\TypeInter{t}}_\eta &= 
   {\TypeInter{t}}_\eta & \text{if } t\not=\neg t'
   \text{ and } t\not= t_1\vee t_2\qquad\qquad\\
   \IntF{\TypeInter{\fvar}}_\eta &= \eta(\fvar)\\
    \IntF{\TypeInter{\Undef}}_\eta &= \{\Undef\}
  \end{array}\begin{array}{rll}  
    \IntF{\TypeInter{\tau_1\lor\tau_2}}_\eta &=\IntF{\TypeInter{\tau_1}}_\eta\cup\IntF{\TypeInter{\tau_2}}_\eta \\
    \IntF{\TypeInter{\neg\tau}}_\eta &=
   (\ExtInter_\Domain\cup\{\Undef\}) \setminus  \IntF{\TypeInter{\tau}}_\eta
   \end{array}\)
\end{definition}    
Notice that the induction used in the definition is well-founded
thanks to the contractivity condition in the definition of types and
the fact that field-types are inductively defined.

The extensional interpretation is defined with respect to some domain
$\Domain$ and interpretation $\TypeInter .$, and maps types into a
domain $\ExtInter_\Domain$ that contains $\Domain$, 
constants $\Constants$ to interpret basic types,
sets of binary relations $\Pd(\Domain \times \Domain_\Omega)$ to
interpret function types, and quasi constant functions
$\Labels \qcfun \Domain_\Undef$ to interpret record types. The fact
that functions are binary relations that can yield a distinguished
element $\Omega$ (which, intuitively, represents a type error) is a
standard technique of semantic subtyping to avoid $\Any\to\Any$ to be a supertype  of
all function types: since it does not play any specific role in our
work we will not further
comment on it (see~\cite{frisch08} for a detailed explanation
or~\citet[Section 3.2]{castagna23a} for a shorter one).

The definitions for the extensional interpretation given on the right-hand side in \cref{d:extint}
 are the same as those by~\citet{castagna11}. They state that the
 empty type is interpreted as the empty set, the
 interpretation of the type variables is given by $\eta$, that unions
 and negations are interpreted as set-theoretic unions and
 complements, and that functions types are interpreted as sets of
 binary relations whose output is in the codomain if the input is in
 the domain.

The novelty of our definition is the interpretation of record
types given on the left-hand side. There are two cases. The easy case is when the tail  of the
record type is either $\crecsign$ or $\orecsign$: in that case the
interpretation of the record type is the set of all quasi constant
functions in $\Labels\qcfun \Pd(\Domain_\Undef)$ that map a label
$\ell$ into an element of the interpretation of $\R(\ell)$ (recall
that for $\ell\not\in\fin(\R)$, $\R(\ell)$ is $\Undef$
for $\tail(\R)=\crecsign$ and $\Any\lor\Undef$ for $\tail(\R)=\orecsign$).
If instead $\tail(\R)$ is a row variable $\rho$, then $\Labels$ is
partitioned in two, the sets $\fin(\R)$ and---by well-kindedness---$\rdef(\rho)$, and the interpretation of $\R$ will be the set of
quasi-constant functions in $\Labels\qcfun\Domain_\Undef$ obtained by
unioning two partial functions: a function
in $\prodc{\ell \in \fin(\R)} \IntF{\TypeInter{\R(\ell)}}_\eta$ for the $\fin(\R)$ labels of $\Labels$, and a
function in $\eta(\rho)$ for the remaining labels of $\Labels$. There
is a caveat: fields can map labels both into values and
$\Undef$. Therefore, the interpretation of field-types must be
slightly different from that of types, since it must map $\Undef$ into
$\{\Undef\}$ and the negation of a field-type is the complement
with respect to $\Domain_\Undef$, rather than $\Domain$: the $\IntF{\TypeInter .}$ does just that.

Given a domain $\Domain$ and a set-theoretic interpretation of the
types into this domain, they form a \emph{model} if the interpretation
and the associated extensional interpretation have the same zeros:
\begin{definition}[Model]\label{def:model}
Let $\Domain$ be a domain and
$
{\TypeInter{\cdot}}: \Types \to \etaInter \to \Pd(\Domain)$. The
 pair $(\Domain,
 {\TypeInter{\cdot}})$ is a \emph{model} if and only, if for
 all $t\in\Types$ and $\eta\in\etaInter$, 
 $
 {\TypeInter{t}}_\eta=\emptyset\iff\ExtInter(t)_\eta=\emptyset$.
 \end{definition}
Every model induces a subtyping relation on types\footnote{Actually,
a model must be convex: see~\citet{castagna11}. We omit this detail
since it is not relevant to our presentation.}:

\begin{definition}[Subtyping]\label{def:frischsubtyping}
  If $(\Domain,
  {\TypeInter{\cdot}})$
  is a model, then it induces a subtyping relation defined as follows:
  \[
    t_1 \leq t_2 \iffdef \forall\eta.
    {\TypeInter{t_1}_\eta} \subseteq
    {\TypeInter{t_2}_\eta}
  \]
\end{definition}
As explained by \citet[Section 2.6]{frisch04}, the interest of defining of a \emph{model} is that we can work
with the interpretation of the model ``\emph{as if}'' the interpretation
of the type
constructors (in particular, the function type constructor) were defined as their
extensional interpretation. So when deducing the properties for the
subtyping relation of a model---and just for the subtyping
relation---we can assume that \raisebox{-.5pt}{$\Pd(\overline{\parbox[][10.7pt][c]{48pt}{$\TypeInter{t_1}_\eta\times\overline{\parbox[][8.8pt][c]{18pt}{$\TypeInter{t_2}_\eta$}}$}}\,)$},
even if this is impossible for cardinality reasons.    

\Cref{def:model} specifies which characteristics a model must
have to induce a subtyping relation (that behaves ``\emph{as if}''), but it does not define any
particular model nor, thus, any particular subtyping relation. In what follows
we define a concrete interpretation domain $\Domain$ (whose elements
are defined by induction) and two specific interpretations, and prove
that they satisfy the conditions to be  models, since they both have
the same zeros as the extensional interpretation (of one of them). This
yields two equivalent definitions of a concrete subtyping relation we
are going to use in the rest of this presentation.
We will define:
\begin{itemize}
  \item An interpretation $\IntQ{\TypeInter{t}_\eta}$ parametrized by an assignment $\eta$,
    for which subtyping $t_1 \IntQ{\leq} t_2$ is defined as
    $\forall\eta.\IntQ{\TypeInter{t_1}_\eta} \subseteq \IntQ{\TypeInter{t_2}_\eta}$;
    (the index $q$ stands for \emph{\textbf{q}uantified}, since subtyping is
    quantified on all variable interpretations);
  \item An interpretation of types directly into sets $\TypeInter t$, that avoids
    quantification over $\eta$, and for which subtyping $t_1 \leq t_2$ is defined
    directly as $\TypeInter{t_1} \subseteq \TypeInter{t_2}$ (notice
    the absence of a variable interpretation argument $\eta$).
\end{itemize}

The interpretation of types is mutually recursive with interpretations of rows and
field-types.
Both of these will also give rise to subtyping relations.
We will use the same notation $\leq$ for the relations in $\Types \times
\Types$, $\Types_\Undef \times \Types_\Undef$ and
$\Rows \times \Rows$.

This interpretation in \cref{def:indint} induces a subtyping relation that it is easy to work with, since it got rid of the
interpretation $\eta$ for the variables. We can consider the
interpretation $\TypeInter{\cdot} : \Types \to \Pd(\Domain)$ as a
function in   $\Types  \to \etaInter\to \Pd(\Domain)$ that is constant
on its second argument: if we apply \cref{def:frischsubtyping} to it, then  the subtyping relation is
defined as simply as
$t_1 \leq t_2 \iffdef \TypeInter{t_1} \subseteq \TypeInter{t_2}$.
But getting rid of $\eta$ makes it difficult to prove that this interpretation is a model and, thus,
that when considering the properties of this subtyping relation, we can work ``as if'' the interpretation
of type constructors were as in the extensional interpretation. To overcome this difficulty we define a second
interpretation, on the same domain, but this interpretation disregards the indexes
of the elements and uses an assignment $\eta$ to interpret the variables.
\begin{definition}[Parametrized intepretation of types and rows]
  \label{def:indinteta}%
  We define a ternary predicate $\IntQ{(\dterm : \tterm)_\eta}$
  (``the element $\dterm$ belongs to $\tterm$ under assignment $\eta$''),
  by induction on the pair $(\dterm, \tterm)$ ordered lexicographically.
  The only differences with the predicate $(\dterm : \tterm)$ (apart from
  recursive calls to the appropriate predicate), are:
  \begin{align*}
    \IntQ{(d : \alpha)_\eta}
    &= d \in \eta(\alpha) \\
    \IntQ{(\dundef : \fvar)_\eta}
    &= \dundef \in \eta(\fvar) \\
    \IntQ{(\domrow{(\ell = \dundef_\ell)_{\ell \in L_1}}{L_2}^V:\rw)_\eta}
    &= (\forall \ell \in L_1. \IntQ{(\dundef_\ell:\rw(\ell))_\eta})
    \textand (\forall \ell \in \rdef(\rw) \setminus L_1.
    \IntQ{(\Undef^\emptyset:\rw(\ell))_\eta}) \\[-1.5mm]
    &\phantom{=}\ \textand \tail(\rw) = \rho \Rightarrow
    \domrow{(\ell = \dundef_\ell)_{\ell \in L_1 \cap
    \rdef(\rho)}}{\Labels{\setminus}\rdef(\rho)}^V \in \eta(\rho)
  \end{align*}
  We define the interpretations $\IntQ{\TypeInter{\cdot}_\eta}$,
  $\IntQF{\TypeInter{\cdot}_\eta}$ and $\IntQR{\TypeInter{\cdot}_\eta}$ as expected.
\end{definition}
While the interpretation of type and field variables is
straightforwardly given by $\eta$, the interpretation of row variables
is less evident.
The first line of the interpretation of a row is the same as
in \cref{def:indint}: in both definitions this line deals with the case when
the tail of $r$ is not a row variable.
The second line covers the case for $\tail(r)=\rho$:  it checks that $\drow \in \eta(\rho)$, where $\drow$
is obtained by restricting the quasi-constant function on the left to the
definition space of $\rho$.

Our goal is to prove that both interpretations give a model of types.
Formally, this corresponds to proving the following equivalences:
Let $\ExtInter(\cdot)$ be the extensional interpretation of
$\IntQ{\TypeInter{.}}$. For all $t \in \Types$:
\begin{equation}
  \label{eq:same_zeros}%
  (\forall \eta.\ExtInter(t)_\eta = \emptyset)
  \iff (\forall \eta.\IntQ{\TypeInter{t}_\eta} = \emptyset)
  \iff \TypeInter{t} = \emptyset
\end{equation}
The leftmost ``iff'' proves that $(\Domain,\IntQ{\TypeInter{.}})$ is a
model, while the rightmost one proves that the subtyping relation
induced by $\TypeInter{.}$ is the same as the one induced by the model
$\IntQ{\TypeInter{.}}$ (since
$\forall \eta.\IntQ{\TypeInter{s}_\eta} \subseteq\IntQ{\TypeInter{t}_\eta} \iff\forall \eta.\IntQ{\TypeInter{s}_\eta} \cap
(\Domain{\setminus}\IntQ{\TypeInter{t}_\eta}) = \emptyset \iff\forall \eta.\IntQ{\TypeInter{s{\wedge}\neg
t}_\eta} = \emptyset \iff \TypeInter{s{\wedge}\neg t} = \emptyset \iff \TypeInter{s}\subseteq\TypeInter{t}$).

For the first equivalence, we prove the following, more precise, statement.
\begin{lemma}
  For all type $t$, for all $\eta$,
  \[
    \IntQ{\TypeInter{t}_\eta} = \emptyset \iff
    \ExtInter(t)_\eta = \emptyset
  \]
\end{lemma}
\begin{proof}
  For all $d$, we show $\IntQ{(d:t)_\eta} \iff d \in \ExtInter(t)_\eta$
  by induction on $t$ in both directions.
  This induction is well-founded because the cases for type constructors do not
  use induction, $\ExtInter(t)_\eta$ is defined on top of
  $\IntQ{\TypeInter{t}_\eta}$, and the number of type connectives is finite by
  regularity of the types.

  We start with the left-to-right implication and detail the case $t = \R$.
  By hypothesis there is $\drow = \domrow{(\ell = \dundef_\ell)_{\ell \in
  L}}{\emptyset}^V$
  such that $\IntQ{(\domrec{\drow}^{V'}:\R)_\eta}$.
  By hypothesis, $\forall \ell \in L. \IntQ{(\dundef_\ell : \R(\ell))_\eta}$,
  so $\dundef_\ell \in \IntQF{\TypeInter{\R(\ell)}_\eta}$,
  and $\forall\ell \in \Labels \setminus L.
  \IntQ{(\Undef^\emptyset:\R(\ell))_\eta}$,
  so $\Undef^\emptyset \in \IntQF{\TypeInter{\R(\ell)}_\eta}$.
  It is easy to see that this implies respectively
  $\dundef_\ell \in \IntF{\TypeInter{\R(\ell)}_\eta}$
  and $\Undef^\emptyset \in \IntF{\TypeInter{\R(\ell)}_\eta}$.
  Moreover, if $\tail(\rho) = \rho$, then
  $\domrow{(\ell = \dundef_\ell)_{\ell \in L \cap
  \rdef(\rho)}}{\Labels{\setminus}\rdef(\rho)}^V \in \eta(\rho)$.
  So $\drow \in \ExtInter_\eta(\R)$.

  Now, for the right-to-left implication.
  Let $\drow = \domrow{(\ell = \dundef_\ell)_{\ell \in L}}{\emptyset}^V \in \ExtInter_\eta(\R)$.
  For all $\ell \in \fin(\R)$, we have by hypothesis $\dundef \in
  \IntF{\TypeInter{\R(\ell)}}$, which implies
  $\IntQF{(\dundef:\R(\ell))_\eta}$.
  Let $\ell \notin \fin(\R)$.
  If $\ell \in L_1$, then $\IntQ{(\dundef_\ell : \R(\ell))_\eta}$ holds.
  If $\ell \notin L_1$: if $\tail(\R) \in \Vars$, by definition $\R(\ell) = \Any
  \lor \Undef$, and if $\tail(\R) \notin \Vars$, $\R(\ell) = \Any \lor \Undef$,
  or $\R(\ell) = \Undef$.
  In any case, $\IntQ{(\Undef^\emptyset : \R(\ell))_\eta}$ holds.
  Finally, if $\tail(\R) = \rho \in \Vars$, then
  $\domrow{(\ell = \dundef_\ell)_{\ell \in L \cap
  \rdef(\rho)}}{\Labels{\setminus}\rdef(\rho)}^V \in \eta(\rho)$.
\end{proof}

We have thus shown that the subtyping relation generated by
$\IntQ{\TypeInter{\cdot}}$ has the expected properties described by
$\ExtInter(\cdot)$.
In particular, it is a \emph{set-theoretic} model because type operators are
interpreted as set operators.

Since it will be easier to work directly with interpretations as sets rather
than to quantify over $\eta$, we now show that the subtyping relation generated
by $\TypeInter{\cdot}$ is equivalent to the parametrized one.
We show that equivalence not only on types, but also on field-types and rows.
The main element of the proof is the canonical assignment $\hat\eta$, defined as
\begin{align}
  \hat\eta(\alpha)
  &= \{ d \in \Domain \mid \alpha \in \Tag(d) \}\\
  \hat\eta(\fvar)
  &= \{ \dundef \in \Domain_\Undef \mid \fvar \in \Tag(\dundef) \}\\
  \hat\eta(\rho) 
  &= \{ \drow \in \Domainrow \mid \rdef(\drow) = \rdef(\rho) \textand \rho \in \Tag(\drow) \}
\end{align}

\begin{lemma}
  \label{l:hatmu}%
  For every $\tterm \in \Types_\Undef \cup \Rows$,
  $\TypeInter{\tterm} = \IntQ{\TypeInter{\tterm}_{\hat\eta}}$.
\end{lemma}
\begin{proof}
  For any $\dterm$ and $\tterm$ we prove that $(\dterm:\tterm) \iff
  \IntQ{(\dterm:\tterm)_{\hat\eta}}$
  by induction on $(\dterm,\tterm)$.
  The only interesting case is when $\tterm = \rw$
  and $\dterm = \domrow{(\ell = \dundef_\ell)_{\ell \in L_1}}{L_2}^V$.
  For all $\ell \in L_1$ we have $(\dundef_\ell : \rw(\ell)) \iff
  \IntQ{(\dundef_\ell : \rw(\ell))_{\hat\eta}}$ by induction hypothesis.
  Let $\ell \notin L_1$.
  We show that $(\Undef^\emptyset : \rw(\ell)) \iff
  \IntQ{(\Undef^\emptyset : \rw(\ell))_\eta}$ by induction on $\rw(\ell)$.
  This induction is well-founded : for $\rw(\ell) = \Undef$, both propositions
  are true and they are false for any other type constructor (in particular
  $\rw(\ell) = \fvar$). The inductive cases on type operators are
  straightforward.
  If $\tail(\rw) \notin \Vars$, we are done.
  Otherwise, let $\tail(\rw) = \rho$.
  On the left side, we have $\rho \in V$.
  On the right side, we have
  $\domrow{(\ell = \dundef_\ell)_{\ell \in
  L\cap\rdef(\rho)}}{\Labels{\setminus}\rdef(\rho)}^V \in \hat\eta(\rho)$.
  By definition of $\hat\eta(\rho)$, this is equivalent to $\rho \in V$.
\end{proof}

\begin{lemma}
  \label{l:hateta_to_eta}%
  Let $W \in \Pf(\Vars)$ and $T_W = \{\tterm \in \Types_\Undef \cup \Rows
  \mid \vars(\tterm) \subseteq W\}$.
  For every $\tterm \in T_W$,
  \[
    \IntQ{\TypeInter{\tterm}_{\hat\eta}} = \emptyset \iff
    \forall\eta. \IntQ{\TypeInter{\tterm}_\eta} = \emptyset.
  \]
\end{lemma}
\begin{proof}
  The right-to-left implication is trivial, by instantiation of the quantifier
  by $\eta'$.
  The left-to-right implication is by contraposition: for an arbitrary $W$ and $T_W$, we prove
  $\forall \tterm \in T_W. (\exists \eta. \IntQ{\TypeInter{\tterm}_\eta} \neq \emptyset
  \implies \IntQ{\TypeInter{\tterm}_{\hat\eta}} \neq \emptyset)$.
  For this, we define the functions $F_W^\eta:\Domain_\Undef \cup \Domainrow
  \cup \{\Omega\} \to \Domain_\Undef \cup \Domainrow \cup \{\Omega\}$
  as $F_W^\eta(\Omega) = \Omega$ and:
  \[
    F_W^\eta(\dterm) =
    \begin{cases}
      c^{\hat V(\dterm)}
      &\text{if } \dterm = c;\\
      \{(F_W^\eta(d_1),F_W^\eta(\domega_1)), \dots,
      (F_W^\eta(d_n),F_W^\eta(\domega_n))\}^{\hat V(\dterm)}
      &\text{if } \dterm = \{(d_1,\domega_1), \dots, (d_n, \domega_n)\}^V;\\
      \domrec{\drow}^{\hat V(\dterm)}
      &\text{if } \dterm = \domrec{\drow}^V;\\
      \Undef^{\hat V(\dterm)}
      &\text{if } \dterm = \Undef^V;\\
      \domrow{(\ell = F_W^\eta(\dundef_\ell))_{\ell \in L_1}}{L_2}^{\hat V(\dterm)}
      &\text{if } \dterm = \domrow{(\ell = \dundef_\ell)_{\ell \in L_1}}{L_2}^V
    \end{cases}
  \]
  where $\hat V(\dterm) =
  \{ \alpha \in W \mid \dterm \in \eta(\alpha) \} \cup
  \{ \fvar \in W \mid \dterm \in \eta(\fvar) \}
  \cup \{ \rho \in W \mid \dterm \in \eta(\rho) \}$.
  The finiteness of $W$ ensures that $\hat V$ is finite.
  We prove the following statement, for an arbitrary $\eta$ and by induction on
  $(\dterm,\tterm)$ ordered lexicographically:
  \[
    \forall \tterm \in T_W. \forall \dterm \in \Domain_\Undef \cup \Rows.
    \IntQ{(\dterm:\tterm)_\eta} \implies \IntQ{(F_W^\eta(\dterm):\tterm)_{\hat\eta}}
  \]
  \begin{itemize}
    \item $\tterm = \alpha$.
      We have
      $\IntQ{(F_W^\eta(\dterm):\alpha)_{\hat\eta}}
      \iff \alpha \in \Tag(F_W^\eta(\dterm))
      \iff \alpha \in \hat V(\dterm)
      \iff \dterm \in \eta(\alpha) \textand \alpha \in W
      \iff \IntQ{(\dterm:\alpha)_\eta}$.
      The last equivalence holds by the hypothesis that $\tterm \in T_W$.
      The case for $\tterm = \fvar$ is similar.
    \item $\tterm = \Undef$ and $\dterm = \Undef^V$.
      $\IntQ{(F_W^\eta(\Undef^V):\Undef)_{\hat\eta}}
      = \IntQ{(\Undef^{\hat V(\dterm)}:\Undef)_\eta}$ holds.
    \item $\tterm = \R$ and $\dterm = \domrec{\drow}^V$.
      By hypothesis, we have $\IntQ{(\drow:\rectorow{\R})_\eta}$.
      By induction, this implies
      $\IntQ{(F_W^\eta(\drow):\rectorow{\R})_{\hat\eta}}$
      and thus $\IntQ{(F_W^\eta(\dterm):\R)_{\hat\eta}}$.
    \item $d = \domrow{(\ell = \dundef_\ell)_{\ell \in L_1}}{L_2}^V$
      and $t = \rw$.
      The statement holds for labels in and outside of $L_1$ by induction
      hypothesis.
      If $\tail(\rw) \notin \Vars$, we are done.
      Let $\tail(\rw) = \rho$.
      By hypothesis, $\rho \in W$.
      By definition of the predicate, there is
      $\drow = \domrow{(\ell = \dundef_\ell)_{\ell \in
      L_1\cap\rdef(\rho)}}{\Labels{\setminus}\rdef(\rho)}^V \in \eta(\rho)$.
      As in the case for type variables, we show:
      $\IntQ{(F_W^\eta(\drow):\rho)_{\hat\eta}}
      \iff \rho \in \Tag(F_W^\eta(\drow))
      \iff \rho \in \hat V(\drow)
      \iff \dterm \in \eta(\rho) \textand \rho \in W$.
    \item Other cases can be found in \cite[Lemma 2.8]{petrucciani19}
      or are direct by falsity of the premise.
      \qedhere
  \end{itemize}
\end{proof}

\begin{lemma}
  \label{l:same_zeros_typeinter}
  For all $t_1$ and $t_2$,
  $t_1 \IntQ{\leq} t_2 \iff t_1 \leq t_2$.
\end{lemma}
\begin{proof}
  By definition, \cref{l:hateta_to_eta} and \cref{l:hatmu}, we show:
  \[
    t_1 \IntQ{\leq} t_2
    \iff \forall \eta. \IntQ{\TypeInter{t_1 \typediff t_2}_\eta} = \emptyset
    \iff \IntQ{\TypeInter{t_1 \typediff t_2}_{\hat\eta}} = \emptyset
    \iff \TypeInter{t_1 \typediff t_2} = \emptyset
    \iff t_1 \leq t_2
    \qedhere
  \]
\end{proof}

\subsection{Subtyping Relation}

\begin{lemma}
  \label{l:normalize_row_easy}%
  Let $\rw$ be an atomic row of definition space $\Labels \setminus L_\rw$ and $L$ such
  that $\fin(\rw) \subseteq L \subseteq \rdef(\rw)$. Then, \[
    \rw \simeq \orow{(\ell = \rw(\ell)}{L_\rw}
    \wedge \row{L'}{\tail(\rw)}{L_\rw}
    \simeq \bigwedge_{\ell \in L} \orow{\ell = \rw(\ell)}{L_\rw}
    \wedge \row{L'}{\tail(\rw)}{L_\rw}
  \] where $L' = \fin(\rw)$ if $\tail(\rw) \in \Vars$ and $L' = L$ otherwise.
\end{lemma}
\begin{proof}
  Straightforward by the definition of the models.
\end{proof}

The function $\Phi$ we define in \cref{sec:subtyping} to decide subtyping
crucially relies on the formula we give in \cref{l:subtyping}.
It gives a characterization of the emptiness of $\bigwedge_{\rw \in
P} \rw \wedge \bigwedge_{\rw \in N} \neg\rw$.
While function $\Phi$ is stated on records, we prefer to state this lemma on
rows, as we will refer back to it in that way for tallying
(\cref{app:tallying}).
The corollary for records follows immediately by $t_1 \leq t_2 \iff
\rectorow{t_1} \leq \rectorow{t_2}$ when $t_1, t_2 \leq \orec{}$.
This lemma  generalizes to rows and polymorphic record types the
decomposition of monomorphic ones defined
by \citet{frisch04}.
The main difference is the addition the third
condition on line~\eqref{eq:subtyping_vars},
that checks whether a row variable appears both in the positive and in the
negative fragment.
\begin{restatable}{lemma}{subtyping}
  \label{l:subtyping}%
  Let $P$ and $N$ be sets of atomic row types $r$ each of definition space
  $\Labels{\setminus}L_\rw$.
  Let $L$ be a finite set of labels such that $\bigcup_{\rw \in P \cup N} \fin(\rw) \subseteq
  L \subseteq \Labels{\setminus}L_\rw$. Let
  $P_\Vars = \{\rw \in P \mid \tail(\rw) \in \Vars\}$ and likewise for $N_\Vars$.
  For every $\rw$, we define its default type $\Def(\rw)$ as: $\Def(\rw) = \Undef$ if $\rw$ is closed, and
  $\Def(\rw) = \Any \lor \Undef$ otherwise.
  The relation
  $\bigwedge_{\rw \in P} \rw \leq \bigvee_{\rw \in N} \rw$
  holds \textit{iff}
  $\forall\iota: N \to L\cup\lbrace\underline{~~}\rbrace$,
  \begin{align}
    \label{eq:subtyping_dom}%
    &\displaystyle\left(\exists \ell \in L. \bigwedge_{\rw\in P}
    \!\rw(\ell)~~~\leq\!\!\bigvee_{\rw\in \iota^{-1}(\ell)}\!\!\!\!\rw(\ell)\right)\\
    \label{eq:subtyping_nvars}%
    &\textor \displaystyle\left(\exists \rw_\circ \in \iota^{-1}(\_) \setminus N_\Vars.(
      \bigwedge_{\rw \in P} \Def(\rw) \leq \Def(\rw_\circ)
    )\;\right)\\
    \label{eq:subtyping_vars}%
    &\textor \left(\exists \rw_\circ \in \iota^{-1}(\_) \cap N_\Vars.
      \exists \rw \in P_\Vars.
    \tail(\rw_\circ) = \tail(\rw)\right)
  \end{align}
\end{restatable}
\begin{proof}
  In the following, we let $L_i = \fin(\rw_i)$ and $\tl_i = \tail(\rw_i)$.
  Using \cref{l:normalize_row_easy}, we decompose the conjunction into:
  \begin{equation}\label{eq:ten}
    \bigwedge_{\rw_p \in P}
    (\orow{(\ell = \rw_p(\ell))_{\ell \in L}}{L_\rw} \wedge
    \row {L_p'}{\tl_p}{L_\rw}) \wedge
    \bigwedge_{\rw_n \in N}
    (\vee_{\ell \in L} \orow{\ell = \neg{\rw_n(\ell)}}{L_\rw}
    \vee \neg\row {L_n'}{\tl_n}{L_\rw})
  \end{equation}
  Where $L_i' = L_i$ if $r_i \in \Vars$ and $L_i' = L$ otherwise.
  We can distribute the intersection of the elements of $N$ on the right of
  \eqref{eq:ten} over the unions in the second brackets.
  We obtain a union of intersections of, each time, $|N|$ elements, where each
  intersection is a possible combination of the individual rows present in the
  second line.
  Each combination is described by a function $\iota: N \to L \cup
  \{\_\}$, where $\iota(\rw_n) = \ell$ means that the element $\orow{\ell =
  \neg{\rw_n(\ell)}}{L_\rw}$ is present in the combination given by $\iota$, while
  $\iota(\rw_n) = \_$ means that the element $\neg\row {L_n'}{\tl_n}{L_\rw}$ is present
  in the combination.
  For each $\rw_n \in N$, let us write
  $\rw^n_\ell = \orow{\ell = \neg{\rw_n(\ell)}}{L_\rw}$ and
  $\rw^n_{-} = \neg\row{L_n'}{\tl_n}{L_\rw}$.
  Therefore the row in~\eqref{eq:ten} is equivalent to:
  \begin{equation}
      \bigwedge_{\rw_p \in P}
      (\orow{(\ell = \rw_p(\ell))_{\ell \in L}}{L_\rw} \wedge
      \row {L_p'}{\tl_p}{L_\rw}) \wedge
      \bigvee_{\iota: N \to L \cup \{\_\}}
      (\bigwedge_{\rw_n \in N} \rw^n_{\iota(\rw_n)})
  \end{equation}
  By distributing the intersection over the union we obtain
  \begin{equation}
      \bigvee_{\iota: N \to L \cup \{\_\}}\left(
      \bigwedge_{\rw_p \in P}
      (\orow{(\ell = \rw_p(\ell))_{\ell \in L}}{L_\rw} \wedge
      \row {L_p'}{\tl_p}{L_\rw}) \wedge
      \bigwedge_{\rw_n \in N} \rw^n_{\iota(\rw_n)}\right)
    \end{equation}
  A union is empty if and only if each summand of the union is
  empty. Therefore the row above is empty if and only if for all $\iota: N \to L \cup \{\_\}$,
  the following is empty:
  \begin{align*}
    &\bigwedge_{\rw_p \in P}
    (\orow{(\ell = \rw_p(\ell))_{\ell \in L}}{L_\rw} \wedge
    \row {L_p'}{\tl_p}{L_\rw}) \wedge
    \bigwedge_{\rw_n \in N} \rw^n_{\iota(\rw_n)}\\
    &\simeq \bigwedge_{\rw_p \in P}
    (\orow{(\ell = \rw_p(\ell))_{\ell \in L}}{L_\rw} \wedge
    \row {L_p'}{\tl_p}{L_\rw}) \wedge
    \bigwedge_{\ell \in L \cup \{\_\}} \bigwedge_{\rw_n \in \iota^{-1}(\ell)} \rw^n_{\ell}\\
    &\simeq \bigwedge_{\rw_p \in P}
    (\orow{(\ell = \rw_p(\ell))_{\ell \in L}}{L_\rw} \wedge
    \row {L_p'}{\tl_p}{L_\rw}) \wedge
    \bigwedge_{\rw_n \in \iota^{-1}(\_)} \rw_{-}^n
    \wedge \bigwedge_{\ell \in L} \bigwedge_{\rw_n \in \iota^{-1}(\ell)} \rw^n_\ell\\
    &= \bigwedge_{\rw_p \in P}
    (\orow{(\ell = \rw_p(\ell))_{\ell \in L}}{L_\rw} \wedge
    \row {L_p'}{\tl_p}{L_\rw})
    \wedge \bigwedge_{\rw_n \in \iota^{-1}(\_)} \neg\row{L_n'}{\tl_n}{L_\rw}
    \wedge \bigwedge_{\ell \in L} \bigwedge_{\rw_n \in \iota^{-1}(\ell)}
    \orow{\ell = \neg{\rw_n(\ell)}}{L_\rw}\\
    &\simeq \bigwedge_{\rw_p \in P}
    (\orow{(\ell = \rw_p(\ell))_{\ell \in L}}{L_\rw} \wedge
    \row {L_p'}{\tl_p}{L_\rw})
    \wedge \bigwedge_{\rw_n \in \iota^{-1}(\_)} \neg\row{L_n'}{\tl_n}{L_\rw}
    \wedge \orow{(\ell = \textstyle{\bigwedge_{\rw_n \in \iota^{-1}(\ell)}}
    \neg{\rw_n(\ell)})_{\ell \in L}}{L_\rw}\\
    &\simeq \orow{(\ell = \bigwedge_{\rw_p \in P} \rw(\ell) \wedge
    \bigwedge_{\rw_n \in \iota^{-1}(\ell)} \neg{\rw(\ell)})_{\ell \in L}}{L_\rw}
    \wedge \bigwedge_{\rw_p \in P} \row{L_p'}{\tl_p}{L_\rw}
    \wedge \bigwedge_{\rw_n \in \iota^{-1}(\_)} \neg\row{L_n'}{\tl_n}{L_\rw}\\
    \begin{split}
      &\simeq \orow{(\ell = \bigwedge_{\rw_p \in P} \rw_p(\ell)
      \wedge \bigwedge_{\rw_n \in \iota^{-1}(\ell)} \neg{\rw_n(\ell)})_{\ell \in L}}{L_\rw}\\
      &\quad\wedge \bigwedge_{\rw_p \in P_{\bar\Vars}} \row{L}{\tl_p}{L_\rw}
      \wedge \bigwedge_{\rw_n \in \iota^{-1}(\_) \cap N_{\bar\Vars}}
      \neg\row{L}{\tl_n}{L_\rw}
      \wedge \bigwedge_{\rw_p \in P_\Vars} \row{L_p}{\tl_p}{L_\rw}
      \wedge \bigwedge_{\rw_n \in \iota^{-1}(\_) \cap N_\Vars}
      \neg\row{L_n}{\tl_n}{L_\rw}
    \end{split}
  \end{align*}

  Let:
  \begin{itemize}
    \item $\rw_\iota^1 = \orow{(\ell = \bigwedge_{\rw_p \in P} \rw_p(\ell) \wedge
        \bigwedge_{\rw_n \in \iota^{-1}(\ell)} \neg{\rw_n(\ell)})_{\ell \in
      L}}{L_\rw}$;
    \item $\rw_\iota^2 = \bigwedge_{\rw_p \in P_{\bar\Vars}} \row {L}{\tl_p}{L_\rw} \wedge
      \bigwedge_{\rw_n \in \iota^{-1}(\_) \cap N_{\bar\Vars}} \neg\row
      {L}{\tl_n}{L_\rw}$;
    \item $\rw_\iota^3 = \bigwedge_{\rw_p \in P_\Vars} \row {L_p}{\tl_p}{L_\rw} \wedge
      \bigwedge_{\rw_n \in \iota^{-1}(\_) \cap N_\Vars} \neg\row {L_n}{\tl_n}{L_\rw}$.
  \end{itemize}
  We can see that $\rw_\iota^1$ is empty \textit{iff}
  condition~\eqref{eq:subtyping_dom} holds,
  $\rw_\iota^2$ is empty \textit{iff} condition~\eqref{eq:subtyping_nvars} does
  (in the case where
  $P_{\bar\Vars}$ is empty, notice that the intersection is equal to $\Any \lor
  \Undef$), and $\rw_\iota^3$ is empty \textit{iff}
  condition~\eqref{eq:subtyping_vars} holds.
  We directly obtain that if one of the conditions holds, then the row
  $\rw_\iota$ is empty.
  We now show that if $\rw_\iota$ is empty, then there is $1 \leq i \leq 3$ such
  that $\rw_\iota^i$ is empty.

  For this, we suppose that none of the subtypes is empty and build an element
  $\drow \in \IntR{\TypeInter{\rw_\iota}}$.
  \begin{enumerate}
    \item Since $\rw_\iota^1$ is not empty, for all $\ell \in L$ there is an
      element $\dundef_\ell^1 \in \IntF{\TypeInter{\bigwedge_{\rw_p \in P} \rw_p(\ell) \wedge
      \bigwedge_{\rw_n \in \iota^{-1}(\ell)} \neg\rw_n(\ell)}}$.
    \item Since $\rw_\iota^2$ is not empty, there is an element
      $\domrow{(\ell = \dundef^2_\ell)_{\ell \in L_2}}{}^{V_2} \in
      \IntR{\TypeInter{\rw_\iota^2}}$.
    \item Since $\rw_\iota^3$ is not empty, there is an element $\drow_3 \in
      \IntR{\TypeInter{\rw_\iota^3}}$.
      However, the restrictions on the set of elements in
      $\IntR{\TypeInter{\rw_\iota^3}}$ only concern their tags so that any
      element $\drow'$ with $\Tag(\drow') = \Tag(\drow_3)$ and $\rdef(\drow') =
      \rdef(\drow_3)$ is in $\IntR{\TypeInter{\rw_3}}$.
      Let $V_3 = \Tag(\drow_3)$.
  \end{enumerate}
  We build the element
  $\domrow{(\ell = \dundef^1_\ell)_{\ell \in L},
  (\ell = \dundef^2_\ell)_{\ell \in L_2{\setminus}L}}{}^{V_3}$.
  This element belongs to $\IntR{\TypeInter{\rw_\iota}}$, which is a
  contradiction.
\end{proof}


\subsection{Subtyping Algorithm}

Naively implementing the above subtyping formula requires backtracking.
Indeed, for all map $\iota: N \to L \cup \{\_\}$, we have to check if subtyping
holds on one of the labels $\ell \in L$.
On that recursive call, since the types are coinductive, we need to assume
that the type we are checking is empty.
So, we are collecting a series of emptiness assumptions along the call stack.
If later a contradiction arises, we need to backtrack to the point where the
wrong assumption was introduced, to then take another branch, in our case, check
subtyping for another $\ell$.
Our function $\Phi$ avoids backtracking to compute subtyping more efficiently,
following \cite[Chapter 7]{frisch04}.

\correctsubalg*
\begin{proof}
  If $\R_\circ \simeq \Empty$, the result holds.
  Otherwise, we prove this by induction on the cardinality of $N$.
  In the following, given a variable $\rho \in V_p$, let us write $L_\rho$ for
  $\Labels \setminus \rdef(\rho)$.

  If $N = \emptyset$, the union $\bigvee_{\R \in N} \R$ is empty.
  So the statement holds if and only if $\R_\circ \simeq \Empty$,
  since $\Phi(\R_\circ, V_p, \emptyset) = \texttt{false}$,
  and $\bigwedge_{\rho \in V_p} \rec{L_\rho}{\rho}$ is never empty.

  Now, let $N = N' \cup \{\R\}$.
  Let $L = \fin(\R_\circ)$.
  $\R$ can be decomposed as
  $\bigwedge_{\ell \in L} \orec{\ell = \R(\ell)} \wedge \rec{L'} \tl$,
  where $L' = L$ if $\tl \notin \Vars$ and $L' = \fin(\R)$ otherwise.
  The left-hand side of the statement is thus equivalent to
  $\R_\circ \wedge \bigwedge_{\rho \in V_p} \rec{L_\rho}{\rho}
  \wedge (\bigvee_{\ell \in L} \orec{\ell = \neg\R(\ell)} \vee
  \neg\rec{L'} \tl)
  \leq \bigvee_{n \in N'} \R_n$.
  We can distribute the intersection over the unions.
  We must then prove an equivalence between $(\R_\circ \simeq \emptyset \textor
  \Phi(\R_\circ, V_p, N))$ and:
  \begin{align}
    \label{eq:correct_subalg1}
    &\forall \ell \in L. \R_\circ \wedge \orec{\ell=\neg\R(\ell)}
    \wedge \bigwedge_{\rho \in V_p} \rec{L_\rho}{\rho}
    \leq \bigvee_{n \in N'} \R_n\\
    \label{eq:correct_subalg2}
    \textand\ & \R_\circ \wedge \bigwedge_{\rho \in V_p}
    \rec{L_\rho}{\rho} \wedge \neg\rec{L'} \tl
    \leq \bigvee_{n \in N'} \R_n
  \end{align}
  We now have to verify that each of these statements are equivalent to
  $\Phi(\R_\circ, V_p, N)$.

  We start with \eqref{eq:correct_subalg1}.
  Let $\ell \in L$ and let $\R_\circ^\ell = \R_\circ \wedge \orec{\ell =
  \neg\R(\ell)}$.
  By the induction hypothesis, we have that $\R_\circ^\ell \wedge
  \bigwedge_{\rho \in V_p} \rec{L_\rho}{\rho} \leq \bigvee_{n \in N'} \R_n
  \iff \R_\circ^\ell \simeq \Empty$ or $\Phi(\R_\circ^\ell, V_p, N')$.
  Since $\R_\circ \nsimeq \Empty$, $\R_\circ^\ell \simeq \Empty \iff
  \R_\circ(\ell) \leq \R(\ell)$.

  We continue with \eqref{eq:correct_subalg2}.
  We define $\Psi(\R_\circ,V_p,\rec{L'}{\tl})$ as
  $(\tl = \orecsign \textor \tl = \tail(\R_\circ) \textor \tl \in V_p)$.
  We show that~\eqref{eq:correct_subalg2} is equivalent to:
  $\Psi(\R_\circ,V_p,\rec{L'}{\tl}) \textor
  (\neg\Psi(\R_\circ,V_p,\rec{L'}{\tl})
  \textand \Phi(\R_\circ, V_p, N'))$.
  It is easy to show using \cref{l:subtyping} that $\R_\circ \wedge
  \bigwedge_{\rho \in V_p} \rec{L_\rho}{\rho} \wedge
  \neg\rec{L'}\tl$ is empty \textit{iff}
  $\Psi(\R_\circ,V_p,\rec{L'}{\tl})$ holds.
  In particular, the first condition of \cref{l:subtyping} never holds since
  $L' = L$ when $\tl \notin \Vars$.
  Then, there are two cases.
  \begin{enumerate}
    \item $\R_\circ \wedge \bigwedge_{\rho \in V_p}
      \rec{L_\rho}{\rho} \wedge \neg\rec{L'} \tl \leq \Empty$.
      This means that $\Psi(\R_\circ,V_p,\rec{L'}{\tl})$ holds, and also that
      $\R_\circ \wedge \bigwedge_{\rho \in V_p}
      \rec{L_\rho}{\rho} \wedge \neg\rec{L'} \tl \leq \bigvee_{n \in N'} \R_n$.
    \item $\R_\circ \wedge \bigwedge_{\rho \in V_p}
      \rec{L_\rho}{\rho} \wedge \neg\rec{L'} \tl \nleq \Empty$.
      This means that $\Psi(\R_\circ,V_p,\rec{L'}{\tl})$ does not hold.
      Thus, there are two possible cases: either (a) $\tail(\R_\circ) =
      \crecsign$ and $\tl = \orecsign$, or (b) $\tl = \rho \notin V_p$.
      We show that
      $\R_\circ \wedge \bigwedge_{\rho \in V_p} \rec{L_\rho}{\rho}
      \wedge \neg\rec{L'}{\tl} \leq \bigvee_{\R \in N'} \R$ is equivalent to
      $\R_\circ \wedge \bigwedge_{\rho \in V_p} \rec{L_\rho}{\rho}
      \leq \bigvee_{\R \in N'} \R$.
      From this, we use the induction hypothesis to obtain the equivalence with
      $\Phi(\R_\circ, V_p, N')$ (since $\R_\circ \nsimeq \Empty$).

      The right-to-left implication of the equivalence is trivial.
      For the converse implication, we use \cref{l:subtyping} on
      $\R_\circ \wedge \bigwedge_{\rho \in V_p} \rec{L_\rho}{\rho}
      \leq \neg\rec{L'} \tl \vee \bigvee_{\R \in N'} \R$.
      By the hypothesis (a) and (b) in the corresponding cases, the second and
      third conditions of that lemma never hold.
      Thus, we have by hypothesis that
      $\forall \iota:N' \cup {\rec{L'}{\tl}} \to L \cup \{\_\}. \exists \ell \in
      L. \R_\circ(\ell) \leq \bigvee_{\R \in \iota^{-1}(\ell)} \R(\ell)$.
      The implication holds because $\forall \ell \in L. \rec{L'}{\tl}(\ell)
      = \Any \lor \Undef$.
  \end{enumerate}

  Summing up, we have proved by induction that if
  $\R_\circ \wedge \bigwedge_{\rho \in V_p} \rec{L_\rho}{\rho}$ is not empty,
  checking that it is a subtype of $\bigvee_{n \in N} \R_n$ is equivalent to
  checking both these two propositions:
  \begin{enumerate}
    \item $\forall \ell \in L\;.\;
      (\;(\R_\circ(\ell) \leq \R(\ell)) \quad\textsf{or}\quad
      \Phi(\R_\circ \wedge \orec{\ell = \neg\R(\ell)}, V_p, N')\;)$
    \item $(\Psi(\R_\circ,V_p,\rec{L'}{\tl})) \quad\textsf{or}\quad
      (\neg\Psi(\R_\circ,V_p,\rec{L'}{\tl}) \textsf{ and } \Phi(\R_\circ, V_p, N'))$
  \end{enumerate}
  To conclude, notice that if
  $(\neg\Psi(\R_\circ,V_p,\rec{L'}{\tl}) \textsf{ and } \Phi(\R_\circ, V_p, N'))$
  holds, then by induction hypothesis we have
  $\R_\circ \wedge_{\rho \in V_p} \rec{L_\rho}{\rho} \leq \bigvee_{n \in N'} \R_n$
  and, \textit{a fortiori}, $\R_\circ \wedge \bigwedge_{\rho \in V_p}
  \rec{L_\rho}{\rho} \leq \R \vee \bigvee_{n \in N'} \R_n$.
  It is therefore useless to check the other proposition $(1)$ above, which thus
  must be checked only when $\Psi(\R_\circ,V_p,\rec{L'}{\tl})$ holds.
  This yields
  \begin{align*}
    &(\:\neg\Psi(\R_\circ,V_p,\rec{L'}{\tl}) \textsf{ and }\ \Phi(\R_\circ, V_p, N')\: )\\
    \textsf{ or } &
    (\ \Psi(\R_\circ,V_p,\rec{L'}{\tl}) \textsf{ and }\ (\forall\ell \in
    L\,.\,(\R_\circ(\ell) \leq \R(\ell) \ \textsf{ or
    }\ \Phi(\R_\circ\wedge\orec{\ell = \neg\R(\ell)}, V_p, N')))
  \end{align*}
  which corresponds to the second clause of the definition of $\Phi$.
\end{proof}

\termsubalg*
\begin{proof}
  The number of disjoint types in a DNF is finite.
  Also, the preprocessing of a conjunction of record types can be defined for
  all such types and computed in a finite number of steps.
  Finally, in function $\Phi$, the number of elements in the third parameter
  decreases at each recursive call.
  Moreover, the subtyping relation on non-record types is decidable
  \cite{castagna11}, from which we get decidability of the subtyping relation on
  field-types.
\end{proof}

\subsection{Substitutions}

\begin{lemma}
  \label{l:rowsub_commute}
  Let $\IntQQ{\TypeInter{\cdot}_\eta}$ be the appropriate interpretation
  among $\IntQ{\TypeInter{\cdot}_\eta}$,
  $\IntQF{\TypeInter{\cdot}_\eta}$ and
  $\IntQR{\TypeInter{\cdot}_{\eta}}$.
  For every $\tterm$, $\sigma$ and $\eta$, if $\eta'$ is defined by
  $\eta'(\vterm) = \IntQQ{\TypeInter{\sigma(\vterm)}_\eta}$,
  then $\IntQQ{\TypeInter{\tterm\sigma}}_\eta = \IntQQ{\TypeInter{\tterm}}_{\eta'}$.
\end{lemma}
\begin{proof}
  For an arbitrary $\sigma$ and $\eta$, we show that
  \[
    \forall \tterm \in \Types_\Undef \cup \Rows.
    \forall \dterm \in \Domain_\Undef \cup \Domainrow.
    \IntQ{(\dterm:\tterm\sigma)_\eta} \iff \IntQ{(\dterm:\tterm)_{\eta'}}
  \]
  by induction on $(\dterm,\tterm)$ and with $\eta'$ defined as before.
  We detail two cases, the others are straightforward.
  \begin{itemize}
    \item $\tterm = \alpha$ and $\dterm = d$.
      On the left, we have $\IntQ{(d:\alpha\sigma)}_\eta
      = \IntQ{(d:\sigma(\alpha))_\eta}$
      and on the right $\IntQ{(d:\alpha)}_{\eta'}
      = d \in \eta'(\alpha)
      = d \in \IntQ{\TypeInter{\sigma(\alpha)}_\eta}
      = \IntQ{(d:\sigma(\alpha))}$.
    \item $\tterm = \rec{(\ell = \tau_\ell)_{\ell \in L}}{\rho}$
      and $\dterm = \domrecord{(\ell = \dundef_\ell)_{\ell \in L_d}}^V$.
      The case for rows is similar.
      We have $\tterm\sigma = 
      \orec{(\ell = \tau_\ell)_{\ell \in L}} \wedge \rec{L}{\sigma(\rho)}$.
      Let $\drow = \domrow{(\ell = \dundef_\ell)_{\ell \in L_d\setminus
      L}}{L}^V$.
      On the left, we have:
      \[
        \IntQ{(\dterm:\tterm\sigma)}_\eta
        = (\forall \ell \in L_d.
        \IntQ{(\dundef_\ell:\tterm(\ell)\sigma)}_\eta) \textand
        (\forall \ell \notin L_d.
        \IntQ{(\Undef^\emptyset:\tterm(\ell)\sigma)_\eta}
        \textand \IntQ{(\drow:\sigma(\rho))}_\eta
      \]
      By induction hypothesis, the first two conditions are equivalent to
      $\forall \ell \in L_d. \IntQ{(\dundef_\ell:\tterm(\ell))}_{\eta'}$
      and $\forall \ell \notin L_d. \IntQ{(\Undef^\emptyset:\tterm(\ell))_{\eta'}}$.
      By the same reasoning as in the previous case,
      the last condition is equivalent to
      $\IntQ{(\drow:\rho)_{\eta'}}$,
      which altogether give $\IntQ{(\dterm:\tterm)_{\eta'}}$.
      \qedhere
  \end{itemize}
\end{proof}

\subpreservesub*
\begin{proof}
  By definition, $t_1 \leq t_2 \iff \TypeInter{t_1 \typediff t_2} = \emptyset$.
  By \cref{l:same_zeros_typeinter}, this is equivalent to
  $\forall \eta.\IntQ{\TypeInter{t_1 \typediff t_2}_\eta}$.
  In particular, this holds for $\eta'$ defined as in \cref{l:rowsub_commute},
  so that $\IntQ{\TypeInter{t_1 \typediff t_2}_{\eta'}} = \emptyset$.
  By \cref{l:rowsub_commute}, this implies $\IntQ{\TypeInter{(t_1 \typediff
  t_2)\sigma}_\eta} = \emptyset$ which means $t_1\sigma \leq t_2\sigma$.
\end{proof}

%

\section{APPENDIX FOR LANGUAGE}
\label{app:language}
\subsection{Syntax and Semantics}

\begin{definition}[Top-level variables]\label{def:tlv}
  The top-level variables of a type (resp.\ field-type, row) are defined as
  $\tlv(t) = \tlv'(t) \cap \typeVars$ (resp.\ $\tlv(\tau) = \tlv'(\tau) \cap
  \fldVars$, $\tlv(\rw) = \tlv'(\rw) \cap \rowVars$).
  \[\begin{array}{ll}
    \tlv'(\alpha) = \{\alpha\} \\
    \tlv'(\fvar) = \{\fvar\} \\
    \tlv'(\row{\ell = \tau,\ldots,\ell=\tau}{\rho}{L}) =
    \tlv'(\rec{\ell=\tau,\ldots,\ell=\tau}{\rho}) = \{\rho\} \\
    \tlv'(T_1\vee T_2) = \tlv'(T_1) \cup \tlv'(T_2) \\
    \tlv'(\neg T) = \tlv'(T) \\
    \tlv'(T) = \emptyset & \text{ otherwise}
  \end{array}\]
\end{definition}

\begin{definition}
  \label{def:vars}%
  Given a type term $\tterm$, we write $\vars(\tterm)$ the set of variables
  occurring in it.
  The following equalities hold.
  \[ \begin{array}{ccc}
    \vars(\alpha) = \{\alpha\}
    & \vars(\tterm_1 \to \tterm_2) = \vars(\tterm_1) \cup \vars(\tterm_2)
    & \vars(b) = \emptyset\\
    \vars(\fvar) = \{\fvar\}
    & \vars(\tterm_1 \vee \tterm_2) = \vars(\tterm_1) \cup \vars(\tterm_2)
    & \vars(\Undef) = \emptyset\\
    \vars(\row{(\ell = \tau_\ell)_{\ell \in L_1}}{\tl}{L_2}) = \{\tl\} \cap \rowVars
    & \vars(\neg \tterm) = \vars(\tterm)
    &\vars(\Empty) = \emptyset\\
    \vars(\R) = \vars(\rectorow{\R})
  \end{array} \]
\end{definition}

\begin{lemma}
  \label{l:cutrow_correct}%
  If $t_1 \leq t_2$ with $t_2 \leq \orec{}$
  then $\delrec{t_1}\ell \leq \delrec{t_2}\ell$.
\end{lemma}
\begin{proof}
  We show that for any $t \leq \orec{}$, we have
  $\IntR{\TypeInter{\delrec t \ell}} =
  \bigcup_{(\domrec{\drow}^V:t)} \delrec \drow \ell$,
  where $\delrec {\domrow{(\ell = \tau_\ell)_{\ell \in L}}{\emptyset}^V} \ell
  = \{\domrow{(\ell = \tau_\ell)_{\ell \in L{\setminus}\ell}}{\{\ell\}}^{(V \setminus V_\ell) \cup V'}
  | V' \subseteq V_\ell\}$,
  where for all $\ell \in L$, $V_\ell \eqdef \{\rho \mid \ell \in
  \dom(\rho)\}$.

  We start with $\bigcup_{(\domrec{\drow}^V:t)} \delrec \drow \ell \subseteq
  \IntR{\TypeInter{\delrec t \ell}}$.
  Let $(\domrec{\domrow{(\ell = \delta_\ell)_{\ell \in L \cup \{\ell\}}}\emptyset^V}^{V_0}:t)$
  and $t$ in DNF.
  Let $\drow_\ell = \domrow{(\ell = \tau_\ell)_{\ell \in L}}{\{\ell\}}^{
  (V \setminus V_\ell) \cup V'}$ with $V' \subseteq V_\ell$.
  The proof is by induction on the top-level type connectives of $t$.
  \begin{itemize}
    \item $t = \rec{(\ell = \tau_\ell)_{\ell \in L'}}{\rho}$.
      \begin{itemize}
        \item $\ell \in L'$ or $\tl \notin \rowVars$.
          Then $\delrec t \ell =
          \row{(\ell = \tau_\ell)_{\ell \in L' \setminus \{\ell\}}}{\tl}{\{\ell\}}$
          and it is clear that $(\drow_\ell:\delrec t \ell)$.
        \item $\ell \notin L'$ and $\tl = \rho \in \rowVars$.
          Then $\delrec t \ell =
          \orow{(\ell = \tau_\ell)_{\ell \in L'}}{\{\ell\}}$.
          By construction, for each $\ell \in L'$, we have either
          $(\delta_\ell:t(\ell))$ if $\ell \in L$, or
          $(\Undef^\emptyset:t(\ell))$ otherwise.
          So $(\drow_\ell:\delrec t \ell)$ since there is no constraint on the
          tagged set of variables in the interpretation of $\delrec t \ell$.
      \end{itemize}
    \item $t = \neg\rec{(\ell = \tau_\ell)_{\ell \in L'}}{\tl}$.
      \begin{itemize}
        \item $\ell \in L'$ or $\tl \notin \rowVars$.
          Without loss of generality, we can suppose $\ell \in L'$ even in the
          second case.
          If $\tau_\ell \neq \Any \lor \Undef$, then
          $\delrec t \ell = \orow{}{\{\ell\}}$ and $(\drow_\ell: \delrec t \ell)$ is
          trivial.
          If $\tau_\ell = \Any \lor \Undef$, then $\delrec t \ell =
          \neg\row{(\ell = \tau_\ell)_{\ell \in L' \setminus \{\ell\}}}{\tl}{\{\ell\}}$
          and there is $\ell' \neq \ell$ such that $(\delta_\ell:t(\ell))$ is false,
          or $(\Undef^\emptyset:t(\ell))$ is false.
          Thus, $(\drow_\ell: \delrec t \ell)$.
        \item $\ell \notin L'$ and $\tl = \rho \in \rowVars$.
          Then $\delrec t \ell = \orow{}{\{\ell\}}$ and $(\drow_\ell: \delrec t
          \ell)$ is trivial.
      \end{itemize}
    \item $t = t_1 \wedge t_2$.
      We have $\delrec t \ell = \delrec{t_1}\ell \wedge \delrec{t_2}\ell$.
      By induction hypothesis, $(\delrec\drow\ell:\delrec{t_i}\ell)$
      for $i \in \{1,2\}$, so $(\delrec\drow\ell:\delrec t \ell)$.
    \item $t = t_1 \vee t_2$.
      We have $\delrec t \ell = \delrec{t_1}\ell \vee \delrec{t_2}\ell$.
      By induction hypothesis, there is $i \in \{1,2\}$ such that
      $(\delrec\drow\ell:\delrec{t_i}\ell)$.
      Thus, $(\delrec\drow\ell:\delrec t \ell)$.
  \end{itemize}

  Now, we consider $\IntR{\TypeInter{\delrec t \ell}} \subseteq
  \bigcup_{(\domrec{\drow}^V:t)} \delrec \drow \ell$.
  If $t \leq \Empty$, then by definition $\delrec t \ell \leq \Empty$.
  Otherwise, let $(\drow_\ell:\delrec t \ell)$ and
  $t = \bigvee_{i \in I} \bigwedge_{\R \in P_i} \R \wedge \bigwedge_{\R \in
  N_i} \neg\R$.
  For each $i \in I$, let $P_\Vars^i = \{\R \in P_i \mid \tail(\R) = \rho
  \textand \ell \notin \rdef(\rho)\}$, similarly for $N_\Vars^i$.
  By the reasoning of \cref{l:decomposition},
  we have \begin{align*}
    t &\simeq \bigvee_{i \in I} \bigvee_{N_i' \subseteq N_i}
    \bigvee_{N'_{i,\Vars} \subseteq N_i' \cap N_\Vars^i} t_{N'_{i,\Vars}}\\
      &= \bigvee_{i \in I} \bigvee_{N_i' \subseteq N_i}
      \bigvee_{N'_{i,\Vars} \subseteq N_i' \cap N_\Vars^i}
      \Big(\orec{\ell = \bigwedge_{\R \in P_i} \R(\ell) \wedge \bigwedge_{\R
        \in N_i{\setminus}N_i'} \neg\R(\ell)}
      \wedge \bigwedge_{\R \in P_i} \rec{\ell}{\cutrow {\rectorow\R} \ell} \\
      &\phantom{={}}\wedge \bigwedge_{\R \in N_i' \setminus N'_{i,\Vars}}
      \neg\rec{\ell}{\cutrow{\rectorow\R}\ell}
      \bigwedge_{\R \in P^i_\Vars} \rec{\fin(\R)}{\tail(\R)}
      \wedge \bigwedge_{\R \in N'_{i,\Vars}} \neg\rec{\fin(\R)}{\tail(\R)}
    \Big)
  \end{align*}
  where $\cutrow{\rectorow\R}\ell$ is the operation defined on positive atomic row
  types in \cref{def:cutrow}.
  This operation coincides with $\delrec \R \ell$ on a positive and atomic $\R$.
  For each $i$ and $N'_{i,\Vars}$, we have
  $\delrec{t_{N'_{i,\Vars}}}\ell
  = \bigwedge_{\R \in P_i} \rec{\ell}{\cutrow{\rectorow\R} \ell}
  \wedge \bigwedge_{\R \in N_i'{\setminus}N'_{i,\Vars}}
  \neg\rec{\ell}{\delrec{\rectorow\R}\ell})$.
  By hypothesis, there are $i \in I$ and $N'_{i,\Vars}$ such that
  $(\drow_\ell:\delrec{t_{N'_{i,\Vars}}}\ell)$.
  Let $\delta_\ell$ be such that $(\delta_\ell:\bigwedge_{\R \in P_i} \R(\ell) \wedge
  \bigwedge_{\R \in N'_i} \neg\R(\ell))$.
  Let $V = (\Tag(\drow_\ell) \cup \{\tail(\R) \mid \R \in P^i_\Vars\}) \setminus
  \{\tail(\R) \mid \R \in N'_{i,\Vars}\}$.
  We take $\drow$ to be $\drow_\ell$ completed by $\ell = \delta$ and with
  $\Tag(\drow) = V$ and we have $(\domrec{\drow}^V:t)$.
  Since $\{\tail(\R) \mid \R \in P^i_\Vars\}$ and $\{\tail(\R) \mid \R \in N'_{i,\Vars}\}$ are
  subsets of $V_\ell$, we have $\drow_\ell \in \delrec \drow \ell$.

  For the second case to work, we need to show that for all $t \leq \orec{}$,
  $\cutrow t \ell \simeq \cutrow{\spl(t)}\ell$,
  where $\spl(t)$ is the type obtained by the previous decomposition.
  First, let $t = \bigwedge_{\R \in P} \R \wedge \bigwedge_{\R \in N} \neg\R$.
  The proof is by induction on $|N|$.
  Let $P_\Vars$ and $N_\Vars$ be defined as before relative to $P$ and $N$.
  \begin{itemize}
    \item $N = \emptyset$. Then, \[
        \spl(t) = \orec{\ell = \bigwedge_{\R \in P} \R(\ell)} \wedge
        \bigwedge_{\R \in P} \rec{\ell}{\cutrow{\rectorow\R}\ell} \wedge
        \bigwedge_{\R \in P_\Vars} \rec{\fin(\R)}{\tail(\R)}
      \]
      and $\delrec{\spl(t)}\ell = \bigwedge_{R \in P} \cutrow{\rectorow\R}\ell = \delrec t \ell$.
    \item $N = N_0 \cup \{\R_n = \rec{\ell = \tau}{\rw}\}$ where $\rw$ is an atomic row.
      By equivalence, this covers all cases where $\tail(\R_n) \notin \rowVars$
      and the ones where $\tail(\R_n) = \rho$ and $\ell \notin \rdef(\rho)$.
      Let $t_0 = \bigwedge_{\R \in P} \R \wedge \bigwedge_{\R \in N_0} \neg\R$,
      so that $\delrec t \ell = \delrec{t_0}\ell \wedge \delrec{\R_n}\ell$.
      The type $t$ is decomposed as follows:
      \begin{align*}
        \spl(t)
        = &\bigvee_{N' \subseteq N_0}
        \bigvee_{N_\Vars' \subseteq N' \cap N_\Vars}
        \bigvee_{k \in \{0,1\}}
        \Big(\orec{\ell = \bigwedge_{\R \in P} \R(\ell)
          \wedge \bigwedge_{\R \in N{\setminus}N'} \R(\ell)
        \bigwedge_{k = 0} \neg\tau} \\
        &\quad\wedge \bigwedge_{\R \in P} \rec{\ell}{\cutrow{\rectorow\R} \ell}
          \wedge \bigwedge_{\R \in N'{\setminus}N_\Vars'}
          \neg\rec{\ell}{\cutrow{\rectorow\R}\ell}
          \bigwedge_{k = 1} \neg\rec{\ell}{\rw} \\
        &\quad\wedge \bigwedge_{\R \in P_\Vars} \rec{\fin(\R)}{\tail(\R)}
        \wedge \bigwedge_{\R \in N'_\Vars} \neg\rec{\fin(\R)}{\tail(\R)}
      \Big)
      \end{align*}
      There are two cases.
      \begin{enumerate}
        \item $\tau = \Any \lor \Undef$.
          Since $\neg\tau = \Empty$, the disjunction with $k = 0$ is empty,
          and the DNF of $\spl(t)$ is equal to
          \begin{align*}
            \bigvee_{N' \subseteq N_0}
            \bigvee_{N'_\Vars \subseteq N' \cap N_\Vars}
            \Big(&\orec{\ell = \bigwedge_{\R \in P} \R(\ell)
                 \wedge \bigwedge_{\R \in N{\setminus}N'} \R(\ell)}\\
                 &\wedge \bigwedge_{\R \in P} \rec{\ell}{\cutrow{\rectorow\R}\ell}
                 \wedge \bigwedge_{\R \in N'{\setminus}N'_\Vars}
               \neg\rec{\ell}{\cutrow{\rectorow\R}\ell}\Big)
               \wedge \neg\rec{\ell}{\rw}
          \end{align*}
          which is equal to $\spl(t_0) \wedge \neg\rec{\ell}{\rw}$.
          So by induction hypothesis and since $\delrec{\R_n}\ell = \neg\rw$, \[
            \delrec{\spl(t)}\ell
            = \delrec{\spl(t_0)}\ell \wedge \neg\rw
            \simeq \delrec{t_0}\ell \wedge \neg\rw
            = \delrec t \ell
          \]
        \item $\tau \neq \Any \lor \Undef$.
          Then, $\delrec{\spl(t)}\ell$ is equal to
          \begin{align*}
            \bigvee_{N' \subseteq N_0}
            \bigvee_{N'_\Vars \subseteq N' \cap N_\Vars} \Big(
              &\bigwedge_{\R \in P} \cutrow{\rectorow\R}\ell
              \wedge \bigwedge_{\R \in N'{\setminus}N'_\Vars}
              \neg(\cutrow{\rectorow\R}\ell)\\
            \vee &\bigwedge_{\R \in P} \cutrow{\rectorow\R}\ell
              \wedge \bigwedge_{\R \in N'{\setminus}N'_\Vars} \neg(\cutrow{\rectorow\R}\ell)
              \wedge \bigwedge_{\R \in N'{\setminus}N'_\Vars} \neg\rw
              \Big)
          \end{align*}
          So that $\delrec{\spl(t)}\ell$ is trivially equivalent to
          $\delrec{\spl(t_0)}\ell$,
          and by induction hypothesis and the fact that $\delrec{\R_n}\ell =
          \orow{}{\ell}$ we have
          $\delrec{\spl(t_0)}\ell \simeq \delrec{t_0}\ell = \delrec{t}\ell$.
      \end{enumerate}
    \item $N = N_0 \cup \{\R_n = \rec{(\ell' = \tau_\ell')_{\ell' \in L}}{\rho}\}$ where $\ell \notin L$.
      Let $t_0 = \bigwedge_{\R \in P} \R \wedge \bigwedge_{\R \in N_0} \neg\R$,
      so that $\delrec t \ell = \delrec{t_0}\ell \wedge \delrec{\R_n}\ell$.
      The type $t$ is decomposed as follows:
      \begin{align*}
        \spl(t)
        = &\bigvee_{N' \subseteq N_0}
        \bigvee_{N_\Vars' \subseteq N' \cap N_\Vars}
        \bigvee_{k \in \{1,2\}}
        \Big(\orec{\ell = \bigwedge_{\R \in P} \R(\ell)
          \wedge \bigwedge_{\R \in N{\setminus}N'} \R(\ell)} \\
        &\quad\wedge \bigwedge_{\R \in P} \rec{\ell}{\cutrow{\rectorow\R} \ell}
          \wedge \bigwedge_{\R \in N'{\setminus}N_\Vars'}
          \neg\rec{\ell}{\cutrow{\rectorow\R}\ell}
          \bigwedge_{k = 1} \neg\orec{\ell = \Any \lor \Undef, (\ell' = \tau_\ell')_{\ell' \in L}} \\
        &\quad\wedge \bigwedge_{\R \in P_\Vars} \rec{\fin(\R)}{\tail(\R)}
        \wedge \bigwedge_{\R \in N'_\Vars} \neg\rec{\fin(\R)}{\tail(\R)}
        \bigwedge_{k = 2} \neg\rec{L}{\rho}
      \Big)
      \end{align*}
      Thus, $\delrec{\spl(t)}$ is trivially equivalent to the type in case (2) above
      (with $\neg\orow{(\ell' = \tau_\ell')_{\ell' \in L}}{\{\ell\}}$ instead of
      $\neg\rw$ in the second line), which is equal to $\delrec{\spl(t_0)}\ell$
      and we conclude in the same way since here also $\delrec{\R_n}\ell =
      \orow{}{\{\ell\}}$.
  \end{itemize}
  Now, if $t = \bigvee_{i \in I} t_i$ where the $t_i$'s are conjunctions,
  we have $\delrec t \ell
  = \bigvee_{i \in I} (\delrec{t_i}\ell)
  \simeq \bigvee_{i \in I} (\delrec{\spl(t_i)}\ell)
  = \delrec{\spl(t)}\ell$.
\end{proof}

\begin{lemma}[Inversion]
  \label{l:inversion}%
  Let $v = \extrecord{v_1}{\ell}{v_2}$.
  If there is a derivation $\decseq \Delta \Gamma v t$,
  then there are derivations
  $\decseq \Delta \Gamma {v_1}{t_1 \leq \orec{\ell = \Undef}}$ and
  $\decseq \Delta \Gamma {v_2}{t_2}$ such that
  $\rec{\ell = t_2}{\cutrow{t_1}\ell} \leq t$.
\end{lemma}
\begin{proof}
  By induction on $\decseq \Delta \Gamma v t$, with a case analysis on the last
  rule used, that has to be of \rulename{Ext}, \rulename{Inter} or
  \rulename{Sub}.
  \begin{description}
    \item[\rulename{Ext}] Straightforward.
    \item[\rulename{Inter}] We apply the induction hypothesis twice.
      Since both types obtained are supertypes of
      $\rec{\ell = t_2}{\delrec{t_1}\ell}$,
      their intersection is also.
    \item[\rulename{Sub}] By induction hypothesis and transitivity of subtyping.
      \qedhere
  \end{description}
\end{proof}

\begin{lemma}[Subject reduction]
  \label{l:sr}%
  Let $e$ be an expression and $t$ a type.
  If $\decseq \Delta \Gamma e t$ and $e \reduces e'$,
  then $\decseq \Delta \Gamma {e'} t$.
\end{lemma}
\begin{proof}
  The proof is by induction on the derivation of $\decseq \Delta \Gamma e t$
  and by a case analysis on the last rule used in the derivation of
  $\decseq \Delta \Gamma e t$.
  We detail the cases related to the rules for records, for the rest, see
  \textit{e.g.} \cite{frisch08}.
  \begin{description}
    \item[\rulename{Emp}] $e = \erecord{}$, so it does not reduce.
    \item[\rulename{Ext}] $e = \extrecord{e_1}{\ell}{e_2}$.
      Necessarily, we have $e' = \extrecord{e_1'}{\ell}{e_2}$ 
      or $e' = \extrecord{e_1}{\ell}{e_2'}$ and this is direct by
      induction hypothesis.
    \item[\rulename{Sel}] If $e = e_0.\ell \reduces e_0'.\ell = e'$
      or if $e = \extrecord{e_1}{\ell'}{e_2}.\ell$ and the reduction occurs in
      $e_1$ or $e_2$, this is direct by induction hypothesis.
      Otherwise, we have $e = \extrecord{v}{\ell'}{v'}.\ell$ and
      $\decseq \Delta \Gamma {\extrecord{v}{\ell'}{v'}} {\orec{\ell = t}}$.
      By \cref{l:inversion}, there are derivations
      $\decseq \Delta \Gamma v {t_1 \leq \orec{\ell' = \Undef}}$
      and $\decseq \Delta \Gamma {v'}{t_2}$
      such that $\rec{\ell'=t_2}{\delrec{t_1}{\ell'}} \leq \orec{\ell = t}$.
      \begin{itemize}
        \item If $\ell = \ell'$, then $e' = v'$.
          Since $\rec{\ell = t_2}{\delrec{t_1}{\ell}} \leq \orec{\ell = t}$,
          we have $t_2 \leq t$ and we conclude by \rulename{Sub}.
        \item If $\ell \neq \ell'$, then $e' = v.\ell$.
          Since $\rec{\ell'=t_2}{\delrec{t_1}{\ell'}} \leq \orec{\ell = t}$,
          we have $t_1 \leq \orec{\ell = t, \ell' = \Undef}$.
          We conclude by rules~\rulename{Sub} and~\rulename{Sel}.
      \end{itemize}
    \item[\rulename{Del}] If $e = \delrecord{\erecord{}}\ell \reduces \erecord{}$,
      since $\delrecord{\erecord{}} \ell = \erecord{}$, we can use the same
      derivation.
      If $e = \delrecord{\extrecord{e_1}{\ell'}{e_2}} \ell$ and the reduction
      occurs in $e_1$ or $e_2$, this is direct by induction hypothesis.
      Otherwise, we have $e = \delrecord{\extrecord{v}{\ell'}{v'}} \ell$.
      By \cref{l:inversion}, there are derivations
      $\decseq \Delta \Gamma v {t_1 \leq \orec{\ell' = \Undef}}$
      and $\decseq \Delta \Gamma {v'}{t_2}$
      such that $\extrec{t_1}{\ell'}{t_2} \leq t$.
      \begin{itemize}
        \item If $\ell = \ell'$, then $e' = \delrecord v \ell$.
          Since $t_1 \leq \orec{\ell = \Undef}$, $v$ can be typed with
          $t_1 = \delrec{t_1}\ell = \delrec{\extrec{t_1}{\ell}{t_2}}\ell$.
          We conclude because $\extrec{t_1}\ell{t_2} \leq t$.
        \item If $\ell \neq \ell'$, then $e' = \extrecord{\delrecord v \ell}{\ell'}{v'}$.
          There is a derivation $\decseq \Delta \Gamma {\delrecord v \ell}
          {\delrec{t_1} \ell}$.
          Since $\extrec{\delrec{t_1}\ell}{\ell'}{t_2} =
          \delrec{\extrec{t_1}{\ell'}{t_2}} \ell$,
          we conclude by rule~\rulename{Ext}.
      \end{itemize}
    \item[\rulename{Inst}] Direct by induction hypothesis.
      \qedhere
  \end{description}
\end{proof}

\begin{lemma}[Generation for values]
  \label{l:generation}%
  Let $v$ be a value such that
  $\decseq \Delta \Gamma v t$ with $t \leq \orec{}$.
  Then $v$ has one of the forms $\erecord{}$ or $\extrecord{v_1}{\ell}{v_2}$.
\end{lemma}
\begin{proof}
  It is easy to see that a derivation for $v$ is obtained by a rule
  \rulename{Emp} followed by rules~\rulename{Ext}, \rulename{Inter} or
  \rulename{Sub}.
  Remark that there are no rules~\rulename{Inst} because it is impossible to
  derive a polymorphic type $t$ for $v$, in particular since for instance
  $\crec{} \nleq \rec{}{\rho}$.
  Moreover, we can show by induction on the depth of the derivation that if
  $\decseq \Delta \Gamma v t$ is derivable, then $t \nsimeq \Empty$.
  The proof is by induction on the derivation.
  \begin{description}
    \item[\rulename{Emp}] This is the base case, where $v = \erecord{}$.
    \item[\rulename{Ext}] $v$ is of the second form.
    \item[\rulename{Inter}] $t_1$ and $t_2$ are reducible to disjunctive normal
      forms $t_1^\R \wedge t_1'$ and $t_2^\R \wedge t_2'$,
      such that $t_1^\R, t_2^\R \leq \orec{}$ and by hypothesis,
      $t_1' \wedge t_2' \leq \Empty$.
      We can show by induction on the derivation of $v$ that this last property
      does not hold if $v$ is not a record expression.
    \item[\rulename{Sub}] We have $\decseq \Delta \Gamma v {t' \leq t}$.
      $t' \nleq \Empty$ since $v$ is a value, so we can apply the induction
      hypothesis. 
      \qedhere
  \end{description}
\end{proof}

\begin{lemma}[Progress]
  \label{l:progress}%
  Let $e$ be a well-typed closed expression, that is,
  $\decseq \emptyset \emptyset e t$ for some $t$.
  If $e$ is not a value, then
  there exists an expression $e'$ such that $e \reduces e'$.
\end{lemma}
\begin{proof}
  The proof is by induction on the derivation of $\decseq \Delta \Gamma e t$
  and by a case analysis on the last rule used in the derivation of
  $\decseq \Delta \Gamma e t$.
  We detail the cases related to records, for the rest, see
  \textit{e.g.}~\cite{frisch08}.
  \begin{description}
    \item[\rulename{Emp}] $e = \erecord{}$ is a value.
    \item[\rulename{Ext}] $e = \extrecord{e_1}{\ell}{e_2}$.
      If $e_1$ or $e_2$ can be reduced, $e$ can also.
      Otherwise, $e_1$ and $e_2$ are values by induction and so is $e$.
    \item[\rulename{Del}] $e = \delrecord{e_0}{\ell}$.
      If $e_0$ can be reduced, so can $e$.
      Otherwise, we have by induction hypothesis that $e_0$ is a value.
      By \cref{l:generation}, either $e_0 = \extrecord{v}{\ell'}{v'}$
      and $e$ reduces with $\rrulename[=]{del}$ or $\rrulename[\neq]{del}$,
      or $e_0 = \erecord{}$ and $e$ reduces with
      $\rrulename{emp}$.
    \item[\rulename{Sel}] $e = e_0.\ell$.
      If $e_0$ can be reduced, so can $e$.
      Otherwise, we have by induction hypothesis that $e_0$ is a value.
      By \cref{l:generation}, $e_0 = \extrecord{v}{\ell'}{v'}$ or $e_0 = \erecord{}$.
      In the first case, $e$ reduces with
      $\rrulename[=]{del}$ or $\rrulename[\neq]{del}$.
      The second case is impossible, since there is no derivation for $e_0$ of
      type $\orec{\ell = t}$.
    \item[\rulename{Inst}]
      Directly by induction hypothesis.
      \qedhere
  \end{description}
\end{proof}

\typesoundness*
\begin{proof}
  Consequence of \cref{l:sr,l:progress}.
\end{proof}

\subsection{Algorithmic Type System}\label{app:typealgo}

\subsubsection{Field selection}

We remind here the definition of the field selection operator.
\typesel*
For an arbitrary type $t \leq \orec{\ell = \Any}$, we define $t.\ell \eqdef
(\dnf(t \wedge \orec{}))$.

\begin{lemma}
  \label{l:remove_tlrv}%
  Let $t = \bigwedge_{\R \in P} \R \wedge \bigwedge_{\R \in N} \neg\R$ and $\rho
  \in \rowVars$ such that $\rho \notin \tlv(\rectorow{t})$.
  Then, $t \wedge \rec{\Labels{\setminus}\rdef(\rho)}{\rho} \leq \Empty
  \iff t \leq \Empty$.
\end{lemma}
\begin{proof}
  Because $\rho \notin \tlv(\rectorow{t})$,
  the set of elements $\domrec{\drow} \in \TypeInter{t ..}$ is exactly the
  elements of $\TypeInter{t}$, where $\rho$ is added from $\Tag(\drow)$.
\end{proof}

\begin{lemma}
  \label{l:correct_sel}%
  Let $t \leq \orec{\ell = \Any}$.
  Then, for all $t_\ell$, $t \leq \orec{\ell = t_\ell} \iff t.\ell \leq t_\ell$.
  In particular, $t.\ell \leq \Any$ and $t \leq \orec{\ell = t.\ell}$.
\end{lemma}
\begin{proof}
  By \cref{l:decomposition}, $t$ is equivalent to:
  \begin{align*}
    \bigvee_{i \in I}
    \bigvee_{N' \subseteq N_i}
    \bigvee_{N'_\Vars \subseteq N' \cap N_\Vars}
    \Bigg(
    &\orec{\ell = \bigwedge_{\R \in P} \R(\ell)
    \wedge \bigwedge_{\R \in N \setminus N'} \neg\R(\ell)}\\
    &\wedge \bigwedge_{\R \in P}
    \rec{\ell = \Any \lor \Undef}{\delrec \R \ell}
    \wedge \bigwedge_{\R \in N' \setminus N'_\Vars}
    \neg\rec{\ell = \Any \lor \Undef}{\delrec \R \ell}\\
    &\wedge \bigwedge_{\R \in P_\Vars} \rec{\fin(\R)}{\tail(\R)}
    \wedge \bigwedge_{\R \in N'_\Vars} \neg\rec{\fin(\R)}{\tail(\R)}
  \Bigg)
  \end{align*}
  where $P_\Vars = \{\R \in P \mid \tail(\R) = \rho \textand \ell \in
  \rdef(\rho)\}$, similarly for $N_\Vars$,
  and because for any atom $\R$, $\cutrow{\rectorow{\R}}{\{\ell\}} = \delrec \R \ell$.
  Then, for any $t_\ell$ it is clear that $t \leq \orec{\ell = t_\ell}$ is
  equivalent to $t.\ell \leq t_\ell$.
\end{proof}



\begin{corollary}
  Let $t \leq \orec{\ell = \Any}$ and $[\sigma_i]_{i \in I}$ be a set of substitutions.
  Then $(\bigwedge_{i \in I} t\sigma_i).\ell \leq \bigwedge_{i \in I} t.\ell\sigma$.
\end{corollary}

\subsubsection{Taming non-structural rules}

To prove soundness and completeness of the algorithmic type system, we go
through an intermediate type system, where the intersection rule is n-ary, and
the introduction of a term variable can perform a renaming of polymorphic variables.
The introduction of a renaming aims at eliminating trivial instantiations (only
renamings) in the uses of \rulename{Inst}.
More details are given by
\ifbibtex\cite[Section I.1]{castagna24b}.\else\textcite[Section I.1]{castagna24b}.\fi
\[
  \inferrule*[left=\rulename{Var},
  right={$\begin{array}{l}
      x \in \dom(\Gamma)\\[-.7mm]
      \dom(\sigma) \cap \Delta = \emptyset\\[-.7mm] \text{and $\sigma$ is a renaming}
  \end{array}$}]
  { }{\intseq \Delta \Gamma x {\Gamma(x)\sigma}}
  \quad  \inferrule*[left=\rulename{Inter},
  right={$|I| > 0$}]
  {(\intseq \Delta \Gamma e {t_i})_{i \in I}}
  {\intseq \Delta \Gamma e {\bigwedge_{i \in I} t_i}}
\]
Rules different from \rulename{Inst} and \rulename{Var} are the same as in the
declarative type system.

\begin{lemma}
  \label{l:equiv_intermediate}%
  $\intseq \Delta \Gamma e t \iff \decseq \Delta \Gamma e t$.
\end{lemma}
\begin{proof}
  The left-to-right direction is straightforward since the rules in the
  intermediate system generalize the ones of the declarative system.
  The right-to-left direction is obtained by replacing instances of
  \rulename{Var} by instances of \rulename{Var} followed by \rulename{Inst} in
  the declarative system, and by replacing occurences of n-ary intersections by
  $n-1$ \rulename{Inter} nodes.
\end{proof}

Next, we want to restrict derivations in the intermediate system to a canonical
form, where the apparition of \rulename{Inst} and \rulename{Sub} nodes is controlled.
For convenience, we introduce the following rule macro:\footnote{%
In the appendix, we also use the metavariable $s$ for types.}
\[
  \inferrule*[left=\rulename{$\leqsub{}$},right={$s \leqsub\Delta t$}]
  {\intseq \Delta \Gamma e s}
  {\intseq \Delta \Gamma e t}
\]
which is a stands for \rulename{Sub} when $s \leq t$ and otherwise for:
\begin{mathpar}
  \inferrule*[Left=\rulename{Sub}]{
    \inferrule*[Left=\rulename{Inter}]{
      \inferrule*[Left=\rulename{Inst},right={$\forall i \in I$}]{
      \intseq \Delta \Gamma e s}{\intseq \Delta \Gamma e {s\sigma_i}}
    }{\intseq \Delta \Gamma e {\bigwedge_{i \in I} s\sigma_i}}
  }{\intseq \Delta \Gamma e t}
\end{mathpar}

\begin{definition}[Canonical derivation]
  A derivation is \emph{canonical} if
  every \rulename{Inst} node it contains is part of a \rulename{$\leqsub{}$}
  pattern and
  every \rulename{$\leqsub{}$} and \rulename{Sub} nodes are either:
  \begin{itemize}
    \item The premise of an \rulename{Del} or \rulename{Sel} node, or
    \item The first premise of a \rulename{Ext} node, or
    \item One of the premises of an \rulename{Abs} or \rulename{App} node.
  \end{itemize}
\end{definition}

\begin{lemma}
  \label{l:renaming}%
  A derivation of $\intseq \Delta \Gamma e t$ can be transformed into a
  derivation $\intseq \Delta \Gamma e {t\sigma}$, for any renaming $\sigma$ such
  that $\dom(\sigma) \cap \Delta = \emptyset$, without changing the structure of
  the derivation.
\end{lemma}
\begin{proof}
  As in \cite[Lemma I.14]{castagna24b}.
\end{proof}

\begin{lemma}
  \label{l:leqsubcommute}%
  Let $I$, $\Delta$ and $t_i' \leqsub\Delta t_i$ for all $i \in I$. Then,
  $\bigwedge_{i \in I} t_i' \leqsub\Delta \bigwedge_{i \in I} t_i$.
\end{lemma}
\begin{proof}
  As in \cite[Proposition I.15]{castagna24b}.
\end{proof}

\begin{lemma}
  \label{l:extrec_leqsub}%
  Let $s' \leqsub\Delta s$, $\ell$ and $\rw$ of definition space $\Labels \setminus \ell$
  such that $\vars(s') \cap \vars(\rw) \subseteq \Delta$.
  Then, $\rec{\ell = s'}\rw \leqsub\Delta \rec{\ell = s}\rw$.
\end{lemma}
\begin{proof}
  Let $\{\sigma_i\}_{i \in I}$ such that $\bigwedge_{i \in I} s'\sigma_i \leq s$,
  with $\dom(\sigma) \subseteq \vars(s')$.
  We have $\bigwedge_{i \in I} \rec{\ell = s'}{\rw} \sigma_i
  \simeq \rec{\ell = \bigwedge_{i \in I} s'\sigma_i}{\rw}
  \leq \rec{\ell = s}{\rw}$.
\end{proof}

\begin{lemma}
  \label{l:canonical}%
  Any derivation of $\intseq \Delta \Gamma e t$ can be transformed into a
  canonical derivation of $\intseq \Delta \Gamma e {t'}$, where
  $t' \leqsub\Delta t$.
\end{lemma}
\begin{proof}
  By induction on the size of the derivation and through a case analysis on the
  root of the derivation tree used.
  \begin{description}
    \item[\rulename{Inst} or \rulename{Sub}] We remove the root and let its
      premise be the new one.
    \item[\rulename{Inter}] Let $t = \bigwedge_{i \in I} s_i$.
      By induction hypothesis, for all $i \in I$ we have derivations
      $\intseq \Delta \Gamma e {s_i'}$ with $s_i' \leqsub\Delta s_i$.
      By rule \rulename{Inter}, we have a derivation of
      $t' = \bigwedge_{i \in I} s_i'$, and $t' \leqsub\Delta t$ is verified by
      \cref{l:leqsubcommute}.
    \item[\rulename{Const}, \rulename{Var}, \rulename{Emp}] Alredy canonical.
    \item[\rulename{Ext}] Let $e = \extrecord{e_1}{\ell}{e_2}$ and
      $t = \rec{\ell = t_2}{\delrec{t_1}{\ell}}$.
      By induction hypothesis, we have canonical derivations
      $\intseq \Delta \Gamma {e_1}{t_1'}$ and $\intseq \Delta \Gamma
      {e_2}{t_2'}$ with $t_1' \leqsub\Delta t_1$ and $t_2' \leqsub\Delta t_2$.
      By rule~\rulename{$\leqsub{}$}, we derive $\intseq \Delta \Gamma
      {e_1}{t_1}$.
      Let $t' = \rec{\ell = t_2'}{\ell}{t_1}$.
      Rule~\rulename{Ext} node gives a canonical derivation of
      $\intseq \Delta \Gamma e t'$.
      We can suppose that the variables in $t_2'$ and $t_1$ are disjoint
      (otherwise we use \cref{l:renaming}), and conclude $t' \leqsub\Delta t$ by
      \cref{l:extrec_leqsub}.
    \item[\rulename{Abs}] By induction hypothesis, for each $i \in I$ we have derivations
      $\intseq{\Delta \cup \vars(t)} {\Gamma,x:t_i} e {s_i'}$ with $s_i'
      \leqsub\Delta s_i$.
      By rule \rulename{$\leqsub{}$}, we have derivations
      $\intseq{\Delta \cup \vars(t)}{\Gamma,x:t_i} e {s_i}$ and we conclude by
      rule~\rulename{Abs}.
    \item[\rulename{App}, \rulename{Sel}, \rulename{Del}] Similar to the
      previous case.
      \qedhere
  \end{description}
\end{proof}

\subsubsection{Soundness and completeness}

\soundnessalgo*
\begin{proof}
  By induction on the algorithmic typing derivation.
  By \cref{l:equiv_intermediate}, it sufices to give a derivation
  $\intseq \Delta \Gamma e t$.
  We proceed by a case analysis on the last rule used in the derivation.
  \begin{description}
    \item[\rulename{Const}, \rulename{Var}] Straightforward.
    \item[\rulename{Abs}] By induction hypothesis, for each $i \in I$ we have
      $\intseq {\Delta \cup \Delta'} {\Gamma, x: t_i} e {s_i'}$.
      Since $s_i' \leqsub{\Delta \cup \Delta'} s_i$, we derive
      $\intseq {\Delta \cup \Delta'} {\Gamma, x: t_i} e {s_i}$ by rule
      \rulename{$\leqsub{}$}.
      We conclude by rule~\rulename{Abs} in the declarative system.
    \item[\rulename{App}] By hypothesis, we have $t \in t_1 \apply\Delta t_2$,
      so $t$ is such that there are two substitution sets with
      $\bigwedge_{j \in J} t_1\sigma_j \leq \bigwedge_{i \in I}
      t_2\sigma_i \to t$.
      By induction hypothesis and \rulename{$\leqsub{}$},
      we obtain $\intseq \Delta \Gamma {e_1}{\bigwedge_{j \in J} t_1\sigma_j}$
      and $\intseq \Delta \Gamma {e_2}{\bigwedge_{i \in I}t_2\sigma_i}$.
      By~\rulename{Sub}, we have $\intseq \Delta \Gamma {e_1}{\bigwedge_{i \in
      I}t_2\sigma_i \to t}$.
      We conclude by rule~\rulename{App} in the declarative system.
    \item[\rulename{Emp}] Straightforward.
    \item[\rulename{Ext}] By induction hypothesis, we have
      $\intseq \Delta \Gamma e t$, $\intseq \Delta \Gamma {e'} {t'}$
      and sets of substitutions $[\sigma_i]_{i \in
      I}$ such that $\bigwedge_{i \in I} t\sigma_i \leq \orec{}$ and
      $\rw = \delrec{(\bigwedge_{i \in I} t\sigma_i)}\ell$.
      By rules~\rulename{Inst} and~\rulename{Inter}, we have
      $\intseq \Delta \Gamma e {\bigwedge_{i \in I} t\sigma_i}$,
      and we conclude with rule~\rulename{Del}.
    \item[\rulename{Del}] Similar to the case for~\rulename{Ext}, without the
      derivation of $t'$.
    \item[\rulename{Sel}] By induction hypothesis, we have
      $\intseq \Delta \Gamma e t$ and a set of substitutions $[\sigma_i]_{i \in
      I}$ such that $\bigwedge_{i \in I} t\sigma_i \leq \orec{\ell = \Any}$
      and $t_\ell = (\bigwedge_{i \in I} t\sigma_i).\ell$.
      By \cref{l:correct_sel}, we have $\bigwedge_{i \in I} t\sigma_i \leq
      \orec{\ell = t_\ell}$, so we conclude with rule~\rulename{$\leqsub{}$}
      and~\rulename{Sel}.
      \qedhere
  \end{description}
\end{proof}



\completenessalgo*
\begin{proof}
  By \cref{l:equiv_intermediate,l:canonical} we transform the input derivation
  into a canonical derivation $\intseq \Delta \Gamma e s$,
  where $s \leqsub\Delta t$.
  The proof is by induction on that derivation (where we use
  \rulename{$\leqsub{}$} instead of the corresponding pattern).
  In the end, we obtain $\algseq \Delta \Gamma e t'$ with $t' \leqsub\Delta s$
  and thus conclude since by transitivity of $\leqsub\Delta$, we have $t'
  \leqsub\Delta t$.
  \begin{description}
    \item[\rulename{Const}, \rulename{Var}] Straightforward.
    \item[\rulename{Abs}] By induction hypothesis, for each $i \in I$, we have
      $\algseq{\Delta \cup \Delta'}{\Gamma,x:t_i} e {s_i}$
      with $s_i' \leqsub\Delta s_i$.
      We conclude with rule~\rulename{Abs} in the algorithmic system.
    \item[\rulename{App}] By induction hypothesis, we have
      $\algseq \Delta \Gamma {e_1}{t}$
      and $\algseq \Delta \Gamma {e_2}{s}$,
      where $t \leqsub\Delta t_1 \to t_2$ and $s \leqsub\Delta t_1$.
      So $t_2 \in t \apply\Delta s$ since $(t_1 \to t_2) \mapply t_1 = t_2$.
    \item[\rulename{Emp}] Straightforward.
    \item[\rulename{Ext}] By induction hypothesis, we have
      $\algseq \Delta \Gamma e s$,
      a set of substitutions such that
      $[\sigma_i]_{i \in I} \Vdash s \leqsub\Delta t \leq \orec{\ell = \Undef}$,
      $\algseq \Delta \Gamma {e'} {s'}$ and $s' \leqsub\Delta t'$.
      Thus, we have $[\sigma_i]_{i \in I} \Vdash s \leqsub\Delta \orec{\ell = \Undef}$.
      Since $t \leq \orec{\ell = \Undef}$.
      Let $\rw = \delrec{(\bigwedge_{i \in I} t\sigma_i)}\ell$.
      By \cref{l:cutrow_correct}, we have $\rw \leq \delrec t \ell$, so that
      $\rec{\ell = s'}{\rw} \leqsub\Delta \orec{\ell = t'}{\delrec t \ell}$.
    \item[\rulename{Del}] By induction hypothesis, we have
      $\algseq \Delta \Gamma e s$ and a set of substitutions such that
      $[\sigma_i]_{i \in I} \Vdash s \leqsub\Delta t \leq \orec{}$.
      Thus, we have $[\sigma_i]_{i \in I} \Vdash s \leqsub\Delta \orec{}$.
      Let $\rw = \delrec{(\bigwedge_{i \in I} t\sigma_i)}\ell$.
      By \cref{l:cutrow_correct}, we have $\rw \leq \delrec t \ell$, so that
      $\rec{\ell = \Undef}{\rw} \leqsub\Delta \rec{\ell = \Undef}{\delrec t \ell}$.
    \item[\rulename{Sel}] By induction hypothesis, we have
      $\algseq \Delta \Gamma e s$ and a set of substitutions such that
      $[\sigma_i]_{i \in I} \Vdash s \leqsub\Delta \orec{\ell = t}$.
      Thus, we have $[\sigma_i]_{i \in I} \Vdash s \leqsub\Delta \orec{\ell =
      \Any}$.
      Let $t_\ell = (\bigwedge_{i \in I} s\sigma_i).\ell$.
      We have $t_\ell \in \sel(s)$, and by \cref{l:correct_sel}
      $t_\ell \leq t$ so $t_\ell \leqsub\Delta t$.
    \item[\rulename{$\leqsub{}$}] Straightforward by induction hypothesis and
      transitivy of $\leqsub{}$.
    \item[\rulename{Sub}] Straightforward by induction hypothesis, inclusion of
      $\leq$ in $\leqsub{}$ and transitivity of $\leqsub{}$.
    \item[\rulename{Inter}] By hypothesis, there is $I$ and derivations
      $\intseq \Delta \Gamma e {s_i}$ for all $\{s_i\}_{i \in I}$,
      with $t = \bigwedge_{i \in I} s_i$.
      Since the derivation is canonical, we know that each of these derivations
      ends with (the same kind of) structural rule.
      According to the other cases, for all $i \in I$ we have derivations
      $\intseq \Delta \Gamma e {s_i'}$, where $s_i' \leqsub\Delta s_i$.
      By \cref{l:leqsubcommute} we have $t' = \bigwedge_{i \in I} s_i'
      \leqsub\Delta \bigwedge_{i \in I} s_i = t$.
      \qedhere
  \end{description}
\end{proof}

\subsubsection{Alternative incomplete type system}
\label{sec:altrules}

\begin{figure}
  \begin{mathpar}
    \inferrule*[left=\rulename{Ext}]
    {\seq \Delta \Gamma e {t \leq \orec{\ell = \Undef}} {}
    \and \seq \Delta \Gamma {e'} {t'} {}}
    {\seq \Delta \Gamma {\extrecord{e}{\ell}{e'}}
    {\rec{\ell = t'}{\delrec t \ell}} {}}
    \and \inferrule*[left=\rulename{Del}]
    {\seq \Delta \Gamma e {t \leq \orec{}} {}}
    {\seq \Delta \Gamma {\delrecord e \ell}
    {\rec{\ell = \Undef}{\delrec t \ell}} {}}
    \and \inferrule*[left=\rulename{Sel}]
    {\seq \Delta \Gamma e t {}}
    {\seq \Delta \Gamma {e.\ell} {t.\ell} {}}
  \end{mathpar}
  The rules \rulename{Const}, \rulename{Var}, \rulename{Abs},
  \rulename{App}, \rulename{Emp} are the same as in the algorithmic system.
  \caption{Alternative algorithmic system}
  \label{fig:altrules}%
\end{figure}

\begin{lemma}
  The type system in \cref{fig:altrules} is sound with respect to the
  declarative type system.
\end{lemma}
\begin{proof}
  Similar to the proof of \cref{t:soundness_algo}.
\end{proof}

\section{APPENDIX FOR TALLYING}
\label{app:tallying}

We start with giving formal definitions of several notions mentioned in the text.

\begin{definition}[Constraints]
  A constraint $(\tterm_1,c,\tterm_2)$ is a triple such that
  $c \in \{\leq,\geq\}$ and $(\tterm_1,\tterm_2) \in
  (\Types \times \Types) \cup (\Types_\Undef \times \Types_\Undef) \cup (\Rows
  \times \Rows)$.
  $\tterm_1$ and $\tterm_2$ must be of the same kind which, in particular,
  implies $\rdef(\tterm_1) = \rdef(\tterm_2)$ if  $\tterm_1$ and $\tterm_2$ are rows.
  We denote with $\Constraints$ the set of all constraints.

  Given two sets of constraints
  $\constrset_1, \constrset_2 \subseteq \Pd(\Constraints)$, we define
  their union as $\constrset_1 \sqcup \constrset_2 = \constrset_1 \cup \constrset_2$ and
  their intersection as $\constrset_1 \sqcap \constrset_2 =
  \{C_1 \cup C_2 \mid C_1 \in \constrset_1 , C_2 \in \constrset_2\}$.
\end{definition}

\begin{definition}
  \label{def:vars_const}%
  Given a constraint-set $C \subseteq \Constraints$, the set of type row and
  field variables occurring in $C$ is defined as
  $\vars(C) = \bigcup_{(\tterm_1,c,\tterm_2) \in C} \vars(\tterm_1) \cup
  \vars(\tterm_2)$,
  with $\vars(\tterm_i)$ defined in \cref{def:vars}.
\end{definition}

\begin{definition}[Ordering]
  \label{def:ordering}%
  Let $V$ and $\Delta$ be sets of variables and $L$ a set of labels.
  An ordering $O_\vterm$ on $V$ is an injective map from $V$ to $\mathbb N$,
  an ordering $O_\ell$ on $L$ is an injective map from $L$ to $\mathbb N$.
  An ordering $O$ on $V$ and $L$ is defined as a lexicographic ordering on
  $\mathbb N \times \mathbb N$ according to $O_\vterm$ and $O_\ell$, where:
  $O(\rho.\ell) = (O_\vterm(\rho), O_\ell(\ell))$,
  $O(\cutvar\rho L) = (O_\vterm(\rho), O'(\rho,L))$ and
  $O(\vterm) = (O_\vterm(\vterm), 0))$ otherwise;
  for all $\vterm_1 \notin \Delta$ and $\vterm_2 \in \Delta$ of the same
  kind, $O(\vterm_1) < O(\vterm_2)$;
  and $O'(\rho,L)$ is an integer obtained in a canonical way from the set $L
  \cup \rdef(\rho)$ and different from any $O_\ell(\ell)$.
\end{definition}

\begin{definition}[Decomposition of row variables]
  \label{d:rowvar_decomposition}%
  We introduce two new constructors:
  (1)~$\rho.\ell$, that we treat as a field variable, and
  (2) $\cutvar \rho L$, that we treat as a row variable of definition space $\rdef(\rho){\setminus} L$.
  Substitution is extended to $(\rho.\ell)\sigma = \tau$ if $\sigma(\rho) \simeq
  \row{\ell = \tau}{\rw}{L'}$, and $(\cutvar \rho L)\sigma = \rw$ if
  $\sigma(\rho) \simeq \row{(\ell = \tau_\ell)_{\ell \in L \cap
  \rdef(\rho)}}{\rw}{L'}$.
  It is undefined if $\sigma(\rho)$ does not have any of these shapes.
\end{definition}
\noindent
Hereafter, the name $\rho$ ranges over row variables and constructors
$\cutvar {\rho} L$, the name $\fvar$  ranges over field
variables and constructors $\rho.\ell$, and the name $\vterm$ ranges
over all kinds of variables plus these new
constructors.
We consider $\cutvar \rho L$ up to the equivalence generated by
identifying $\rho \in \rowVars$ with $\cutvar \rho L$ for all $L \subseteq
\Labels{\setminus}\rdef(\rho)$,
and $\cutvar{(\cutvar \rho {L_1})}{L_2}$ with $\cutvar \rho {(L_1 \cup L_2)}$.
With an abuse of notation we write $\rho$ for the row
$\prow{}{\rho}{\Labels{\setminus}\rdef(\rho)}$---in particular in rows constraints (e.g., $(\rho, c, \rw)$)---when no confusion arises.

\begin{definition}[Constraint solution]
  \label{def:solution}%
  Let $C \subseteq \Constraints$ be a constraint-set.
  A solution to $C$ is a substitution $\sigma$ such that
  $\forall(\tterm_1,\leq,\tterm_2){\in} C\,.\, \tterm_1\sigma \leq \tterm_2\sigma$ and
  $\forall(\tterm_1,\geq,\tterm_2) {\in} C\,.\, \tterm_1\sigma \geq \tterm_2\sigma$ hold.
  If $\sigma$ is a solution to $C$, we write $\sigma \Vdash C$.
  In particular, $\tterm_1\sigma$ and $\tterm_2\sigma$ must be defined, that is,
  for all $\rho \in \dom(\sigma)$:
  \begin{itemize}
    \item[-] $\forall \ell \in \Labels.\ \rho.\ell \in \vars(C)
      \implies \exists \tau. \exists \rw. \sigma(\rho) \simeq \row{\ell = \tau} \rw
      {\Labels{\setminus}\rdef(\rho)}$, and
    \item[-] $\forall L \subseteq \Labels.\ \cutvar{\rho}{L} \in \vars(C) \implies
      \exists (\tau_\ell)_{\ell \in L}. \exists r.
      \sigma(\rho) \simeq \row{(\ell = \tau_\ell)_{\ell \in
      L}}{\rw}{\Labels{\setminus}\rdef(\rho)}$.
  \end{itemize}
 where $\vars(C)$ is the set of all type/row/field variables occurring in $C$.
\end{definition}
\noindent
The tallying algorithm
is parametric on a total order $O$ on variables, used to
ensure that the step (\oldstylenums{5}) of tallying produces
contractive types (see \cref{pr:wellformedness}).
A set of constraints is \emph{well-ordered} if for all constraints of the shape
$(\vterm, c, \tterm)$ and for all variables in $\tterm$ that are not
occurring under a type constructor,  $O(\vterm) <
O(\vterm')$ holds (we call such variables the \emph{toplevel
variables} of $T$).

\subsection{Examples Regarding the Restriction of Solutions to Atomic Rows}
\label{sec:app-atomic}

We give two examples where we discuss considering only atomic rows instead
of Boolean combinations of them in the grammar.
The first one shows that we can find tallying solutions even for unions of
records thanks to the unification technique.
The second demonstrates that this technique is not enough to recover all desired
solutions, and therefore that Boolean combinations of rows, as we adopt them in
our system, are welcome.

\begin{example}
  \label{ex:expansionrows}%
  Let us consider the types of the example on \cpageref{ex:figure}, that we
  rewrite as follows:
  $t_1 = \rec{p = \Int}{\rho} \to \rec{p = \Float}{\rho}$ and
  $t_2 = \crec{s = \texttt{"circle"}, p = \Int, d = \Float}
  \vee \crec{s = \texttt{"polygon"}, p = \Int, e = \Int}$.
  We would like to be able to unify the parameter of the function with the
  argument, even though the latter is a union.
  As explained in \cite[\S C.2.1]{castagna15}, the problem of computing $t_1
  \apply\emptyset t_2$ can be reduced to solving $\{(t', \leq, t_2 \to \alpha)\}$,
  where $\alpha$ is a fresh variable, and $t' = \bigwedge_{i \in I} t_1\sigma_i$,
  where the $\sigma_i$ are renamings of $\rho$.
  The cardinality of $I$ will be increased during the search for a solution.

  With a cardinality $|I| = 1$, we need in particular to find a solution for the
  constraint $(\rec{p = \Int}{\rho}, \geq, t_2)$.
  A component-wise unification gives the most precise solution (since we
  restrict ourselves to atomic rows):
  $\sigma(\rho) = \crow{s = \texttt{"circle"} \vee \texttt{"polygon"}, d =
  \Float \vee \Undef, e = \Int \vee \Undef}{\{s\}}$.

  This is however not the solution we want.
  Incrementing the cardinality of $I$, we now look for a solution of
  $\{((\rec{p = Int}{\rho_1} \to \rec{p = \Int}{\rho_1})
    \wedge (\rec{p = Int}{\rho_2} \to \rec{p = \Int}{\rho_2}),
  \leq, t_2 \to \alpha)\}$.
  Thus, we look in particular for a solution to the constraint
  $(\rec{p = Int}{\rho_1} \vee \rec{p = Int}{\rho_2}, \geq, t_2)$.
  A component-wise unification gives the solution
  $\sigma(\rho_1) = \crow{s = \texttt{"circle"}, d = \Float}{\{s\}},
  \sigma(\rho_2) = \crow{s = \texttt{"polygon"}, e = \Int}{\{s\}}$,
  thanks to which we retrieve the desired solution, even with the restriction
  that rows are all atomic.
\end{example}

\begin{example}[Necessity of connectives on rows]
  \label{ex:boolrows}%
  This second examples illustrates why considering only atomic rows is not satisfactory.
  Take $t = \rec{}{\rho} \to \rec{}{\rho}$ and
  $s = \orec{a = \Undef} \wedge \neg\crec{}$ and let us look for a solution of
  $t \apply\emptyset s$.
  As in the previous example, this problem can be reduced to
  solving $\{(t', \leq, s \to \alpha)\}$, where $\alpha$ is a fresh variable,
  and $t' = \bigwedge_{i \in I} t\sigma_i$, where the $\sigma_i$ are renamings of $\rho$.

  At first, we try $|I| = 1$, so we look for a solution of $\{(\rec{}{\rho} \to
  \rec{}{\rho}, \leq, \orec{a = \Undef} \wedge \neg\crec{} \to \alpha)\}$.
  After normalization, we obtain the constraint-set
  $\{(\alpha, \geq, \rec{}{\rho}),
  (\rho, \geq, \orow{a = \Undef}{\emptyset} \wedge \neg\crow{}{\emptyset})\}$.
  If we were to consider only atomic rows, we could only give the solution
  $\sigma(\rho) = \orow{}\emptyset$ and $\sigma(\alpha) = \orec{}$.

  To try to find a more precise solution, we run tallying again after
  incrementing the cardinal of $I$.
  This yields the following constraint-set:
  $\{((\rec{}{\rho_1} \to \rec{}{\rho_1}) \wedge (\rec{}{\rho_2} \to
  \rec{}{\rho_2}), \leq, s \to \alpha)\}$,
  which normalizes first to 
  $\{(\alpha, \geq, \rec{}{\rho_1} \wedge \rec{}{\rho_2}),
  (\rho_1 \vee \rho_2, \geq, \orow{a = \Undef}\emptyset \wedge \neg\crow{}\emptyset)\}$,
  and then (assuming $O(\rho_1) \leq O(\rho_2)$) to
  $\{(\alpha, \geq, \rec{}{\rho_1} \wedge \rec{}{\rho_2}),
  (\rho_1, \geq, \orow{a = \Undef}\emptyset \wedge \neg\crow{}\emptyset \wedge \neg\rho_2)\}$.
  With atomic rows only, we still do not have a satisfactory solution,
  and further expansions do not help.
\end{example}

\subsection{General Decomposition of Rows}

\begin{definition}
  \label{def:cutrow}%
  Let $\rw = \row{(\ell = \tau_\ell)_{\ell \in L_1}}{\tl}{L_2}$
  and $L$ a finite set of labels.
  We define $\cutrow \rw L =
  \row{(\ell = \tau_\ell)_{\ell \in L_1{\setminus}L}}{\tl'}{L_2 \cup L}$,
  where $\tl' = \orecsign$ if $\tl \in \Vars$ and $L \nsubseteq L_1 \cup L_2$,
  and $\tl' = \tl$ otherwise.
\end{definition}

\begin{lemma}
  \label{l:normalize_row}%
  Let $\rw$ be an atomic row and $L$ a set of labels.
  Let $\rw' = \row{L \cap \fin(\rw)}{\cutrow \rw {(L \cap \fin(\rw))}}{}$
  if $\tail(\rw) = \rho$ and $\rdef(\rho) \cap L \neq \emptyset$, and
  $\rw' = \row{L}{\cutrow \rw L}{}$ otherwise.
  \begin{enumerate}
    \item $\rw \simeq \orow{(\ell = \rw(\ell))_{\ell \in L}}{} \wedge \rw'
      \simeq \bigwedge_{\ell \in L} \orow{\ell = \rw(\ell)}{} \wedge \rw'$
    \item $\row{L \cap \fin(\rw)}{\cutrow \rw {(L \cap \fin(\rw))}}{}
      \simeq \row L {\cutrow \rw L} {} \wedge
      \row{\fin(\rw)}{\tail(\rw)}{}$ if $\tail(\rw) \in \Vars$.
  \end{enumerate}
\end{lemma}
\begin{proof}
  If $r' = \row L {\cutrow \rw L}{}$, the proof is straightforward.
  Otherwise, let $\rw = \row{(\ell = \tau_\ell)_{\ell \in
  L_1}}{\rho}{L_2}$, with $L \nsubseteq L_1 \cup L_2$.
  For the first item, since $\rw(\ell) = \Any \lor \Undef$ for any $\ell
  \notin L_1 \cap L_2$, we must show
  $\rw \simeq \orow{(\ell = \tau_\ell)_{\ell \in L \cap L_1},
  (\ell = \Any \lor \Undef)_{\ell \in L\setminus L_1}}{}
  \wedge \row{L \cap L_1, (\ell = \tau_\ell)_{\ell \in L_1 \setminus
  L}}{\rho}{}$.
  For the second item, we must show
  $\row{L \cap L_1, (\ell = \tau_\ell)_{\ell \in L_1 \setminus L}}{\rho}{}
  \simeq \orow{L,(\ell = \tau_\ell)_{\ell \in L_1{\setminus}L}}{}
  \wedge \row{\fin(\rw)}{\rho}{}$.
  Both of them are straightforward.
\end{proof}

\begin{lemma}
  \label{l:decomposition}%
  Let $P$ and $N$ be sets of atomic rows of the same definition space and $L$ be a
  finite set of labels.
  Let $P_\Vars = \{\rw \in P \mid \tail(\rw) = \rho \textand L \cap
  \rdef(\rho) \neq \emptyset\}$,
  similarly for $N_\Vars$.
  The relation $\bigwedge_{\rw \in P} \rw \leq \bigvee_{\rw \in N} \rw$
  holds \textit{iff} for every map
  $\iota: N \to L\cup\{\_\}$,
  for every $N' \subseteq \iota^{-1}(\_) \cap N_\Vars$:
  \begin{align*}
    &\left(\exists \ell \in L. \bigwedge_{\rw \in P}
    \rw(\ell) \leq \bigvee_{\rw \in \iota^{-1}(\ell)} \rw(\ell)\right)
    \textor \left(
      \bigwedge_{\rw \in P} {\cutrow \rw L} \leq
      \bigvee_{\rw \in \iota^{-1}(\_) \setminus N'}
      {\cutrow \rw L}
    \;\right)\\
    &\textor \left(
      \bigwedge_{\rw \in P_\Vars} \row{\fin(\rw)}{\tail(\rw)}{}
      \leq \bigvee_{\rw \in N'} \row{\fin(\rw)}{\tail(\rw)}{}
    \;\right)
  \end{align*}
\end{lemma}

\begin{proof}
  For each $\rw \in P_\Vars \cup N_\Vars$,
  let $\rw_{-} = \row{L \cap \fin(\rw)}{\cutrow \rw {(L \cap \fin(\rw)}}{}$
  and for each $\rw \in (P \cup N) \setminus (P_\Vars \cup N_\Vars)$,
  let $\rw_{-} = \row{L}{\cutrow \rw L}{}$.

  Using \cref{l:normalize_row}, we decompose the type
  in the statement into:
  \begin{equation}
    \label{eq:first-step-decomposition}%
    \bigwedge_{\rw \in P} (\orow{(\ell = \rw(\ell))_{\ell \in L}}{} \wedge
    \rw_{-} \wedge
    \bigwedge_{\rw \in N}
    (\bigvee_{\ell \in L} \neg\orow{\ell = \rw(\ell)}{}
    \vee \neg\rw_{-})
  \end{equation}
  We can distribute the intersection of the elements of $N$ on the right of
  \eqref{eq:first-step-decomposition} over the unions in the second brackets.
  We obtain a union of intersections of, each time, $|N|$ elements, where each
  intersection is a possible combination of the individual types present in the
  second line.
  Each combination is described by a function $\iota: N \to L \cup
  \{\_\}$, where $\iota(n) = \ell$ means that the element $\orow{\ell =
  \neg\rw(\ell)}{}$ is present in the combination given by $\iota$, while
  $\iota(n) = \_$ means that the element $\neg\rw_{-}$ is present in the
  combination.
  For each $\rw \in N$ and $\ell \in L$, let us write $\rw_\ell = \orow{\ell =
  \rw(\ell)}{}$.
  Therefore the type in~\eqref{eq:first-step-decomposition} is equivalent to:
  \begin{equation}
    \bigwedge_{\rw \in P} (\orow{(\ell = \rw(\ell))_{\ell \in L}}{} \wedge
    \rw_{-}) \wedge
    \bigvee_{\iota: N \to L \cup \{\_\}}
    (\bigwedge_{\rw \in N} \neg\rw_{\iota(\rw)})
  \end{equation}
  By distributing the intersection over the union we obtain
  \begin{equation}
    \bigvee_{\iota: N \to L \cup \{\_\}}\left(
      \bigwedge_{\rw \in P}
      (\orow{(\ell = \rw(\ell))_{\ell \in L}}{} \wedge
      \rw_{-}) \wedge
    \bigwedge_{\rw \in N} \neg\rw_{\iota(\rw)}\right)
  \end{equation}
  A union is empty if and only if each summand of the union is
  empty. Therefore the type above is empty if and only if for all $\iota: N \to L \cup \{\_\}$,
  the following type is empty:
  \begin{align*}
    &\bigwedge_{\rw \in P}
    (\orow{(\ell = \rw(\ell))_{\ell \in L}}{} \wedge
    \rw_{-}) \wedge
    \bigwedge_{\rw \in N} \neg\rw_{\iota(\rw)}\\
    &\simeq \bigwedge_{\rw \in P}
    (\orow{(\ell = \rw(\ell))_{\ell \in L}}{} \wedge
    \rw_{-}) \wedge
    \bigwedge_{\ell \in L \cup \{\_\}} \bigwedge_{\rw \in \iota^{-1}(\ell)}
    \neg\rw_\ell\\
    &\simeq \bigwedge_{\rw \in P}
    (\orow{(\ell = \rw(\ell))_{\ell \in L}}{} \wedge
    \rw_{-}) \wedge
    \bigwedge_{\rw \in \iota^{-1}(\_)} \neg\rw_{-}
    \wedge \bigwedge_{\ell \in L} \bigwedge_{\rw \in \iota^{-1}(\ell)} \neg\rw_\ell\\
    &\simeq \bigwedge_{\rw \in P}
    (\orow{(\ell = \rw(\ell))_{\ell \in L}}{} \wedge
    \rw_{-})
    \wedge \bigwedge_{\rw \in \iota^{-1}(\_)} \neg\rw_{-}
    \wedge \orow{(\ell = \textstyle{\bigwedge_{\rw \in \iota^{-1}(\ell)}}
    \neg\rw(\ell))_{\ell \in L}}{}\\
    &\simeq \orow{(\ell = {\textstyle{\bigwedge_{\rw \in P}}} \rw(\ell)
    \wedge {\textstyle{\bigwedge_{\rw \in \iota^{-1}(\ell)}}} \neg\rw(\ell))_{\ell \in L}}{}
    \wedge \bigwedge_{\rw \in P} \rw_{-}
    \wedge \bigwedge_{\rw \in \iota^{-1}(\_)} \neg\rw_{-}
  \end{align*}
  Let $\rw_1 = 
  \orow{(\ell = {\textstyle{\bigwedge_{\rw \in P}}} \rw(\ell)
    \wedge {\textstyle{\bigwedge_{\rw \in \iota^{-1}(\ell)}}}
  \neg\rw(\ell))_{\ell \in L}}{}$.
  The last type is equivalent to
  \begin{align*}
    \rw_1
    &\wedge \bigwedge_{\rw \in P} \row L {\cutrow \rw L}{}
    \wedge \bigwedge_{\rw \in P_\Vars} \row{\fin(\rw)}{\tail(\rw}{}\\
    &\wedge \bigwedge_{\rw \in \iota^{-1}(\_) \setminus N_\Vars}
    \neg\row L {\cutrow \rw L}{}
    \wedge \bigwedge_{\rw \in \iota^{-1}(\_) \cap N_\Vars}
    (\neg\row L {\cutrow \rw L}{} \vee
    \neg\row {\fin(\rw)}{\tail(\rw)}{})\\
    \simeq \rw_1
    &\wedge \bigwedge_{\rw \in P} \row L {\cutrow \rw L}{}
    \wedge \bigwedge_{\rw \in P_\Vars} \row{\fin(\rw)}{\tail(\rw}{}
    \wedge \bigwedge_{\rw \in \iota^{-1}(\_) \setminus N_\Vars}
    \neg\row L {\cutrow \rw L}{}\\
    &\wedge \bigvee_{N' \subseteq \iota^{-1}(\_) \cap N_\Vars}
    (\bigwedge_{\rw \in (\iota^{-1}(\_) \cap N_\Vars) \setminus N'}
    \neg\row L {\cutrow \rw L}{}
    \wedge \bigwedge_{\rw \in N'}
    \neg\row {\fin(\rw)}{\tail(\rw)}{})\\
    \simeq &\bigvee_{N' \subseteq \iota^{-1}(\_) \cap N_\Vars}
    \Bigg(\rw_1
      \wedge \bigwedge_{\rw \in P} \row L {\cutrow \rw L}{}
      \wedge \bigwedge_{\rw \in \iota^{-1}(\_) \setminus N'}
      \neg\row L {\cutrow \rw L}{}\\
           &\wedge \bigwedge_{\rw \in P_\Vars} \row{\fin(\rw)}{\tail(\rw)}{}
           \wedge \bigwedge_{\rw \in N'} \neg\row{\fin(\rw)}{\tail(\rw)}{}
         \Bigg)
  \end{align*}
  This type is empty if and only if the conjunctions are all empty for each
  $\iota$ and $N' \subseteq \iota^{-1}(\_) \cap N_\Vars$.
  Take $\iota$ and $N'$ and let $\rw_2 = 
  \bigwedge_{\rw \in P} \row L {\cutrow \rw L}{}
  \wedge \bigwedge_{\rw \in \iota^{-1}(\_) \setminus N'}
  \neg\row L {\cutrow \rw L}{}$
  and $\rw_3 =
  \wedge \bigwedge_{\rw \in P_\Vars} \row{\fin(\rw)}{\tail(\rw)}{}
  \wedge \bigwedge_{\rw \in N'} \neg\row{\fin(\rw)}{\tail(\rw)}{}$.
  Let $\rw_\iota = \bigwedge_{1 \leq i \leq 3} \rw_i$.
  It is immediate that $\rw_1$ is empty \textit{iff} the first condition of the
  statement holds, $\rw_2$ is empty \textit{iff} the second does,
  and $\rw_3$ is empty \textit{iff} the third does.
  We directly obtain that if one of the conditions holds, then the type
  $\rw_\iota$ is empty.
  We now show that if $\rw_\iota$ is empty, then there is $1 \leq i \leq 3$ such
  that $\rw_i$ is empty.

  For this, we suppose that none of the subtypes is empty and build an element
  $\drow \in \IntR{\TypeInter{\rw_\iota}}$.
  \begin{enumerate}
    \item Since $\rw_1$ is not empty, for all $\ell \in L$ there is an
      element $\dundef_\ell^1 \in \IntF{\TypeInter{\bigwedge_{\rw \in P} \rw(\ell) \wedge
      \bigwedge_{\rw \in \iota^{-1}(\ell)} \neg\rw(\ell)}}$.
    \item Since $\rw_2$ is not empty, there is an element
      $\domrow{(\ell = \dundef^2_\ell)_{\ell \in L_2}}{}^{V_2} \in
      \IntR{\TypeInter{\rw_2}}$.
    \item Since $\rw_3$ is not empty, there is an element $\drow_3 \in
      \IntR{\TypeInter{\rw_3}}$.
      However, the restrictions on the set of elements in
      $\IntR{\TypeInter{\rw_3}}$ only concern their tags so that any
      element $\drow'$ with $\Tag(\drow') = \Tag(\drow_3)$ and $\rdef(\drow') =
      \rdef(\drow_3)$ is in $\IntR{\TypeInter{\rw_3}}$.
      Let $V_3 = \Tag(\drow_3)$.
  \end{enumerate}
  We build the element
  $\domrow{(\ell = \dundef^1_\ell)_{\ell \in L},
  (\ell = \dundef^2_\ell)_{\ell \in L_2{\setminus}L}}{}^{V_3}$.
  This element belongs to $\IntR{\TypeInter{\rw_\iota}}$, which is a
  contradiction.
\end{proof}

\subsection{Explanation of Normalization of Constraints}

Before giving a formal description of the whole normalization algorithm as
inference rules in \cref{sec:normalization},
we describe it as a recursive function $\norm(t,M)$ that takes a type
$t$ as input and returns a set of normalized constraint-sets necessary for $t
\leq \Empty$ to hold.
The additional argument $M$ is a set of visited types, that guarantees
termination of recursion on infinitary types.
We consider the implicit constant parameter $\Delta$ containing the monomorphic variables.

We first rewrite each constraint $(\tterm, c, \tterm')$ into a set of constraints
$\{(\tterm_i, \leq, \Empty)\}_{i \in I}$, where $\tterm_i$ is a conjunction of
atoms (basic types, arrows, records), type and field variables (in the
latter case,
$\Undef$ can also be present), or their negations. This first constraint-set is obtained by transforming the type into DNF and
putting each summand of the outer union into a separate constraint.
For each constraint $(\tterm_i, \leq, \Empty)$ we then isolate the smallest
(w.r.t., the order $O$) top-level type variable or
field variable $\vterm$ not in $\Delta$, to obtain a constraint
of the form $(\vterm, c, \tterm_i')$, that gives a lower (i.e., $c$ is $\geq$ when $\vterm$ is
negated in $\tterm_i$) or an upper (i.e., $c$ is $\leq$ otherwise)
bound on that variable, where
$\tterm_i'$ is obtained from $\tterm_i$ simply by removing  $\vterm$.

There may be no variable in $\tterm_i$, or they may all be in the parameter $\Delta$.
A variable in $\Delta$ is monomorphic, cannot be instantiated and is treated as
a constant.
In that case, we erase monomorphic variables because they cannot help to satisfy
the constraint.
On basic types, we can directly see if the constraint holds.
For arrow constructors, we apply subtyping to decompose the types into
constraints on their subtypes.
Until now, all  these steps (apart from dealing with field variables) are similar
to those of the existing tallying algorithm for type variables by~\citet{castagna15}.
For a conjunction of record atoms, the technique is more elaborated.

The normalization of a conjunction of records $t$ is defined as the
normalization of its underlying row through an auxiliary procedure on rows:
\[
  \norm(t,M) = \rownorm(\rectorow t, M \cup \{t\})
\]

The main technical part of the normalization step of tallying is defining this
auxiliary procedure.
It uses the following two operators.
\begin{definition}
  Let $\rw = \row{(\ell = \tau_\ell)_{\ell \in L_1}}{\tl}{L_2}$, $\Delta$ be a
  set of variables, $\ell \in \Labels{\setminus}L_2$ and $L \in \Pf(\Labels)$.
  \begin{itemize}
    \item $\rw[\ell] \eqdef \rho.\ell$ if $\tl = \rho \notin \Delta$ and $\ell
      \in \rdef(\rho)$, and $\rw[\ell] \eqdef \rw(\ell)$ otherwise.
    \item $\cutrowbis {\row{(\ell = \tau_\ell)_{\ell \in L_1}}{\tl}{L_2}} L \Delta
      \eqdef \row{(\ell = \tau_\ell)_{\ell \in L_1{\setminus}L}}{\tl'}{L_2 \cup L}$;
      where if $\tl = \rho$ and $\rdef(\rho) \cap L \neq \emptyset$:
      $\tl' = \orecsign$ if $\rho \in \Delta$ and $\tl' = \cutvar \rho L$
      otherwise; and $\tl' = \tl$ otherwise.
  \end{itemize}
\end{definition}

Let us consider a row in DNF
$\rw_0 = \bigwedge_{\rw \in P} \rw \wedge \bigwedge_{\rw \in N} \neg\rw$.
Let $\rho_0$ be the smallest top-level variable of $\rw_0$ not in $\Delta$ and
$L = \rdef(\rw_0){\setminus}\rdef(\rho_0)$,
or $L = \emptyset$ if no such variable exists.
Let $S_\Delta = \{\rw \in S \mid \tail(\rw) = \rho \in \Delta \textand
\rdef(\rho) \cap L \neq \emptyset\}$, for $S = P,N$.
Normalization $\rownorm(\rw_0,M)$ is defined as:
\begin{equation*}
  \bigsqcap_{\iota: N \to L \cup \{\_\}}
  \Bigg(
      \bigsqcup_{\ell \in L}
      \fieldnorm\Big(
        \bigwedge_{\rw \in P} \rw[\ell]
        \wedge\negspace \bigwedge_{\rw \in \iota^{-1}(\ell)} \negspace\neg{\rw[\ell]},
      M\Big)
      \sqcup
      \hspace{-0.2em}\bigsqcap_{N' \in \mathcal N}
      \Bigg(
      \tailnorm\Big(
        \bigwedge_{\rw \in P} (\cutrowbis \rw L \Delta)
        \wedge\negspace \bigwedge_{\rw \in \iota^{-1}(\_) \setminus N' \in S} \negspace\negspace\neg (\cutrowbis \rw L \Delta),
      M\Big)
    \Bigg)
\end{equation*}
where $\mathcal N = \{N' \subseteq \iota^{-1}(\_) \cap N_\Delta
  \mid \bigwedge_{\rw \in P_\Delta} \prow{\fin(\rw)}{\tail(\rw)}{} \nleq
\bigvee_{\rw \in N'} \prow{\fin(\rw)}{\tail(\rw)}{}\}$.

Here, the idea of the algorithm is the following.
If there is no polymorphic top-level variable (these can only be row variables),
we let $L = \emptyset$, which does not decompose the row, and only calls the
function $\tailnorm(\rw_0,M)$ recursively.
This function decomposes the record over the whole set of
top-level labels using the formula for subtyping.
We do not loose solutions doing so: due to the absence of top-level polymorphic
variables, the emptiness of $\rw_0$ can only be satisfied by the emptiness of
one of the components.

If there is a smallest polymorphic top-level variable $\rho_0$, we take $L =
\rdef(\rw_0) \setminus \rdef(\rho_0)$.
In other words, we take $L$ to be the set of labels on the left of the variable
$\rho_0$.
In this way, we decompose the row on one side over the elements in $L$, which
does not affect $\rho_0$.
These elements are handled by the auxiliary function $\fieldnorm$, that
normalizes constraints on fields in a way similar to what is done on types.
On the other side, we obtain constraints over $\rho_0$, where $\rho_0$ is in a
row of the shape $\row{}{\rho_0}{\Labels{\setminus}\rdef(\rho_0)}$.
Then, the recursive call to $\tailnorm$ singles out the variable $\rho_0$ in
order to obtain an upper or lower bound for $\rho_0$, as we do for type and
field variables (the definitions of $\fieldnorm$ and $\tailnorm$ are given below).

While $\rho_0$ is not affected by the decomposition thanks to our choice of $L$,
there can be polymorphic top-level variables $\rho'$ such that $\rdef(\rho') \cap
L \neq \emptyset$. This is why we have introduced the two new operators
$\rw[\ell]$ and $\cutrowbis \rw L \Delta$.

\paragraph{Fields}

A field-type is always equivalent to a DNF that is a disjunction of conjunctions
of either of the shape
$\tau \wedge \bigwedge_{\fvar \in P} \fvar \wedge \bigwedge_{\fvar \in N}
\fvar$, where $\tau$ is either $\Undef$ or a type $t$.
If there is a smallest variable $\fvar_0 \in P \cup N$ not in $\Delta$, we
single out this variable, in the same way as is done for type variables in our
algorithm and in \cite{castagna15}.
If all top-level field variables are monomorphic:
\begin{itemize}
  \item If $\tau = t$, $\tau$ can be instantiated to an empty type only if $t$
    can, so we apply $\norm(t,M)$.
  \item If $\tau = \Undef$, $\tau$ can never be instantiated to an empty type
    since $\Undef \nleq \Empty$, so normalization fails.
\end{itemize}
We use the notation $\vterm \lightning \tterm$ to indicate that $\vterm$ is the
smallest top-level variable in $\tterm$.
\[
  \fieldnorm(\tau,M) =
  \begin{cases}
    \{\{(\fvar_0, \leq, \neg\tau \vee \bigvee_{\fvar \in P{\setminus}\{\fvar_0\}}
    \neg\fvar \vee \bigvee_{\fvar \in N} \fvar)\}\},
    &\text{ if } \exists \fvar_0 \in P.
    \fvar_0 \lightning \tau \\
    \{\{(\tau \wedge \bigwedge_{\fvar \in P} \fvar \wedge \bigwedge_{\fvar \in
    N{\setminus}\{\fvar_0\}} \neg\fvar \leq \fvar_0)\}\},
    &\text{ if } \exists \fvar_0 \in N. \fvar_0 \lightning \tau \\
    \norm(t,M),
    &\hspace{-5em}\text{ if } \tau' = t \textand (P \cup N) \setminus \Delta = \emptyset\\
    \emptyset,
    &\hspace{-5em}\text{ if } \tau' = \Undef \textand (P \cup N) \setminus \Delta = \emptyset
  \end{cases}
\]

\paragraph{Tails}

By design, the input of this function is a row such that:
\begin{itemize}
  \item Either there is a polymorphic top-level variable that is an a row $\rw_0$
    with $\fin(\rw) = \emptyset$.
    Then, we single out this variable on the left of a new constraint.
  \item Or there is no polymorphic top-level row variable.
    In that case, we decompose the row over all the labels using the subtyping
    formula.
\end{itemize}
\[
  \tailnorm(\rw, M) =
  \begin{cases}
    \{\{(
        \rho_0, \leq,
        \bigvee_{\rw \in P{\setminus}\{\rw_0\}} \neg \rw \vee
        \bigvee_{\rw \in N} \rw
    )\}\}, &\text{ if } \rw_0 \in P \\
    \{\{(
        \rho_0, \geq,
        \bigwedge_{\rw \in P} \rw \wedge_{\rw \in N{\setminus}\{\rho\}} \neg \rw
)\}\}, &\text{ if } \rw_0 \in N \\
\bigsqcap_{\iota \in I}
      \bigsqcup_{\ell \in L}
      \fieldnorm\Big(
        \bigwedge_{\rw \in P} \rw(\ell)
        \wedge \bigwedge_{\rw \in \iota^{-1}(\ell)} \neg\rw(\ell),
        M \Big), &\text{if } \tlv(\rw) \subseteq \Delta
  \end{cases}
\]
where in the two first cases, there is $\rho_0 \in (P \cup N) \setminus \Delta$
such that there is $\rw_0 \in P \cup N$ with $\rw_0 =
\row{}{\rho_0}{\Labels{\setminus}\rdef(\rw)}$,
and where in the last case:
\begin{itemize}
  \item $L = \bigcup_{\rw \in P \cup N} \fin(\rw)$;
  \item $P_\Delta = \{\rw \in P \mid \tail(\rw) \in \Vars\}$
    and $N_\Delta = \{\rw \in N \mid \tail(\rw) \in \Vars\}$;
  \item $I = \{\iota: N \to L \cup \{\_\} \mid
      (\forall \rw_\circ \in \iota^{-1}(\_) \setminus N_\Delta.
    \bigwedge_{\rw \in P} \Def(\rw) \nleq \Def(\rw_\circ))
    \textand (\forall \rw_0 \in \iota^{-1}(\_) \cap N_\Delta.
    \forall \rw_p \in P_\Delta. \tail(\rw_\circ) \neq \tail(\rw_p))
  \}$,
  where $\Def(\rw)$ is defined as in \cref{l:subtyping}.
\end{itemize}

\subsection{Constraint Normalization}
\label{sec:normalization}

\begin{figure}[p!]
\begin{framed}
  \begin{mathpar}
    \inferrule*[right=\rulename{Nempty}]{ }{\normseq \Sigma \emptyset {\{\emptyset\}}}
    \and \inferrule*[right=\rulename{Njoin}]
    {(\normseq \Sigma {\{(\tterm_i,c_i,\tterm_i')\}} {\constrset_i})_{i \in I}}
    {\normseq \Sigma {\{(\tterm_i,c_i,\tterm_i') \mid i \in I\}} {\sqcap_{i \in I} \constrset_i}}
    \and \inferrule*[right=\rulename{Nsym}]
    {\normseq \Sigma {(\tterm,\geq,\tterm')} \constrset}
    {\normseq \Sigma {(\tterm',\leq,\tterm)} \constrset}
    \and \inferrule*[right=\rulename{Nzero}]
    {\normseq \Sigma {\{(\tterm \wedge \neg \tterm',\leq,\Empty)\}} \constrset
    \and \tterm' \neq \Empty}
    {\normseq \Sigma {\{(\tterm,\leq,\tterm')\}} \constrset}{}
    \and \inferrule*[right=\rulename{Ndnf}]
    {\normseq \Sigma {\{(\dnf(\tterm),\leq,\Empty)\}} \constrset
    \and \hspace{-1em} \tterm \neq \dnf(\tterm)}
    {\normseq \Sigma {\{(\tterm,\leq,\Empty)\}} \constrset}{}
    \and \inferrule*[right=\rulename{Nunion}]
    {\normseq \Sigma {\{(\tterm_i,\leq,\Empty) \mid i \in I\}} \constrset}
    {\normseq \Sigma {\{(\vee_{i \in I}\,\tterm_i,\leq,\Empty)\}} \constrset}{}
  \end{mathpar}
  {\small Where $\tterm_i$ in~\rulename{Nunion} are single normal forms.}
\end{framed}
  \caption{Normalization rules for all kinds}
  \label{fig:nrules-all}
\end{figure}

\begin{figure}[p!]
  \begin{framed}
  \begin{mathpar}
    \inferrule*[right=\rulename{Nhyp}]
    {t \in \Sigma \and \tlv(t) = \emptyset}
    {\normseq \Sigma {\{(t, \leq, \Empty)\}} {\{\emptyset\}}}
    \and \inferrule*[right=\rulename{Nassum}]
    {\normseqbis {\Sigma \cup \{t\}} {\{(t, \leq, \Empty)\}} \constrset \and t \notin \Sigma}
    {\normseq \Sigma {\{(t, \leq, \Empty)\}} \constrset}
    \and \inferrule*[right=\rulename{Ntlv}]
    {\tlv(\tterm) = \emptyset
      \and \vterm' \lightning_O P \cup N
      \and \constrset =
      {\begin{cases}
          \{\single(\vterm',\tterm_0)\} &\vterm' \notin \Delta \\
          \normseq \Sigma {\{(\tterm,\leq,\Empty)\}} {} &\vterm' \in \Delta
    \end{cases}}}
    {\normseq \Sigma
      {\{(\tterm_0 = \bigwedge_{\vterm \in P} \vterm
      \wedge \bigwedge_{\vterm \in N} \neg \vterm \wedge \tterm, \leq, \Empty)\}}
    \constrset}
    \and \inferrule*[right=\rulename{Nopt}]
    { }
    {\normseq \Sigma {\{(\Undef,\leq,\Empty)\}} \emptyset}
    \and \inferrule*[right=\rulename{Nbasic}]
    {\constrset = {\begin{cases}
          \{\emptyset\} &\text{if } t, \leq, \Empty \\
          \emptyset &\text{if } t \nleq \Empty
    \end{cases}}}
    {\normseqbis \Sigma
      {\{(t = \bigwedge_{i \in P} b_i \wedge \bigwedge_{j \in N} \neg b_j, \leq, \Empty)\}}
    {\constrset}}
    \and \inferrule*[right=\rulename{Narrow}]
    {\exists j \in N. \forall P' \subseteq P.
      {\begin{cases}
          \normseq \Sigma
          {\{t_j^1 \wedge \bigwedge_{i \in P'} \neg t_i^1, \leq, \Empty\}}
          {\constrset_{P'}^1}\\
          {\begin{cases}
              \normseq \Sigma
              {\{\bigwedge_{i \in P{\setminus}P'} t_i^2 \wedge \neg t_j^2, \leq, \Empty\}}
              {\constrset_{P'}^2}
        &P' \neq P\\
        \constrset_{N'}^2 = \emptyset
        &\text{otherwise}
          \end{cases}}
    \end{cases}}}
    {\normseqbis \Sigma
      {\{(\bigwedge_{i \in P} (t_i^1 \to t_i^2)
      \wedge \bigwedge_{j \in N} \neg(t_j^1 \to t_j^2), \leq, \Empty)\}}
      {\bigsqcup_{j \in N} \bigsqcap_{P' \subseteq P} (\constrset_{P'}^1 \sqcup
    \constrset_{P'}^2)}}
    \and \inferrule*[right=\rulename{Nrec}]
    {\normseq \Sigma
      {\{(\rectorow{\bigwedge_{\R \in P} \R \wedge \bigwedge_{\R \in
      N} \neg\R}, \leq, \Empty)\}}
    \constrset}
    {\normseqbis \Sigma
      {\{(\bigwedge_{\R \in P} \R \wedge \bigwedge_{\R \in N} \neg\R, \leq, \Empty)\}}
    \constrset}
  \end{mathpar}
\end{framed}
  \caption{Normalization rules for type and field single normal forms}
  \label{fig:nrules-typ}
\end{figure}

\begin{figure}[p!]
  \begin{framed}
  \begin{mathpar}
    \inferrule*[right=\rulename{Nrow}]
    {\forall \iota:N {\to} L {\cup} \{\_\}.
      {\begin{cases}
          \forall \ell \in L. \normseq \Sigma
          {\{(\bigwedge_{\rw \in P} \rw[\ell]
              \wedge \bigwedge_{\rw \in \iota^{-1}(\ell)} \neg\rw[\ell],
          \leq, \Empty)\}}
          \constrset_\ell^\iota\\
          \forall N' \in \mathcal N.
          \normseq \Sigma
          {\{(\bigwedge_{\rw \in P} \cutrowbis \rw L \Delta
              \wedge \bigwedge_{\rw \in \iota^{-1}(\_){\setminus}N'}
          \neg(\cutrowbis \rw L \Delta), \leq, \Empty)\}}
          \constrset_{N'}^\iota
      \end{cases}} \\
      \rho \lightning_O \tlv(\rw_0)
      \and \rho \notin \Delta
      \and L = \rdef(\rw_0){\setminus}\rdef(\rho)
      \and L \neq \emptyset\\
      P_\Delta = \{\rw \in P \cap \Delta \mid \tail(\rw) = \rho_p \textand \rdef(\rho_p) \cap
      L \neq \emptyset\}\\
      N_\Delta = \{\rw \in N \cap \Delta \mid \tail(\rw) = \rho_n \textand
      \rdef(\rho_n) \cap L \neq \emptyset\}\\
      \mathcal N = \{N' \subseteq \iota^{-1}(\_) \cap N_\Delta
        \mid \bigwedge_{\rw \in P_\Delta} \prow{\fin(\rw)}{\tail(\rw)}{} \wedge
        \bigwedge_{\rw \in N'} \neg\prow{\fin(\rw)}{\tail(\rw)}{}
      \nleq \Empty\}
    }
    {
      \normseq \Sigma
      {\{(\rw_0 = \bigwedge_{\rw \in P} \rw
      \wedge \bigwedge_{\rw \in N} \neg\rw, \leq, \Empty)\}}
      {\bigsqcap_{\iota:N \to L \cup \{\_\}}
        \left(\bigsqcup_{\ell \in L} \constrset_\ell^\iota
          \sqcup \bigsqcap_{N' \in \mathcal N}
          \constrset_{N'}^\iota
    \right)} }
    \and \inferrule*[right=\rulename{Ntail-mono}]
    {
      \forall \iota \in I_1 \cap I_2. \forall \ell \in L.
      \normseq \Sigma
      {\{(\bigwedge_{\rw \in P} \rw(\ell)
          \wedge \bigwedge_{\rw \in \iota^{-1}(\ell)} \neg\rw(\ell)
     , \leq, \Empty)\}}
      {\constrset^\iota_\ell} \\
      I_1 = \{\iota:N \to L \cup \{\_\} \mid
        \forall \rw_\circ \in \iota^{-1}(\_){\setminus}N_\Delta.
        \bigwedge_{\rw \in P{\setminus}P_\Delta} \Def(\rw) \nleq
      \Def(\rw_\circ)\} \\
      I_2 = \{\iota:N \to L \cup \{\_\} \mid
        \forall \rw_\circ \in \iota^{-1}(\_) \cap N_\Delta.
        \forall \rw_p \in P_\Delta.
      \tail(\rw_\circ) \neq \tail(\rw_p)\}\\
      \tlv(P \cup N) \subseteq \Delta
      \and L = \bigcup_{\rw \in P \cup N} \fin(\rw)\\
      P_\Delta = \{\rw \in P \mid \tail(\rw) \in \Vars\}
      \and N_\Delta = \{\rw \in N \mid \tail(\rw) \in \Vars\}
    }
    {\normseq \Sigma {\{(\bigwedge_{\rw \in P} \rw \wedge \bigwedge_{\rw \in N}
      \neg\rw, \leq, \Empty)\}}
    {\bigsqcap_{\iota \in I_1 \cap I_2} \bigsqcup_{\ell \in L} \constrset_\ell^\iota}}
    \and \inferrule*[right=\rulename{Ntail-tlv}]
    {
      \rho \lightning_O \tlv(P \cup N)
      \and \rho \notin \Delta
      \and \exists\rw' \in P \cup N.
      \rw' = \prow {} \rho L
    }
    {\normseq \Sigma
      {\{(\rw_0 = \bigwedge_{\rw \in P} \rw \wedge \bigwedge_{\rw \in N} \neg\rw, \leq, \Empty)\}}
      {\{\single(\row{} \rho L,\rw_0)\}}
    }
  \end{mathpar}
\end{framed}
  \caption{Normalization rules for row single normal forms}
  \label{fig:nrules-rec}
\end{figure}

We formalize normalization as a judgment $\normseq \Sigma C \constrset$, which
states that under the environment $\Sigma$ (which, informally, contains the
types that have already been processed at this point), $C$ is normalized to
$\constrset$.
The main judgment is derived according to the rules in
\cref{fig:nrules-all,fig:nrules-typ,fig:nrules-rec}.
Given a type, field or row variable $\vterm$ and a conjunction of types, field-types or rows respectively, we define
$\single(\vterm, \tterm \wedge \vterm) = \{(\vterm, \leq, \neg\tterm)\}$ and
$\single(\vterm, \tterm \wedge \neg\vterm) = \{(\vterm, \geq, \tterm)\}$.
We call single normal form a DNF that has no topmost disjunction.

If $\normseq \emptyset C \constrset$, then $\constrset$ is the result of the
normalization of $C$.
We now prove soundness and termination of the constraint normalization
algorithm.

\begin{definition}
  We define the family $(\leq_n)_{n \in \mathbb N}$ of subtyping relations as \[
    t \leq_n s \iffdef \forall\eta.\TypeInter{t}_n \eta \subseteq \TypeInter{s}_n \eta
  \] where $\TypeInter{\cdot}_n$ is the rank $n$ interpretation of a type,
  defined as \[
    \TypeInter{t}_n \eta =
    \{ d \in \TypeInter{t}_\eta \mid \domheight(d) \leq n \}
  \] and $\domheight(d)$ is the height of an element in $\Domain$, defined as
  \begin{align*}
    \dots\\
    \domheight(\domrec{\drow}^V)
    &= 1 + \domheight(\drow)\\
    \domheight(\domrow{(\ell = \delta_\ell)_{\ell \in L}}{L'}^V)
    &= \max(1, (\domheight(\delta_\ell))_{\ell \in L})
  \end{align*}
\end{definition}

\begin{lemma}
  \label{l:subtyping_rectorow_levels}%
  Let $t \leq \orec{}$. Then,
  $t \leq_{n+1} \Empty \iff \rectorow t \leq_n \Empty$.
\end{lemma}
\begin{proof}
  By definition and a trivial well-founded
  induction on type operators, we have
  $\IntQ{\TypeInter{t}_\eta} =
  \{\domrec\drow^V \mid \dbar \in \IntQ{\TypeInter{\rectorow t}_\eta}\}$.
  Thus, by definition of height, we have
  $\IntQ{\TypeInter{t}_{n+1}}\eta =
  \{\domrec\drow^V \mid \dbar \in
  \IntQ{\TypeInter{\rectorow t}_n}\eta\}$.
\end{proof}

\begin{definition}
  Given a constraint-set $C$ and a substitution $\sigma$, we define the rank $n$
  satisfaction predicate $\Vdash_n$ as \[
    \sigma \Vdash_n C \iffdef
    \forall(\tterm_1,\leq,\tterm_2) \in C. \tterm_1 \leq_n \tterm_2 \textand
    \forall(\tterm_1,\geq,\tterm_2) \in C. \tterm_1 \geq_n \tterm_2
  \]
\end{definition}

\begin{lemma}
  \begin{enumerate}
    \item $\sigma \Vdash_0 C$ for all $\sigma$ and $C$.
    \item $\sigma \Vdash C \iff \forall n. \sigma \Vdash_n C$.
  \end{enumerate}
\end{lemma}
\begin{proof}
  Consequence of \cite[Lemma C.7]{castagna15} and
  \cref{l:subtyping_rectorow_levels}.
\end{proof}

\begin{definition}[Marshalling]
  \label{def:marshall}%
  Given a conjunction of atomic rows
  $\rw_0 = \bigwedge_{\rw \in P} \rw \wedge \bigwedge_{\rw \in N} \rw$,
  a finite set of labels and a set of variables, we define marshalling as
  \[
    \marshall(\rw, L, \Delta) =
    \bigwedge_{\rw \in P{\setminus}P_0} \rw
    \wedge \bigwedge_{\rw \in N{\setminus}N_0} \neg\rw
    \wedge \bigwedge_{\rw \in P_0}
    \row{(\ell = \rw[\ell])_{\ell \in L}}{\cutrowbis \rw L \Delta}{}
    \wedge \bigwedge_{\rw \in N_0}
    \neg\row{(\ell = \rw[\ell])_{\ell \in L}}{\cutrowbis \rw L \Delta}{}
  \]
  where $P_0 = \{\rw \in P \mid \tail(\rw) = \rho \notin \Delta \textand
  \rdef(\rho) \cap L \neq \emptyset\}$, similarly for $N_0$.
\end{definition}

\begin{lemma}
  \label{l:marshall}%
      If $\sigma \Vdash \{(\marshall(\rw,L,\Delta),\leq,\Empty\}$, then
      $\sigma \Vdash \{(\rw,\leq,\Empty)\}$.
\end{lemma}
\begin{proof}
  We show that $\marshall(\rw_0,L,\Delta)\sigma = \rw_0\sigma$.
  Let $\rw = \row{(\ell = \rw(\ell))_{\ell \in L_1}}{\rho}{} \in P_0 \cup N_0$
  and $\rw' = \row{(\ell = \rw[\ell]_{\ell \in L}} {\cutrowbis \rw L
  \Delta}{}$.
  Then, $\rw' = \row{(\ell = \rw(\ell)_{\ell \in L_1},
  (\ell = \rho.\ell)_{\ell \in L{\setminus}L_1}}
  {\cutvar \rho L}{}$.
  We have $\sigma(\rho) \simeq \row{(\ell = \tau_\ell)_{\ell \in
  L{\setminus}L_1}}{\rw''}{\Labels{\setminus}\rdef(\rho)}$
  because $\sigma$ is a solution for
  $\{\marshall(\rw_0,L,\Delta),\leq,\Empty\}$.
  Therefore, $\rw\sigma = \marshall(\rw,L,\Delta)\sigma$.
\end{proof}

Given a set $\Sigma$ of types and rows, we write $C(\Sigma)$ for the
constraint-set $\{(\tterm,\leq,\Empty) \mid \tterm \in \Sigma\}$.

\begin{lemma}[Soundness]
  \label{l:soundness-normalization}%
  Let $C$ be a constraint-set.
  If $\normseq \emptyset C \constrset$, then for all constraint-set
  $C' \in \constrset$ and all substitution $\sigma$, we have
  $\sigma \Vdash C' \implies \sigma \Vdash C$.
\end{lemma}
\begin{proof}
  We prove the following stronger statements.
  \begin{enumerate}
    \item Assume $\normseq \Sigma C \constrset$.
      For all $C' \in \constrset$, $\sigma$ and $n$,
      if $\sigma \Vdash_n C(\Sigma)$ and $\sigma \Vdash_n C'$,
      then $\sigma \Vdash_n C$.
    \item Assume $\normseqbis \Sigma C \constrset$.
      For all $C' \in \constrset$, $\sigma$ and $n$,
      if $\sigma \Vdash_n C(\Sigma)$ and $\sigma \Vdash_n C'$,
      then $\sigma \Vdash_{n+1} C$.
  \end{enumerate}
  The cases \rulename{Nempty}, \rulename{Njoin}, \rulename{Nsym},
  \rulename{Nzero}, \rulename{Ndnf}, \rulename{Nunion}, \rulename{Nhyp},
  \rulename{Nassum}, \rulename{Nbasic} and \rulename{Narrow} are given
  in~\cite[Lemma C.10]{castagna15}.
  \begin{description}
    \item[\rulename{Ntlv}]
      Let $\tterm_0 = \tterm \wedge \bigwedge_{\vterm \in P} \vterm \wedge
      \bigwedge_{\vterm \in N} \neg\vterm$
      and $\vterm'$ be the smallest type variable with respect to the
      order in $P \cup N$.
      If $\vterm' \in P \setminus \Delta$, then we have
      $\sigma \Vdash_n \{(\vterm', \leq, \neg \tterm_{\vterm'})\}$
      with $\tterm_{\vterm'} = \tterm \wedge \bigwedge_{\vterm \in P{\setminus}\vterm}
      \vterm \wedge \bigwedge_{\vterm \in N} \neg\vterm$,
      thus $(\vterm')\sigma \leq_n \neg \tterm_{\vterm'}\sigma$.
      This is equivalent to $\tterm\sigma \leq_n \Empty$,
      and we conclude $\sigma \Vdash_n \{(\tterm, \leq, \Empty)\}$.
      If $\vterm' \in N \setminus \Delta$, the result follows as well.
      If $\vterm' \in \Delta$, then $P \cup N \subseteq \Delta$ by
      \cref{def:ordering}.
      We have $\TypeInter{\tterm_0}
      = \{\dterm \in \TypeInter{\tterm} \mid P \subseteq \Tag(\dterm) \textand N \cap
      \Tag(\dterm) = \emptyset\}$.
      Since $\tterm_0$ is non empty, the variables in $P$ and $N$ are different, and
      since those variables cannot be instantiated,
      we can satisfy $\tterm_0 \leq \Empty$ if and only if $\tterm \leq \Empty$
      is satisfied.
    \item[\rulename{Nopt}] Direct by emptiness of $\constrset$.
    \item[\rulename{Nrec}]
      Let $t = \bigwedge_{\R \in P} \R \wedge \bigwedge_{\R \in N} \neg\R$.
      By induction, we have $\sigma \Vdash_n \{(\rectorow{t} \leq \Empty)\}$.
      This is by definition equivalent to $\rectorow{t}\sigma \leq_n \Empty$,
      thus $\rectorow{t\sigma} \leq_n \Empty$
      and by \cref{l:subtyping_rectorow_levels} $t\sigma \leq_{n+1} \Empty$,
      which concludes $\sigma \Vdash_{n+1} \{(t \leq \Empty)\}$.
    \item[\rulename{Nrow}]
      The result is direct if
      $\constrset = \emptyset$.
      Otherwise, we have $C' = \bigcup_{\iota \in N \to L \cup \{\_\}} C_\iota'$,
      where $C_\iota' \in \bigsqcup_{\ell \in L} \constrset_\ell^\iota
      \sqcup \bigsqcap_{N' \in \mathcal N}
      \constrset_{N'}^\iota$.
      For all $\iota:N \to L \cup \{\_\}$,
      there are two cases.
      \begin{enumerate}
        \item In the first case, there is $\ell \in L$ such that $C_\iota' \in \constrset_\iota^\ell$.
          Then, by induction, \[
            \sigma \Vdash_n \{(\bigwedge_{\rw \in P} \rw[\ell]
              \wedge \bigwedge_{\rw \in \iota^{-1}(\ell)} \neg\rw[\ell]
          \leq \Empty)\} \]
        \item In the second case, we have
          $C_\iota' = \bigcup_{N'} C_\iota^{N'}$,
          where $C_\iota^{N'} \in \constrset^\iota_{N'}$.
          For all $N'$, there are two subcases.
          \begin{enumerate}
            \item In the first subcase, $N' \in \mathcal N$.
              Then, by induction, \[
                \sigma \Vdash_n \{(\bigwedge_{\rw \in P}
                  \cutrowbis \rw L \Delta
                  \wedge \bigwedge_{\rw \in \iota^{-1}(\_){\setminus}N'}
                  \neg(\cutrowbis \rw L \Delta)
              \leq \Empty)\} \]
            \item In the second subcase, $N' \in \iota^{-1}(\_) \cap N_\Delta \setminus \mathcal N$.
              Then we have \[
                \sigma \Vdash_n {\{(\bigwedge_{\rw \in P_\Delta}
                    \row{\fin(\rw)}{\tail(\rw)}{}
                    \wedge \bigwedge_{\rw \in N'}
                    \neg\row{\fin(\rw)}{\tail(\rw)}{},
              \leq, \Empty)\}} \]
          \end{enumerate}
      \end{enumerate}
      In other words,
      $\forall \iota: N \to L \cup \{\_\}.
      \forall N' \subseteq \iota^{-1}(\_) \cap N_\Delta:$
      \begin{align*}
        &\displaystyle\left(\exists \ell \in L. \bigwedge_{\rw \in P}
        \rw[\ell]\sigma \leq_n\!\! \bigvee_{\rw \in \iota^{-1}(\ell)} \rw[\ell]\sigma\right)
        \enspace\textor\enspace \displaystyle\left(
          \bigwedge_{\rw \in P} (\cutrowbis \rw L \Delta)\sigma
          \leq_n\negthickspace \bigvee_{\rw \in
          \iota^{-1}(\_) \setminus N'}\negthickspace\negthickspace
          (\cutrowbis \rw L \Delta)\sigma
        \;\right) \\
        &\textor \left(\bigwedge_{\rw \in P_\Delta}
          \row{\fin(\rw)}{\tail(\rw)}{}
          \wedge \bigwedge_{\rw \in N'}
        \neg\row{\fin(\rw)}{\tail(\rw)}{} \right)
      \end{align*}
      Let $P_0 = \{\rw \in P \mid \tail(\rw) = \rho \notin \Delta
      \textand \rdef(\rho) \cap L \neq \emptyset\}$,
      same for $N_0$.
      By \cref{l:decomposition}, the definition of substitution
      and the fact that $\rw[\ell] = \rw(\ell)$ and $\cutrowbis \rw L
      \Delta = \cutrow \rw L$ for all $\rw \notin
      P_0 \cup N_0$,
      we have $\marshall(\rw_0,L,\Delta)\sigma \leq_n \Empty$,
      that is $\sigma \Vdash \{\marshall(\rw_0,L,\Delta),\leq,\Empty)\}$,
      By \cref{l:marshall}, we conclude $\sigma \Vdash_n \{(\rw_0 \leq \Empty)\}$.
    \item[\rulename{Ntail-mono}]
      The result is direct if $\bigsqcap_{\iota \in I} \bigsqcup_{\ell \in L}
      \constrset_\ell^\iota = \emptyset$.
      Otherwise, we have $C' = \bigcup_{\iota \in I} C_\iota'$,
      where $\constrset_\iota' \in \bigsqcup_{\ell \in L} \constrset_\iota^\ell$.
      By definition of $\sqcup$, for all $\iota \in I$, there is $\ell \in L$
      such that $C_\iota' \in \constrset_\iota^\ell$.
      By induction, $\sigma \Vdash_n \{(\bigwedge_{\rw \in P} \rw(\ell) \wedge
      \bigwedge_{\rw \in \iota^{-1}(\ell)} \neg\rw(\ell))\}$.
      In other words, \[
        \forall \iota \in I. \exists \ell \in L.
        \bigwedge_{\rw \in P} \rw(\ell)\sigma
        \wedge \bigwedge_{\rw \in \iota^{-1}(\ell)} \neg\rw(\ell)\sigma
        \leq_n \Empty
      \] Moreover, by hypothesis that $L \subseteq \bigcup_{\rw \in P \cup N} \fin(\rw)$.
      Also, for each $\iota \notin I$,
      one of the two conditions~\eqref{eq:subtyping_nvars}
      or~\eqref{eq:subtyping_vars} of \cref{l:subtyping} is satisfied.
      By this corollary, $\bigwedge_{\rw \in P} \rw\sigma \wedge \bigwedge_{\rw
      \in N} \neg\rw\sigma \leq_n \Empty$.
      We conclude $\sigma \Vdash_n \{(\bigwedge_{\rw \in P} \rw \wedge \bigwedge_{\rw \in N} \neg\rw \leq \Empty)\}$.
    \item[\rulename{Ntail-tlv}]
      Let $\rw_0 = \bigwedge_{\rw \in P} \rw \wedge \bigwedge_{\rw \in N} \rw$.
      By hypothesis, there is $\rw' = \prow{}{\cutvar \rho {L_1}}{L_2} \in P \cup N$
      such that $\rho \lightning_O \rw'$.
      There are two similar cases, we detail the one where $\rw' \in P$.
      Let $\rw_0' = \bigwedge_{\rw \in P{\setminus}\{\rw'\}}
      \neg\rw \wedge \bigwedge_{\rw \in N} \rw$.
      By hypothesis, we have $\sigma \Vdash_n \{(\rw'\sigma \leq \neg \rw_0')\}$,
      thus $\rw'\sigma \leq_n \rw_0'\sigma$.
      This is equivalent to $\rw_0\sigma \leq_n \Empty$,
      and we conclude $\sigma \Vdash_n \{(\rw_0 \leq \Empty)\}$.
      \qedhere
  \end{description}
\end{proof}

We introduce a notion of plinth generalizing the one of \citet{frisch04} to
types, field-types and rows.
This notion is used to prove termination of the algorithm.
\begin{definition}[Plinth]
  \label{def:plinth}%
  A plinth $\beth \subset \Types_\Undef \cup \Rows$ is a set of types, field-types and rows with the following properties:
  \begin{itemize}
    \item $\beth$ is finite;
    \item $\beth$ contains $\Empty$, $\Any$, $\Undef$, $\orow{}\emptyset$ and is closed under
      Boolean connectives ($\vee$, $\wedge$, $\neg$);
    \item for all type $t_1 \to t_2 \in \beth$, we have $t_1 \in \beth$ and
      $t_2 \in \beth$;
    \item for all type $\R \in \beth$, we have $\rectorow{\R} \in \beth$;
    \item for all row $\rw \in \beth$ of definition space $\Labels{\setminus}L_\rw$,
      let $L$ be the set of labels appearing explicitely in $\beth$
      and $V$ the set of row variables in $\beth$, we have:
      \begin{itemize}
        \item for all $\ell \in L$, $\rw(\ell) \in \beth$
        \item for all $L' \subseteq L$, $\forall V' \subseteq V$,
          $\cutrowbis \rw {L'} {V'} \in \beth$
          and $\marshall(\rw,L',V') \in \beth$.
      \end{itemize}
  \end{itemize}
\end{definition}
Every finite set of types, field-types and types is included in a plinth.
Indeed, for a regular type $t$, the set of its subtrees $S$ is finite, while
rows and field-types are inductively defined.
The definition of the plinth ensures that the closure of $S$ under Boolean
connectives is also finite.
Moreover, if a type, field-type or row belongs in a plinth, the set of its
subtrees also does.
Finally, if a record or a row belongs to a plinth, the rows obtained by
marshalling or by removing fields also belongs to the plinth, but only if these
operations are done with respect to the set of labels present in the plinth, in
order to guarantee finiteness of the plinth.

\begin{lemma}[Termination]
  \label{l:termination-normalization}%
  Let $C$ be a finite constraint set.
  The normalization of $C$ terminates.
\end{lemma}
\begin{proof}
  Let $B$ be the set of types occuring in $C$.
  As $C$ is finite, $B$ is finite as well.
  Let $\beth$ be a plinth such that $B \subseteq \beth$.
  Then, when we normalize a constraint $(t, \leq, \Empty)$ during the process of
  $\normseq \emptyset C {}$, $t$ would belong to $\beth$.
  We prove the lemma by induction on $(|\beth\setminus\Sigma|,U,|C|)$,
  $U$ is the number of unions $\vee$ occurring in the constraint-set $C$ (over
  any kind) plus the number of constraints $(\tterm_1, \leq, \tterm_2)$ where
  $\tterm_2 \neq \Empty$ or $\tterm_1$ is not in DNF, and $C$ is the
  constraint-set to be normalized.
  We detail the original cases, others are described in the proof \cite[Lemma
  C.14]{castagna15}.
  \begin{description}
    \item[\rulename{Nopt}] Terminates immediately.
    \item[\rulename{Nrec}]
      None of the indices decreases, but the next rule to apply must be one of
      \rulename{Nrow}, \rulename{Ntail-mono} or \rulename{Ntail-tlv}.
    \item[\rulename{Nrow}]
      Although $(|\beth\setminus\Sigma|,U,|C|)$ may not change, the next rule to apply must be
      one of \rulename{Ndnf}, \rulename{Nhyp}, \rulename{Nassum},
      \rulename{Ntlv}, \rulename{Nopt} for $\constrset_\ell^\iota$
      and \rulename{Ntail-tlv} for $\constrset_{N'}^\iota$.
    \item[\rulename{Ntail-mono}] Although $(|\beth\setminus\Sigma|,U,|C|)$ may
      not change, the next rule to apply must be one of
      \rulename{Ndf}, \rulename{Nhy}, \rulename{Nassum}, \rulename{Ntlv},
      \rulename{Nopt}.
    \item[\rulename{Ntail-tlv}] Terminates immediately.
      \qedhere
  \end{description}
\end{proof}

\begin{definition}[Normalized constraint]
  A constraint is said to be \emph{normalized} if it is of the shape
  $(\vterm, c, \tterm)$, where $\vterm$ and $\tterm$ are of the same kind.
  A constraint-set is \emph{normalized} if all constraints are.
\end{definition}

\begin{lemma}
  Let $C$ be a constraint-set and $\normseq \emptyset C \constrset$.
  Then, all constraint-sets $C' \in \constrset$ are normalized.
\end{lemma}
\begin{proof}
  Straightforward by induction on the algorithm derivation.
\end{proof}

\begin{lemma}[Finiteness]
  \label{l:finiteness-normalization}%
  Let $C$ be a constraint-set and $\normseq \emptyset C \constrset$.
  Then, $\constrset$ is finite.
\end{lemma}
\begin{proof}
  It is easy to prove that each normalizing rule generates a finite set of
  finite sets of normalized constraints.
\end{proof}

\begin{definition}
  Let $C$ be a normalized constraint-set and $O$ an ordering on $\vars(C)$
  and on the labels occurring in $C$.
  We say $C$ is well-ordered if for all normalized constraint
  $(\vterm_1,c,\tterm) \in C$ and for all $\vterm_2 \in \tlv(\tterm)$,
  $O(\vterm_1) < O(\vterm_2)$ holds.
\end{definition}

\begin{lemma}
  \label{l:wellordered-normalization}%
  Let $C$ be a constraint-set and $\normseq \emptyset C \constrset$.
  Then for all normalized constraint-set $C' \in \constrset$,
  $C'$ is well-ordered.
\end{lemma}
\begin{proof}
  There are two different ways to generate normalized constraints:
  \begin{description}
    \item[\rulename{Ntlv}] We single out the type variable $\vterm'$ whose
      order is minimum, because $\tlv(\tterm) = \emptyset$.
    \item[\rulename{Ntail-tlv}] We single out the row variable $\rho$ whose
      order is minimum.
      \qedhere
  \end{description}
\end{proof}

\subsection{Constraint Merging and Saturation}
\label{sec:merging}

\begin{figure}
  \begin{framed}
  \begin{mathpar}
    \inferrule*[right=\rulename{MLB}]
    {\forall i \in I. (\vterm, \geq, \tterm_i) \in C \and |I| \geq 2}
    {\merseq C {(C \setminus \{(\vterm, \geq, \tterm_i) \mid i \in I\}\cup
    \{(\vterm, \geq, \bigvee_{i \in I} \tterm_i)\})}}
    \and \inferrule*[right=\rulename{MUB}]
    {\forall i \in I. (\vterm, \leq, \tterm_i) \in C \and |I| \geq 2}
    {\merseq C {(C \setminus \{(\vterm, \leq, \tterm_i) \mid i \in I\}\cup
    \{(\vterm, \leq, \bigvee_{i \in I} \tterm_i)\})}}
  \end{mathpar}
\end{framed}
  \caption{Merging rules}
  \label{fig:mrules}
\end{figure}

\begin{figure}
  \begin{framed}
  \begin{mathpar}
    \inferrule*[right=\rulename{Shyp}]
    {\satseq{\Sigma_p}{C_\Sigma \cup \{(\vterm,\geq,\tterm_1), (\vterm,\leq,\tterm_2)\}}C \constrset\and (\tterm_1,\tterm_2) \in \Sigma_p}
    {\satseq{\Sigma_p}{C_\Sigma}{\{(\vterm,\geq,\tterm_1),(\vterm,\leq,\tterm_2)\} \cup C}\constrset}
    \and \inferrule*[right=\rulename{Sassum}]
    {(\tterm_1,\tterm_2) \notin \Sigma_p
      \and \normseq \emptyset {\{(\tterm_1,\leq,\tterm_2)\}} \constrset\\
      \constrset' =\{\{(\vterm,\geq,\tterm_1),(\vterm,\leq,\tterm_2)\} \cup C
      \cup C_\Sigma\} \sqcap \constrset\\
    \forall C' \in \constrset'.\msseq{\Sigma_p \cup \{(\tterm_1,\tterm_2)\}} \emptyset{C'}{\constrset_{C'}}}
    {\satseq{\Sigma_p}{C_\Sigma}{\{(\vterm,\geq,\tterm_1),(\vterm,\leq,\vterm_2)\} \cup C}{\bigsqcup_{C' \in \constrset'} \constrset_{C'}}}
    \and \inferrule*[right=\rulename{Sdone}]
    {\forall \vterm, \tterm_1, \tterm_2.\nexists\{(\vterm,\geq,\tterm_1),(\vterm,\leq,\tterm_2)\} \subseteq C}
    {\satseq{\Sigma_p}{C_\Sigma} C {\{C \cup C_\Sigma\}}}
  \end{mathpar}
  {\small Where $\msseq{\Sigma_p}{C_\Sigma} C \constrset$ means that there
    exists $C'$ such that $\merseq C {C'}$ and
  $\satseq{\Sigma_p}{C_\Sigma}{C'}{\constrset}$.}
\end{framed}
  \caption{Saturation rules}
  \label{fig:srules}
\end{figure}

After normalization, we have a set of constraint-sets where the same variable
might have several upper and lower bounds given by different constraints on that
variable.
After this step, we wish to obtain unique upper and lower bounds for the same
variable.
This is done in two phases. For each normalized constraint-set $C$ and following
the order on variables, we:
\begin{itemize}
  \item Merge two constraints $(\vterm, \leq, \tterm_1)$ and
    $(\vterm, \leq, \tterm_2)$ into $(\vterm, \leq, \tterm_1 \wedge \tterm_2)$.
  \item Merge two constraints $(\vterm, \geq, \tterm_1)$ and
    $(\vterm, \geq, \tterm_2)$ into $(\vterm, \geq, \tterm_1 \vee \tterm_2)$.
\end{itemize}
Once the set $C$ has been transformed into $C'$ where there are no such
constraints left, we saturate the set to verify that the lower bound is indeed a
subtype of the upper bound, once again following the order on variables.
From two constraints $(\tterm_1, \leq, \vterm)$ and $(\vterm, \leq, \tterm_2)$,
we normalize the constraint $(\tterm_1 \wedge \neg\tterm_2, \leq, \Empty)$.
We obtain a set of constraint-sets $\constrset$.
We add the new resulting constraint-sets accordingly to the existing ones with
$\constrset \sqcap \{C'\}$, then apply the step of merging and saturation again
on each constraint-sets in $\constrset \sqcap \{C'\}$.
Termination of this step is assured as in the step of constraint
normalization by an additional argument $M$ to the merge function, saving
visited types.

Formally, this part is almost exactly the same as in the original algorithm in
\cite[Section C.1.2]{castagna15}.
One simply needs to replace occurrences of $\alpha$ and $t$ by the more general
meta-variables $\vterm$ and $\tterm$.
The rules are given in \cref{fig:mrules,fig:srules}.
Termination is proved thanks to the generalized definition of plinths
(\cref{def:plinth}).

\subsection{Harmonization}
\label{sec:harmonization}

Thanks to the last step, we now have a set of constraint-sets such that in all
constraint-sets, each variable in the domain of the constraint-set has at most
one upper bound and one lower bound.
Yet, for a row variable $\rho$ of the original type constraints, we can have
several constraints for the derived constructors $\cutvar \rho L$, for instance
a constraint-set containing both
$(\cutvar \rho {L_1}, \leq, \tterm_1)$ and
$(\cutvar \rho {L_2}, \leq, \tterm_2)$,
with $L_1 \cap \rdef(\rho) \neq L_2 \cap \rdef(\rho)$.
There could also be occurrences of another construction $\cutvar \rho {L_3}$ in
$\tterm_1$ with a different $L_3$.
We want a unique decomposition of $\rho$ in each constraint-set, with
a unique set $L_0$ such that only $\cutvar \rho L_0$
may appear in the domain of the constraint-set.
From this, we will build a solution such that
$\sigma(\rho) \simeq {(\ell = \tau_\ell)_{\ell \in L_0}}{\rw}$.
There can be occurrences of $\cutvar \rho {L'}$ not in the domain, or of
$\rho.\ell'$, but substitution will be correctly defined since we take $L_0$ to
cover all such occurrences.

Formally, a harmonized constraint is defined as follow.
\begin{definition}[Harmonized constraint]
  Let $C \subseteq \Constraints$ be a saturated constraint-set.
  We say $C$ is \emph{harmonized} if for each type variable $\rho \in \dom(C)$,
  there is a finite set of labels $L$ such that:
  \begin{enumerate}
    \item $\forall (\cutvar \rho {L_i}, c, \rw_i) \in C.\ L_i{\setminus}\rdef(\rho)
      = L\setminus\rdef(\rho)$;
    \item $\forall \cutvar \rho {L'} \in \vars(C).\ L' \subseteq L$;
    \item $\forall \rho.\ell \in \vars(C).\ \ell \in L\cap\rdef(\rho)$.
  \end{enumerate}
\end{definition}
If a constraint-set is not harmonized, then it is of the shape
$C \cup \{\row{}{\cutvar \rho L}{L}, c, \rw\}$, where $L_0 \nsubseteq L$,
when we define
$L_0 = \bigcup_{\rho.\ell \in \vars(C) \cup \vars(\rw)} \{\ell\}
\cup \bigcup_{\cutvar \rho {L'} \in \vars(C) \cup \vars(\rw)} L'$.
To harmonize the decomposition of the variable row,
we normalize the constraint-set
$\{(\row{(\ell = \rho.\ell)_{\ell \in L_0{\setminus}L}}
{\cutvar \rho {L_0}}{L}, c, \rw)\}$.
We integrate the obtained set of constraint-sets $\constrset$ to the existing
constraints with $\{C\} \sqcap \constrset$ and apply merging and
harmonization recursively.
The rules are given in \cref{fig:hrules}.

\begin{figure}
  \begin{framed}
  \begin{mathpar}
    \inferrule*[right=\rulename{Hdone}]
    {\text{$C$ is harmonized}}
    {\harmseq \Sigma C {\{C\}}}
    \and \inferrule*[right=\rulename{Harm}]
    {
      L_0 = \bigcup_{\rho.\ell \in \vars(C) \cup \vars(\rw)} \{\ell\}
      \cup \bigcup_{\cutvar \rho {L'} \in \vars(C) \cup \vars(\rw)} L'
      \and L_0 \nsubseteq L\\
      \and \normseq \emptyset
      {\{(\row{(\ell = \rho.\ell)_{\ell \in L_0{\setminus}L}}{\cutvar \rho {L_0}}{L}, c, \rw)\}}
      \constrset\\
      \forall C_n \in \{C\} \sqcap \constrset.
      (\msseq \emptyset \emptyset {C_n}{\constrset_n})\\
      \and \forall C_m \in \constrset_n \textand C_m \notin \Sigma.
      (\harmseq {\Sigma \cup C_m} {C_m} {\constrset_n^m})
    }
    {\harmseq \Sigma {C \cup \{(\row{}{\cutvar \rho L}{L},c,\rw)\}}
      {\bigsqcup_{C_n \in \{C\} \sqcap \constrset}
    \bigsqcup_{C_m \in \constrset_n} \constrset_n^m}}
  \end{mathpar}
\end{framed}
  \caption{Harmonization rules}
  \label{fig:hrules}
\end{figure}

\begin{lemma}[Soundness]
  \label{l:soundness-harmonization}%
  Let $C$ be a finite saturated constraint-set.
  If $\harmseq \emptyset C \constrset$, then for all constraint-set $C' \in
  \constrset$ and all substitution $\sigma$, we have $\sigma \Vdash C' \implies
  \sigma \Vdash C$.
\end{lemma}
\begin{proof}
  The proof is by induction on the derivation tree.
  We prove the more general statement for all $\Sigma$.
  The base case $\harmseq \Sigma C {\{C\}}$ is trivial.
  In the inductive case, let
  $C = C_0 \cup \{(\row{}{\cutvar \rho L}{L},c,\rw)\}$.
  By definition of $\sqcup$, there are $C_n \in \{C\} \sqcap \constrset$ and
  $C_m \in \constrset_n$ such that $C' \in \constrset_n^m$.
  Since $\sigma \Vdash C'$, we have by induction hypothesis on
  $\harmseq {\Sigma \cup C_m} {C_m} {\constrset_n^m}$ that $\sigma \Vdash C_m$.
  By definition of $\sqcup$ again, there is $C_n \in \{C_0\} \sqcap \constrset$
  such that $C_m \in \constrset_n$.
  Since $\sigma \Vdash C_m$, we have by soundness of constraint merging on
  $\msseq \emptyset \emptyset {C_n}{\constrset_n}$ that $\sigma \Vdash C_n$.
  We have $C_n = C_0 \cup C_n'$ with $C_n' \in \constrset$.
  Since $\sigma \Vdash C_n$, we have $\sigma \Vdash C_n'$ and by
  \cref{l:soundness-normalization} $\sigma \Vdash
  \{(\row{(\ell = \rho.\ell)_{\ell \in L_0{\setminus}L}}{\cutvar \rho
  {L_0}}{L}, c, \rw)\}$.
  Thus, $\sigma(\rho) \simeq
  \row{(\ell = \tau_\ell)_{\ell \in L \cap \rdef(\rho)},
    (\ell = \tau_\ell)_{\ell \in
  L_0{\setminus}L}}{\rw'}{\Labels{\setminus}\rdef(\rho)}$.
  Then, $(\row{(\ell = \rho.\ell)_{\ell \in
  L_0{\setminus}L}}{\cutvar \rho {L_0}}{L})\sigma
  = (\row{}{\cutvar \rho L} L)\sigma
  = \row{(\ell = \tau_\ell)_{\ell \in L_0{\setminus}L}}{\rw'}{L}$.
  Therefore, $\sigma \Vdash \{(\row{}{\cutvar \rho L}{L},c,\rw)\}$.
  We conclude $\sigma \Vdash C_0$ as well because $\sigma \Vdash C_n$
  and $C_0 \subseteq C_n$.
\end{proof}

\begin{lemma}[Termination]
  \label{l:termination-harmonization}%
  Let $C$ be a finite saturated constraint-set.
  The harmonization of $C$ terminates.
\end{lemma}
\begin{proof}
  Let $B$ be the set of types, type fields and rows in $C$, finite since $C$ is
  finite.
  Let $\beth$ be a plinth such $B \subseteq \beth$.
  When adding a set of constraint $C_m$ to $\Sigma$ during harmonization, every
  constraint of $C_m$ belongs to $(\beth \times \{\leq,\geq\} \times \beth)$.
  The proof is by induction on $(|(\beth \times \{\leq,\geq\}, \times \beth)| -
  |\Sigma|, |C|)$ lexicographically ordered.
  \begin{description}
    \item[\rulename{Hdone}] Terminates immediatly.
    \item[\rulename{Harm}] Normalization, merging and saturation all terminate.
      In the recursive step of harmonization that are applied, $|(\beth \times
      \{\leq,\geq\} \times \beth)| - |\Sigma|$ decreases.
      \qedhere
  \end{description}
\end{proof}

\begin{lemma}[Finiteness]
  \label{l:finiteness-harmonization}%
  Let $C$ be a finite saturated constraint-set and $\harmseq \emptyset C \constrset$.
  Then $\constrset$ is finite.
\end{lemma}
\begin{proof}
  By induction on the derivation and by \cref{l:finiteness-normalization} and
  finiteness of constraint merging.
\end{proof}

\begin{lemma}
  \label{l:harmonized}%
  Let $C$ be a finite saturated constraint-set and $\harmseq {} C \constrset$.
  Then for all constraint-set $C' \in \constrset$, $C'$ is harmonized.
\end{lemma}
\begin{proof}
  Direct by induction on the derivation, finite by
  \cref{l:termination-harmonization}.
\end{proof}

\begin{lemma}
  \label{l:wellordered-harmonization}%
  Let $C$ be a well-ordered saturated constraint-set and $\harmseq \emptyset C
  \constrset$.
  Then for all harmonized constraint-set $C' \in \constrset$, $C'$ is
  well-ordered.
\end{lemma}
\begin{proof}
  Consequence of \cref{l:wellordered-normalization} and conservation of
  well-orderedness by merging and saturation.
\end{proof}

\subsection{From Constraints to Equations}
\label{sec:solving}

Once normalization, merging and harmonization are done,
we have a set of well-ordered constraint-sets at hand, where all variables have unique
lower and upper bounds, and for each row variable $\rho$, there is at most a
unique occurrence of $\cutvar \rho L$ (where $L$ can be empty if $\rho$ has not
been decomposed), such that also every occurrence of $\rho.\ell$ has $\ell
\in L$.
We are now able to rewrite each constraint-set $C$ into an equivalent equation
system.

\begin{definition}[Equation system]
  An equation system $E$ is a set of equations of the form $\vterm = \tterm$
  such that there exists at most one equation in $E$ for every variable
  $\vterm$, and $\vterm$ and $\tterm$ are of the same kind.
  We define the domain of $E$, written $\dom(E)$, as the set
  $\{\vterm \mid \exists \tterm. \vterm = \tterm \in E\}$.
\end{definition}

\begin{definition}[Equation system solution]
  Let $E$ be an equation system.
  A solution to $E$ is a substitution $\sigma$ such that
  $\forall (\vterm = \tterm) \in E. \sigma(\vterm) \simeq \tterm\sigma$ holds.
  If $\sigma$ is a solution to $E$, we write $\sigma \Vdash E$.
\end{definition}

Given a constraint-set $C$, we will use the notation $(\tterm_1 \leq \vterm \leq
\tterm_2) \in C$ to indicate $\{ (\tterm_1, \leq, \vterm), (\vterm, \leq,
\tterm_2) \} \subseteq C$.
We assume that every variable and every term $\rho.\ell$, $\cutvar \rho L$ in
$\dom(C)$ have an upper and a lower bound,
without loss of generality because a constraint with
bottom or top types can always be added if needed:
\begin{itemize}
  \item For a type variable $\alpha$, we can add constraints
    $(\Empty, \leq, \alpha)$ or $(\alpha, \leq, \Any)$;
  \item For a field variable $\fvar$ (or $\rho.\ell$), we can add constraints
    $(\Empty, \leq, \fvar)$ or $(\fvar, \leq, \Any \lor \Undef)$;
  \item For a row variable $\cutvar \rho L$ (or $\rho$), we can add constraints
    $(\Empty, \leq, \cutvar \rho L)$ or
    $(\cutvar \rho L, \leq, \orow{}{\Labels{\setminus}(\rdef(\rho)\cup L)})$.
\end{itemize}

We rewrite the set $C$ to a set of equations with a function $\solve(C)$,
where $\alpha'$, $\fvar'$, $\fvar_\ell$ and $\rho'$ are fresh variables.
We write $(\tterm_1 \leq \vterm \leq \tterm_2)$ for the constraints
$\{(\tterm_1, \leq, \vterm), (\vterm, \leq, \tterm_2)\} \subseteq C$.
\begin{align*}
  \solve(C)
  &= \{ \alpha = (t_1 \vee \alpha') \wedge t_2
  \mid (t_1 \leq \alpha \leq t_2) \in C \}\\
  &\cup \{ \fvar = (\tau_1 \vee \fvar') \wedge \tau_2
    \mid (\tau_1 \leq \fvar \leq \tau_2) \in C
  \textand \fvar \neq \rho.\ell \}\\
  &\cup \{ \rho = \row{
    (\ell = \tau_\ell)_{\ell \in L}}{\rw}{\Labels{\setminus}\rdef(\rho)}\\
    &\quad\mid (\text{if } \exists L'. (\rw_1 \leq \cutvar \rho {L'} \leq \rw_2) \in C,
    \text{ then } L = L' \cap \rdef(\rho)
    \textand \rw = (\rw_1 \vee \rho') \wedge \rw_2,\\
    &\quad\text{ else } L = \{\ell \mid \rho.\ell \in \vars(C)\}
    \textand \rw = \rho')\\
    &\quad\textand (\forall \ell \in L.
    \text{ if } (\tau_\ell^1 \leq \rho.\ell \leq \tau_\ell^2) \in C
    \text{ then } \tau_\ell = (\tau_\ell^1 \vee \fvar_\ell) \wedge \tau_\ell^2
    \text{ else } \rw = \fvar_\ell)
  \}  
\end{align*}
where $\alpha'$, $\fvar'$, and $\rho'$ are fresh variables.

For type and field variables (not generated by the decomposition of a row
variable), we obtain an equation by means of the type connectives, where the
union entails a lower bound and the intersection an upper bound.
For a row variable $\rho$, there is $L$ and constraints
$(\tau_\ell^1 \leq \rho.\ell \leq \tau_\ell^2) \in C$ for all $\ell \in L$ and
$(\rw_1 \leq \cutvar \rho L \leq \rw_2)$, with the potentially missing
constraints obtained with the default values.
Since we have decomposed $\rho$ into the labels in $L$ and a part of definition
space $\Labels \setminus (\rdef(\rho) \cup L)$, we build an equation for $\rho$ by
concatenating the independent types for $\rho.\ell$ and $\cutvar \rho \ell$
together.

To prove soundness of the transformation, we define the rank $n$ satisfaction
predicate $\Vdash_n$ for equation systems, which is similar to the one for
constraint-sets.
\begin{lemma}[Soundness]
  \label{l:soundness-solve}%
  Let $C \subseteq \Constraints$ be a well-ordered saturated constraint-set and
  $E$ its transformed equation system.
  Then for all substitutions $\sigma$, if $\sigma \Vdash E$,
  then $\sigma \Vdash C$.
\end{lemma}
\begin{proof}
  We write $O(C_1) < O(C_2)$ if $O(\vterm_1) < O(\vterm_2)$ for all $\vterm_1
  \in \dom(C_1)$ and all $\vterm_2 \in \dom(C_2)$.
  We prove a stronger statement:
  \begin{description}
    \item[(*)] For all $\sigma$, $n$ and $C_\Sigma \subseteq C$,
      if $\sigma \Vdash_n E$, $\sigma \Vdash_n C_\Sigma$, $\sigma \Vdash_{n-1} C \setminus C_\Sigma$
      and $O(C \setminus C_\Sigma) < O(C_\Sigma)$,
      then $\sigma \Vdash_n C \setminus C_\Sigma$.
  \end{description}
  Here $C_\Sigma$ denotes the set of constraints that have been checked.
  The proof proceeds by induction on $|C \setminus C_\Sigma|$, and is similar to
  the proof of \cite[Lemma C.33]{castagna15} for type variables only.
  The base case $C \setminus C_\Sigma = \emptyset$ is straightforward.
  Let $C \setminus C_\Sigma \neq \emptyset$ and let us consider the case of row
  variables.
  Take $\rho$ with the maximal order in $\dom(C \setminus C_\Sigma)$.
  There are a set $L$ and corresponding equations
  $\cutvar \rho L = \rw = (\rw_1 \wedge \rho') \wedge \rw_2$ and
  $(\rho.\ell = \tau_\ell = (\tau_\ell^1 \vee \fvar_\ell) \wedge \tau_\ell^2)_{\ell \in L}$.
  As $\sigma \Vdash_n E$, we have
  $\sigma(\rho) \simeq_n (\row{\ell = \tau_\ell}{\rw}{L_\rho}\sigma)$,
  where $L_\rho = \Labels \setminus \rdef(\rho)$.
  Then, for all $\ell \in L$: \[
    (\rho.\ell)\sigma \wedge \neg\tau_\ell^2\sigma
    \simeq_n ((\tau_\ell^1 \vee \fvar_\ell) \wedge \tau_\ell^2)\sigma \wedge \neg \tau_\ell^2\sigma
    \simeq_n \Empty
  \] And similarly for $(\cutvar \rho L)\sigma \wedge \neg(\cutvar \rho L) \sigma$.
  On the other hand, for all $\ell \in L$, we have: \[
    \tau_\ell^1\sigma \wedge \neg(\rho.\ell)\sigma
    \simeq_n \tau_\ell^1 \wedge \neg((\tau_\ell^1 \vee \fvar_\ell) \wedge
    \tau_\ell^2)\sigma
    \simeq_n \tau_\ell^1\sigma \wedge \neg\tau_\ell^2\sigma
  \] And similarly for $\rw_1\sigma \wedge \neg(\cutvar \rho L)\sigma$.
  It remains to show that $\tau_\ell^1\sigma \leq_n \tau_\ell^2\sigma$ holds for
  all $\ell \in L$ and that $\rw_1\sigma \leq_n \rw_2\sigma$ hold, that is
  $\sigma \Vdash_n \{(\tau_\ell^1 \leq \tau_\ell^2)_{\ell \in L}, (\rw_1 \leq
  \rw_2\}$.
  The rest of the proof goes as in \cite[Lemma C.33]{castagna15}.
  We use the fact that the order of $\rho.\ell$ and $\cutvar \rho L$ is directly
  superior to the order of $\rho$, and thus maximal in $\dom(C \setminus
  C_\Sigma)$.
\end{proof}

\begin{lemma}[Completeness]
  \label{l:completeness-solve}%
  Let $C \subseteq \Constraints$ be a saturated normalized constraint-set and
  $E$ its transformed equation system.
  Then for all substitution $\sigma$, if $\sigma \Vdash C$,
  there exists $\sigma'$ such that $\dom(\sigma') \cup \dom(\sigma) = \emptyset$
  and $\sigma \cup \sigma' \Vdash E$.
\end{lemma}
\begin{proof}
  Let $\sigma'
  = \{ \sigma(\alpha)/\alpha' \mid \alpha \in \dom(C) \}
  \cup \{ \sigma(\fvar)/\fvar' \mid \fvar \in \dom(C) \textand \fvar \neq \rho.\ell \}
  \cup \{ (\cutvar \rho L)\sigma/\rho' \}
  \cup \{ (\rho.\ell)\sigma/\fvar_\ell^\rho \mid \rho.\ell \in \dom(C) \}$,
  where $L$ is obtained as in the definition of $\solve$ either from $\cutvar \rho {L'}
  \in \dom(C)$ and $L = L' \cap \rdef(\rho)$, or by $L = \{ \ell \mid \rho.\ell
  \in \vars(C) \}$.
  The case for type and field variables is as in \cite[Lemma C.34]{castagna15}.
  Let us consider an equation
  $\rho = \row{(\ell = \tau_\ell)_{\ell \in L}}{\rw}
  {\Labels{\setminus}\rdef(\rho)} \in E$.
  Correspondingly, there exist $(\rw_1 \leq \cutvar \rho L \leq \rw_2) \in C$
  and for all $\ell \in L$, there exist $(\tau_\ell^1 \leq \rho.\ell \leq
  \tau_\ell^2) \in C$ (without loss of generality, we suppose $C$ to be
  saturated with default values).
  As $\sigma \Vdash C$, then $\rw_1\sigma \leq (\cutvar \rho L)\sigma \leq
  \rw_2\sigma$, and for all $\ell \in L$, $\tau_\ell^1\sigma \leq
  (\rho.\ell)\sigma \leq \rw_2\sigma$, and the operations $(\cutvar \rho
  L)\sigma$ and $(\rho.\ell)\sigma$ are defined.
  Thus,
  \begin{align*}
    \row{(\ell = \tau_\ell)_{\ell \in L}}{\rw}{}(\sigma \cup \sigma')
    &= \row{(\ell = (\tau_\ell^1(\sigma \cup \sigma') \vee (\sigma \cup \sigma')(\fvar_\ell^\rho))
      \wedge \tau_\ell^2(\sigma \cup \sigma'))_{\ell \in L}}
    {\\&\qquad(\rw_1(\sigma \cup \sigma') \vee (\sigma \cup \sigma')(\rho')) \wedge \rw_2(\sigma \cup \sigma')}{}\\
    &= \row{(\ell = (\tau_\ell^1\sigma \vee (\rho.\ell)\sigma) \wedge \tau_\ell^2\sigma)_{\ell \in L}}
      {(\rw_1\sigma \vee (\cutvar \rho L)\sigma \wedge \rw_2\sigma}{}\\
    &\simeq \row{(\ell = (\rho.\ell)\sigma \wedge \tau_\ell^2\sigma)_{\ell \in L}}
      {(\cutvar \rho L)\sigma \wedge \rw_2\sigma}{}\\
    &\simeq \row{(\ell = (\rho.\ell)\sigma)_{\ell \in L}}
    {(\cutvar \rho L)\sigma}{}\\
    &= \sigma(\rho)
  \end{align*}
  The last line is justified by $\sigma$ being a solution to $C$, so being of
  the shape $\sigma(\rho) = \row{(\ell = \tau_\ell')_{\ell \in L}}{\rw'}{}$.
\end{proof}

\begin{definition}
  \label{def:wellordered-equations}%
  Let $E$ be an equation system and $O$ an ordering on $\dom(E)$ and on the
  labels occurring in $E$.
  We say that $E$ is well-ordered if for all $\vterm = \tterm_\vterm \in E$
  and $\vterm' \in \tlv(\tterm_\vterm) \cap \dom(E)$,
  we have $O(\vterm_1) < O(\vterm')$.
\end{definition}

\begin{lemma}
  \label{l:wellordered-solving}%
  Let $C$ be a well-ordered saturated normalized constraint-set and $E$ its
  transformed equation system.
  Then $E$ is well-ordered.
\end{lemma}
\begin{proof}
  We have $\dom(E) = \dom(C) \cap (\typeVars \cup \fldVars) \cup
  \{ \rho \mid \exists \ell. \rho.\ell \in \dom(C)
  \textor \exists L. \cutvar \rho L \in \dom(C) \}$.
  The case for type and field variables is as in \cite[Lemma C.36]{castagna15},
  but uses the fact that $(\tlv(\tterm_1) \cup \tlv(\tterm_2)) \cap \dom(E) =
  (\tlv(\tterm_1) \cup \tlv(\tterm_2)) \cap \dom(C)$.
  Now, consider $\rho = \rw_0$ with $\rw_0 = \row{(\ell = \tau_\ell)_{\ell \in L}}{\rw}{}$ obtained
  from $(\rw_1 \leq \cutvar \rho L \leq \rw_2) \in C$ with
  $\rw = (\rw_1 \wedge \rho') \vee \rw_2$,
  and for all $\ell \in L$ from $(\tau_\ell^1 \leq \rho.\ell^2) \in C$.
  We have $\tlv(\rw_0) = \tlv(\rw)$.
  Since $C$ is well-ordered, for all $(\rho_2 \in \tlv(\rw_1) \cup \tlv(\rw_2))
  \cap \dom(C)$, $O(\rho) < O(\rho_2)$.
  Moreover, $\rho'$ is a fresh row variable in $C$, that is $\rho' \notin \dom(C)$.
  And then $\rho' \notin \dom(E)$.
  Therefore, $\tlv(\rw) \cap \dom(E) = (\tlv(\rw_1) \cup \tlv(\rw_2)) \cap
  \dom(C)$ and the result follows.
\end{proof}

\subsection{Solution of Equation Systems}
\label{sec:solution}

We have now obtained a set of equation-sets $E$, that we must each transform
into a substitution $\sigma$.
We do this in the same way as in \cite[\S 3.2.2]{castagna15}, but with all kinds
of variables rather than just type variables.
In the set of equations, there is no construction $\rho.\ell$ or $\cutvar \rho L$ anymore.
We define a function $\Unify(E)$ as $\Unify(\emptyset) = \{\}$, and otherwise:
\begin{enumerate}
  \item Select in $E$ the equation $\vterm = \tterm$ for the smallest $\vterm$
    w.r.t. the order;
  \item Let $E'$ be the set of equations obtained by replacing in
    $E{\setminus}\{\vterm = \tterm\}$ every occurrence of $\vterm$ by
    $\mu\vterm'.(\tterm\esubs{\vterm}{\vterm'})$ ($\vterm'$ fresh);
  \item Let $\sigma = \Unify(E')$ and return
    $\{\vterm = (\mu\vterm'.\tterm\esubs{\vterm}{\vterm'})\sigma\} \cup \sigma$.
\end{enumerate}
The ordering on the variables guarantees the regularity of the obtained types.
For the elements $\sigma(\vterm) = \mu\vterm'.\tterm$ where $\vterm' \notin
\vars(\tterm)$, we can remove the introduced $\mu$-abstraction.

It is straightforward to extend the proofs and definitions.
\begin{definition}[General solution]
  Let $E$ be an equation system.
  A general solution to $E$ is a substitution $\sigma$ from $\dom(E)$ to
  $\Types_\Undef \cup \Rows$ such that
  $\forall \vterm \in \dom(\sigma). \vars(\sigma(\vterm)) \cap \dom(\sigma) = \emptyset$ and
  $\forall \vterm = \tterm \in E. \sigma(\vterm) \simeq \tterm\sigma$ holds.
\end{definition}

\begin{definition}[Equivalent substitutions]
  Let $\sigma, \sigma'$ be two substitutions.
  We say $\sigma \simeq \sigma'$ if and only if $\forall \vterm. \sigma(\vterm)
  \simeq \sigma'(\vterm)$.
\end{definition}

\begin{proposition}
  \label{l:unify}%
  Let $E$ be a well-ordered equation system.
  Let $\Unify(E)$ be the procedure described by \citet{castagna15} to build a
  substitution.
  \begin{description}
    \item[Soundness] If $\sigma = \Unify(E)$, then $\sigma \Vdash E$.
    \item[Completeness] For all substitution $\sigma$, if $\sigma \Vdash E$,
      then there exist $\sigma_0$ and $\sigma'$ such that $\sigma_0 = \Unify(E)$
      and $\sigma \simeq \sigma' \circ \sigma_0$.
    \item[Termination] The algorithm $\Unify(E)$ terminates.
  \end{description}
\end{proposition}

The last property we verified is well-formedness.
As in \cite{castagna15}, a type is well-formed if and only if the recursion
traverses a constructor, and this property is guaranteed thanks to the order on
variables.
\begin{proposition}[Well-formedness]\label{prop:wellformedness}
  If $\sigma = \Unify(E)$, then for all $\vterm \in \dom(\sigma)$,
  $\sigma(\vterm)$ is well-formed.
\end{proposition}
\begin{proof}
  Assume that there exists an ill-formed $\sigma(\vterm)$.
  That is, $\sigma(\vterm) = \mu x.t$ where $x$ occurs at the top-level of $t$.
  According to the algorithm $\Unify()$, there exists a sequence of equations
  $(\vterm =) \vterm_0 = \tterm_{\vterm_0}, \dots, \vterm_n = \tterm_{\vterm_n}$
  such that $\vterm_i$ is at top-level in $\tterm_{\vterm_{i-1}}$ and $\vterm_0$
  is at top-level in $\tterm_{\vterm_n}$ and where $i \in \{1, \dots, n\}$ and
  $n \geq 0$.
  We must necessarily have all the $\vterm_i$ of the same kind.
  Indeed, for type (resp. row) variables only types (resp. row) variables can
  appear at top-level.
  Now, if $\vterm_0$ is a field variable, there can be a type variable
  $\vterm_i$ at top-level in the field-type $\tterm_{i-1}$.
  But then, $\vterm_0$ cannot be at top-level in $\vterm_n$ since $\tterm_i$ is
  a type, and field variables cannot appear at top-level in a type.
  Since all $\vterm$ and $\tterm$ must be of the same kind and according to
  \cref{def:wellordered-equations}, we have $O(\vterm_{i-1}) < O(\vterm_i)$
  and $O(\vterm_n) < O(\vterm_0)$.
  Therefore, we have $O(\vterm_0) < O(\vterm_1) < \dots < O(\vterm_n) <
  O(\vterm_0)$, which is impossible.
  Thus the result follows.
\end{proof}

\subsection{The Complete Algorithm}

The procedure $\Sol C {}$ to solve type tallying of a constraint-set $C$
proceeds as follows.
\begin{enumerate}
  \item $C$ is normalized into a finite set $\constrset$ of well-ordered
    normalized constraint-sets (\cref{sec:normalization}).
  \item Each constraint-set $C_i \in \constrset$ is merged and saturated
    into a finite set $\constrset_{C_i}$ of well-ordered
    constraint-sets.
    Then, all these sets are collected into another set $\constrset'$
    (\textit{i.e.}, $\constrset' = \bigsqcup_{C_i \in \constrset}
    \constrset_{C_i}$) (\cref{sec:merging}).
  \item Each constraint-set $C_i' \in \constrset'$ is
    harmonized into a finite set $\constrset_{C_i'}$ of well-ordered
    harmonized
    constraint-sets.
    Then, all these sets are collected into another set $\constrset''$
    (\textit{i.e.}, $\constrset'' = \bigsqcup_{C_i' \in \constrset'}
    \constrset_{C_i'}$) (\cref{sec:harmonization}).
    This step is specific to row variables.
  \item For each constraint-set $C_i'' \in \constrset''$, we
    transform $C_i''$ into an equation system $E_i$ and then construct a general
    solution $\sigma_i$ from $E_i$ (\cref{sec:solving}).
  \item Finally, we collect all the solutions $\sigma_i$, yielding a set
    $\Theta$ of solutions to $C$ (\cref{sec:solution}).
\end{enumerate}
In the original algorithm for type variables, failing at the step of normalization means
that there is no solution overall, even when increasing the cardinality of the
substitution sets sought by dove-tail order in the general algorithm
described in the beginning of \cref{sec:tallying} (see \cite[\S 3.2.3]{castagna15}).
Whether here a failure in the step of normalization or harmonization means the
absence of a solution overall is still an open question.

We write $\Sol C \Theta$ if $\Sol C {}$ terminates with $\Theta$, and we
call $\Theta$ the solution of the type tallying problem for $C$.

\soundnesstallying*
\begin{proof}
  Consequence of \cref{l:soundness-normalization},
  soundness of merging, \cref{l:soundness-harmonization},
  \cref{l:soundness-solve} and soundness of $\Unify$ (\cref{l:unify}).
\end{proof}

\terminationtallying*
\begin{proof}
  Consequence of \cref{l:termination-normalization}, termination of
  merging, \cref{l:termination-harmonization}, finiteness of the
  constraint-sets and termination of $\Unify$ (\cref{l:unify}).
\end{proof}

\finitetallying*
\begin{proof}
  The first item is a consequence of \cref{l:finiteness-normalization},
  \cref{l:finiteness-harmonization} and finiteness of the constraints after
  merging.
  The second one is a consequence of \cref{l:wellordered-normalization},
  well-orderedness of merging, \cref{l:wellordered-harmonization},
  \cref{l:wellordered-solving} and well-orderedness of $\Unify$
  (\cref{l:unify}).
\end{proof}

%

\end{document}